\documentclass[letter,10pt]{article}
\usepackage{amsthm,amsfonts,amssymb,euscript}
\usepackage{graphicx}
\usepackage{color}
\usepackage{amsmath}
\usepackage{slashed}
\usepackage{mathtools}

 \usepackage{xcolor}
 \usepackage[citebordercolor={green}]{hyperref}

\usepackage{tikz}
\usetikzlibrary{arrows.meta}
\usetikzlibrary{calc}
\usetikzlibrary{intersections}
\usetikzlibrary{decorations.markings}

\usepackage{caption}

\numberwithin{equation}{section}

\def\l{\left(}
\def\r{\right)}

\newtheorem{theorem}{Theorem}[section]
\newtheorem{lemma}[theorem]{Lemma}
\newtheorem{proposition}[theorem]{Proposition}
\newtheorem{corollary}[theorem]{Corollary}
\newtheorem{definition}[theorem]{Definition}
\newtheorem{remark}[theorem]{Remark}

\newcommand{\bea}{\begin{eqnarray}}
\newcommand{\eea}{\end{eqnarray}}
\def\beaa{\begin{eqnarray*}}
\def\eeaa{\end{eqnarray*}}
\def\ba{\begin{array}}
\def\ea{\end{array}}
\def\be#1{\begin{equation} \label{#1}}
\def \eeq{\end{equation}}
\newcommand{\lab}{\lababel}

\def\a{{\alpha}}

\def\b{{\beta}}
\def\be{{\beta}}
\def\ga{\gamma}
\def\Ga{\Gamma}
\def\de{\delta}
\def\De{\Delta}
\def\ep{\epsilon}
\def\ka{\kappa}

\def\la{\lambda}

\def\Si{\Sigma}
\def\om{\omega}
\def\Om{\Omega}

\def\ep{\epsilon}
\def\vphi{\varphi}

\def\th{\theta}

\def\ze{\zeta}
\def\ka{\kappa}
\def\nab{\nabla}

\def\pr{{\partial}}
\def\les{\lesssim}

\def\MM{{\mathcal M}}

\def\JJ{\mathcal{J}}
\def\II{{\mathcal I}}

\def\HH{{\mathcal H}}
\def\g{{\mathcal G}}

\def\OO{{\mathcal O}}

\def\UU{{\mathcal U}}

\def\KK{{\mathcal K}}

\def\RR{{\mathcal R}}

\def\HH{{\mathcal H}}

\def\D{{\bf D}}

\def\Jk{{\bf J}}

\def\g{{\bf g}}

\def\SSS{{\mathbb S}}

\def\f12{{\frac 1 2}}

\def\dual{{\,\,^*}}
\def\div{{\mbox {div}\,}}
\def\curl{{\mbox {curl}\,}}

\def\Lb{{\,\underline{l}}}

\def\Lb{{\,\underline{L}}}

\def\Xh{\,^{(h)}X}

\def\trch{{\mathrm{tr}}\, \chi}
\def\chih{{\widehat \chi}}
\def\chib{{\underline \chi}}
\def\chibh{{\underline{\chih}}}

\def\etab{{\underline \eta}}
\def\omb{{\underline{\om}}}
\def\bb{{\underline{\b}}}
\def\aa{\protect\underline{\a}}
\def\xib{{\underline \xi}}

\def\nabb{{\nab}}

\def\tr{{\mathrm{tr}}}
\def\atr{\,^{(a)}\mathrm{tr}}

\def\trchb{{\tr \,\chib}}

\def\atrch{\atr\, \chi}
\def\atrchb{\atr\, \chib}

\def\hot{\widehat{\otimes}}

\def\fb{\protect\underline{f}}

\def\f12{\frac 1 2}
\def\lab{\label}

\def\bsplit{\begin{split}}

\newcommand{\Mext}{{\,{}^{(ext)}\mathcal{M}}}

\DeclareFontFamily{U}{mathx}{\hyphenchar\font45}
\DeclareFontShape{U}{mathx}{m}{n}{
      <5> <6> <7> <8> <9> <10>
      <10.95> <12> <14.4> <17.28> <20.74> <24.88>
      mathx10
      }{}
\DeclareSymbolFont{mathx}{U}{mathx}{m}{n}
\DeclareFontSubstitution{U}{mathx}{m}{n}
\DeclareMathAccent{\widecheck}{0}{mathx}{"71}

\def\Gac{\widecheck{\Ga}}

\def\nabzero{{\overset{\circ}{ \nab}}}

\def\Db{\dot{\D}}

\def\eS{\,^{S} \hspace{-1.5pt} e}

\DeclareFontFamily{U}{mathx}{\hyphenchar\font45}
\DeclareFontShape{U}{mathx}{m}{n}{
      <5> <6> <7> <8> <9> <10>
      <10.95> <12> <14.4> <17.28> <20.74> <24.88>
      mathx10
      }{}
\DeclareSymbolFont{mathx}{U}{mathx}{m}{n}
\DeclareFontSubstitution{U}{mathx}{m}{n}
\DeclareMathAccent{\widecheck}{0}{mathx}{"71}

\def\psit{\tilde{\psi}}

\def\pat{\widetilde{\partial}}

     \def\Jk{\bf j}

\def\pa{\partial}
 \def\grad{{\mathrm{grad}}}

\def\dual{\prescript{*}{}}
\def\anti{\prescript{(a)}{}}
\def\Fb{\underline{F}}
\def\eS{\,^{(S)}e}

\def\trchS{\, ^{(S)}\trch}

\def\Mint{{\protect \, ^{(int)}\MM}}
 \def\Jk{\mathfrak{J}}

\def\ip{\mathcal{I}^+}
\def\hp{\mathcal{H}^+}

\def\Nv{\,^{(0)}N}

\def\ezero{ ^{(0)}  \hspace{-1pt}    e}
\def\fbzero{\,  ^{(0)} \hspace{-2 pt} \fb}
\def\lazero{\,  ^{(0)}\hspace{-2pt} \la}

\begin{document}

\title{Regularity of the Future Event Horizon in Perturbations of Kerr}
\author{Xuantao Chen  and Sergiu  Klainerman}
\date{}

\maketitle

\begin{abstract}
The goal of the paper  is to show that the event horizons  of the spacetimes  constructed in  \cite{KS}, see also  \cite{KS-Schw}, in the proof
  of the nonlinear stability  of slowly rotating  Kerr    spacetimes $\KK(a_0,m_0)$,   are necessarily smooth null hypersurfaces. Moreover   we show that   the result remains true for 
   the entire  range  of $|a_0|/m_0$   for which  stability
 can be established.
\end{abstract}

\tableofcontents

\section{Introduction}  
    It is known that in dynamical situations  the event horizon\footnote{ Recall  that an embedded    achronal hypersurface  $\HH$ in a Lorentzian manifold          $\MM$ is       a future  horizon if  it   is  ruled by  future null geodesics, i.e. every point $p\in \HH$ belongs to  a future, inextendible,  null geodesic $\Ga\subset\HH$, called a generator (Note that  $\Ga$  is allowed  to have past end points  in $\HH$ but no future  end points) of $\HH$.
    The primary example  is that of the    future event horizon      $\HH^+:=\partial \JJ^{-} (\ip)$      of an   asymptotically flat Lorentzian  manifold  with  complete  future null infinity $\II^+$, see  for example \cite{HE}, page 312 or   \cite{Wald}.} 
 (EV)      is  typically  non-smooth\footnote{  This  is expected to be the case  for  black hole mergers,  see \cite{Ga-Reall}   and the  references within.  The authors  of   \cite{Ga-Reall}    argue  that in   realistic  physical  situations    the  EV is smooth at late times and  take this, see Section 2.2,  as an assumption in their analysis.},   even non-differentiable on a dense set,  see \cite{C-Ga}, \cite{Chrusc}.     
  We  show that,  in contrast to   the  general   situation, the future event horizons  of the spacetimes  constructed in  \cite{KS}, see also  \cite{KS-Schw}  are necessarily smooth null hypersurfaces. 
    Recall that,  according to the main result of \cite{KS},  
the future globally hyperbolic   development  of  a general,   asymptotically  flat,    initial data set, sufficiently close  (in a suitable  topology)  to a   $Kerr(a_0, m_0) $   initial data set,  
  for sufficiently  small $|a_0|/m_0$,  has a complete    future null infinity  $\II^+$ and converges 
in  its causal past  $\JJ^{-1}(\II^{+})$  to another  nearby Kerr spacetime $Kerr(a_f, m_f)$ with parameters    $(a_f, m_f)$,  close to $(a_0, m_0)$, and possesses a future event horizon
 $\HH^+$. The goal of  the present paper  is to show that  $\HH^+$ is  in fact smooth.  We also show  that   the result remains true for  the entire  range  of $|a_0|/m_0$   for which  stability  can be established.

\begin{theorem}[Regularity]
\lab{thm:regularity}
The     event horizon $\hp$  of  the  perturbed  Kerr spacetime $\KK(a,m) $  constructed in \cite{KS}  is   regular\footnote{In the future of the initial Cauchy hypersurface.  See   Definition \ref{Def:RegularNC}  for  the  precise    notion   of regularity  used  here.  We note that  \cite{DHRT} also contains a proof  of   regularity of  EV  for the special case  of  the ultimately  Schwarzschildian   spacetimes   constructed in 
       that work.   See also \cite{CDGH}.}. Moreover,  as already shown in \cite{KS}, in the coordinates used there,
$\HH^+$ lies near $r=r_+:=m_f+\sqrt{m_f^2-a_f^2}$, where $m_f$, $a_f$ are the final  mass and angular momentum. The result  holds in fact   for  the full sub-extremal case  $|a_f|<m_f$  provided that    the     main  estimates\footnote{ It is important
 to note    that  our    results    depend only on  uniform  bounds   with respect to the perturbation parameter   and not 
   on   decay    with respect to  the advanced time $v$, see Remark  \ref{remark:v-decay}.  Such estimates, however,  cannot be proved in the absence  of  a  stability proof.  }  derived in \cite{KS}   remain valid. 
\end{theorem}

We also prove the following  uniqueness result.

\begin{theorem}[Uniqueness]
\label{thm:uniqueness}
 Under the same assumptions, any point $p$ that is not on $\HH^+$ either lies in the past of null infinity, or satisfies the property that, any future null geodesic emanating from it hits the spacelike hypersurface $\{r=r_+-\delta_{\HH}\}$ in the trapping region ($\delta_{\HH}$ is a small positive constant).  In particular, any future complete null hypersurface, different  from  $\HH^{+}$, cannot stay in the region $\{|r-r_+|\leq \delta_{\HH}\}$ where $\HH^+$ is located, see Figure \ref{Figure:Intro}.
\end{theorem}

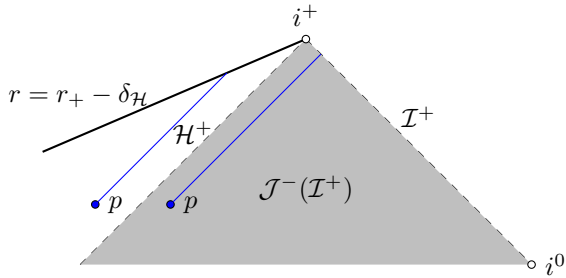
\begin{figure}
\begin{tikzpicture}
\useasboundingbox (-8,-0.5) rectangle (4,4); 
\def\size{3}
    \def\height{3}

    \draw[dashed] (-\size, 0) -- (0,\height) node[midway,above] {${\cal H}^+$};
    \draw[dashed]  (0,\height) -- (\size, 0) node[midway,above=0.2cm] {${\cal I}^+$};

    \draw[thick] (0,\height) -- (-3.5,\height/2) node[midway,left=0.2cm] {$r=r_+-\delta_{\HH}$};
    
    \fill[gray!50] (-\size,0)--(\size,0)--(0,\height)--(-\size,0) ;
    \node at (0,\height/3) {$\JJ^-(\mathcal I^+)$};

    \node[draw, circle, fill=white, inner sep=1pt, label=above:{$i^+$}] at (0,\height) {};
    \node[draw, circle, fill=white, inner sep=1pt, label=right:{$i^0$}] at (\size,0) {};

    \node[draw, circle, fill=blue, inner sep=1pt,label=right:{$p$}] at (-2.8,0.8) {};
    \draw[blue] (-2.8,0.8)--(-21/20,51/20) {}
    ;
    \node[draw, circle, fill=blue, inner sep=1pt,label=right:{$p$}] at (-1.8,0.8) {};
    \draw[blue] (-1.8,0.8)--(0.2,2.8) {}
    ;
\end{tikzpicture}
\caption{The event horizon}\lab{Figure:Intro}
\end{figure}

We describe below the main  features  of the   paper.
\begin{enumerate}

\item The strategy of the proof   for the regularity theorem  \ref{thm:regularity}  is described  
in Section  \ref{section:strategy}.

\item Section  \ref{Sect:preliminaries}  contains a  detailed  description  of  the properties  of the background spacetimes  we  consider,
 based on the   stability results of \cite{KS}, \cite{GKS}.     As mentioned  in the abstract     our results hold     true for   the entire  range  of $|a|/m$   for which  stability  can be established and, moreover,     they   only rely  on 
  uniform bounds and not  on   the more  precise   decay estimates  derived in \cite{KS}, \cite{GKS}, see Remark
   \ref{remark:v-decay}.  Most of our analysis is  restricted to the interior region $\Mint$, more precisely in a small  neighborhood 
   $\RR_\delta:=\{|r-r_+|\leq \delta<1- \frac{|a|}{m} \}$ of  EV,  see \eqref{def:RR_de}.  
    As in \cite{KS}, \cite{GKS},   the  perturbation estimates  in Section \ref{subsect;Kerr-perturbations} are described relative to a non-integrable   frame, yet  through the proof we need to invoke various integrable  frames,   
   such as  that described in Section \ref{subsect:adapted-integr}.  The general  change of frame formulas   are described in  Section \ref{subsect:NullFrameTransformation} and the main estimates    for the integrable frame  are done  in  Section \ref{section:estim-integrable}.
   
  \item  Section \ref{subsection:Causalstructure-nullcone}  contains two   important     criteria for  the regularity of null cones over spheres  as well as  a description of the main steps in the proof of the regularity  theorem,  mentioned  above.
   
 \item   Section \ref{section:RegularityC_vS_v}  contains  a detailed   proof of the steps  described in  Section \ref{section:strategy}. The main quantitative estimates are obtained in Proposition \ref{Prop:Boundsfor-f} and Proposition \ref{prop:estimate-trch'},   which makes use of the sign of $\om$ near the horizon, see Lemma \ref{le:nab-omega-near-extremal}.

 \item Section \ref{section:ProofUniqueness}  contains the proof of the uniqueness theorem.  The main  quantitative ingredient, described in Proposition \ref{prop:u-difference-estimate-uniqueness}, makes again use of the sign of $\om$  near  $\HH^+$.

  \end{enumerate}


\section{Preliminaries}
\label{Sect:preliminaries}
\subsection{Horizontal structures}

We  work within the framework   of general horizontal structures  $\big(e_3, e_4, H=\{e_3, e_4\}^\perp \big) $, where $e_3$ and $e_4$ are null vector fields with  $g(e_3,e_4)=-2$,  introduced  in \cite{GKS}. Given an arbitrary  orthonormal basis $\{e_1,e_2\}$ of $ H$    we  consider the null  frame  $\{e_3,e_4,e_a\}$ ($a=1,2$) and the associated Ricci and curvature  coefficients\footnote{Here $\dual R$ is defined by $\dual R_{\a\b\mu\nu}=\frac 12 \in_{\mu\nu}^{\quad \rho\sigma} R_{\a\b\rho\sigma}$, with $\in$ the volume form on $\MM$.}
\begin{equation*}
    \chi_{ab}=g(D_a e_4,e_b),\quad \chib_{ab}=g(D_a e_3,e_b),      \quad  \eta_a=\frac 12 g(D_3 e_4,e_a),\quad \etab_a=\frac 12 g(D_4 e_3,e_a), \quad \zeta_a=\frac 12 g(D_a e_4,e_3), 
\end{equation*}
\begin{equation*}
    \om=\frac 14 g(D_4 e_4,e_3),\quad \omb=\frac 14(D_3 e_3,e_4),\quad  \xi_a=\frac 12 g(D_4 e_4,e_a), \quad \xib=\frac 12 g(D_3 e_3,e_a),
\end{equation*}
\begin{equation*}
    \a_{ab}=R_{a4b4},\quad \b_a=\frac 12 R_{a434},\quad \rho=\frac 14 R_{3434},\quad \dual\rho=\frac 14\dual R_{3434},\quad \bb_a=\frac 12 R_{a334},\quad \underline \a_{ab}=R_{a3b3}.
\end{equation*}
Here $D$ denotes  the spacetime covariant derivative operator. We further decompose 
\beaa
\chi_{ab}=\chih_{ab}+\frac  12  \trch\, \de_{ab}+\frac 1 2 \atrch \in_{ab},  \qquad  \chib_{ab}=\chibh_{ab}+\frac  12  \trchb\, \de_{ab}+\frac 1 2  \atrchb \in_{ab},
\eeaa
where the trace and anti-trace are defined by
\begin{equation*}
    \trch:=\delta^{ab}\, \chi_{ab},\quad \trchb:=\delta^{ab}\, \chib_{ab},\quad \atrch:=\in^{ab}\chi_{ab},\quad \atrchb:=\in^{ab}\chib_{ab}.
\end{equation*}
and the horizontal volume form $\in_{12}=-\in_{21} =\frac 12\in (e_1, e_2,e_3,e_4)=1$. 
Recall  that   the horizontal structure is integrable if and only if   the asymmetric traces $\atrch, \atrchb$ vanish identically.

For a spacetime vector field $X$, we define its projection onto the horizontal structure $H$ by
\begin{equation*}
    \Xh:=X+\frac 12 g(X,e_3)e_4+\frac 12 g(X,e_4)e_3.
\end{equation*}
A $k$-covariant tensor field $U$ is called horizontal, if 
\begin{equation*}
    U(X_1,\cdots, X_k)=U(\Xh_1,\cdots,\Xh_k).
\end{equation*}
The horizontal covariant derivative operator $\nab$ is defined by
\begin{equation*}
    \nab_X Y:=\, ^{(h)}(D_X Y)=D_X Y-\frac 12 \chib(X,Y)e_4-\frac 12 \chi(X,Y)e_3
\end{equation*}
for two horizontal vector fields $X,Y$. Similarly, one can define $\nab_3 X$ and $\nab_4 X$ as the projections of $D_3 X$ and $D_4 X$. Then the horizontal covariant derivative can be generalized for tensors in the standard way
\begin{equation*}
    \nab_Z U(X_1,\cdots,X_k)=Z (U(X_1,\cdots X_k))-U(\nab_Z X_1,\cdots, X_k) -\cdots - U(X_1, \cdots, \nab_Z X_k),
\end{equation*}
and similarly for $\nab_3 U$ and $\nab_4 U$.

The left dual of a horizontal $1$-form  $\psi$ and a horizontal covariant $2$-tensor $U$  are defined by
\begin{equation*}
    \dual \psi_a:=\in_{ab}\, \psi_b,\quad (\dual U)_{ab}:=\in_{ac}\, U_{cb}.
\end{equation*}
For two horizontal $1$-forms $\psi$, $\phi$, we also define
\begin{equation*}
    \psi\cdot \phi:=\delta^{ab}\psi_a\phi_b,\quad \psi\wedge\phi:=\in^{ab} \psi_a \phi_b,\quad (\psi\hot\phi)_{ab}=\psi_a\phi_b+\psi_b\phi_a-\delta_{ab}\psi\cdot\phi.
\end{equation*}
In particular $|\psi|:=(\psi\cdot\psi)^\frac 12$ with the straightforward generalization to general horizontal covariant tensors. 

Similarly we define the derivative operators
\begin{equation*}
    \div \psi:=\delta^{ab}\, \nab_a \psi_b,\quad \curl \psi:=\, \in^{ab} \nab_a \psi_b,\quad (\nab\hot \psi)_{ab}:=\nab_a \psi_b+\nab_b \psi_a-\delta_{ab}\, \div\psi.
\end{equation*}


\subsection{Null frame transformations}\lab{subsect:NullFrameTransformation}

A general frame transformation 
$(f,\fb,\lambda)$ between two null  horizontal structures  $(e_3, e_4, H)$ and $(e_3', e_4', H')$ can be written in the form, see Section 2.2 in \cite{KS},
\begin{equation}\label{eq:general-null-frame-transformations}
    \begin{split}
    e_4'&=\lambda\left(e_4+f^b e_b+\frac 14|f|^2 e_3\right),\\
    e_a'&=e_a+\frac 12 \fb_a f^b e_b+\frac 12\fb_a e_4+\left(\frac 12 f_a+\frac 18|f|^2 \fb_a\right)e_3,\\
    e_3'&=\lambda^{-1}\left(\left(1+\frac 12 f\cdot\fb+\frac {1}{16}|f|^2|\underline f|^2\right)e_3+\left(\fb_b+\frac 14|\fb|^2 f^b\right)e_b+\frac 14|\fb|^2 e_4\right).
    \end{split}
\end{equation}
The  Ricci coefficients   transform according to   Proposition 2.2.3 in \cite{KS}.  In our work we will consider frame transformations of the  form \footnote{ Note that the general transformation can be obtained by composing this  with a second transformation which reverses the roles of $e_3-e_4$, i.e. with  $f=0$ and 
  $\fb$ nontrivial.}
\bea
\lab{eq:trasportation.formulas-1}
    e_4'=\la\l e_4+f^a e_a+\frac 14 |f|^2 e_3\r,\quad e_a'=e_a+\frac 12 f_a e_3,\quad e_3'=\la^{-1} e_3.
\eea
\begin{definition}\label{def:psit}
If $\psi=\psi_{a_1\ldots a_k}$ is an $H$-horizontal  tensor,  we  define  the  associated  $H'$-horizontal  tensor by the formula
\bea
\lab{eq:def-psi'}
\psit_{a_1\ldots a_k} =\psit( e_{a_1}',\ldots, e_{a_k}')=\psi( e_{a_1},\ldots, e_{a_k}).
\eea
When taking derivatives with respect to  the primed frame, we define  for an $H$-horizontal  $k$-tensor $\psi$
\begin{equation*}
     \nab_b'\psi_{a_1\cdots a_k}:=\nab_b'\tilde\psi_{a_1\cdots a_k}=e_b'(\tilde\psi_{a_1\cdots a_k})-\tilde\psi(D_{e_b'}e_{a_1}',\cdots,e_{a_k}')-\cdots-\tilde\psi (e_{a_1}',\cdots,D_{e_b'}e_{a_k}').
\end{equation*}
We also define $\nab_3'\psi$ and $\nab_4'\psi$ similarly.
By abuse of language,        we  will  often  refer to $\psit$  as simply  $\psi$ \footnote{ In  \cite{KS} and \cite{GKS}   $\psit$   was systematically   denoted  by  $\psi$.}. 
\end{definition}

\begin{lemma}
\label{Le:trasportation.formulas}
Consider a  null frame $\{e_3, e_4, e_a\}$   for which $\xib=0$, i.e.  $e_3$ is proportional to a geodesic vectorfield. Then,
\begin{enumerate}
\item Under the transformation \eqref{eq:trasportation.formulas-1},  we have\footnote{  Note that  $\nab'f$ and $\nab'_4 f$ in the transformation formulas above are defined  according to the definition  \eqref{eq:def-psi'}.}
\bea
\label{eq:trasportation.formulas}
\bsplit
\lambda^{-1}\chi_{ab}' &= \chi_{ab}  +  \nab_a'f_b + f_a \eta_b + f_a \ze_b-\frac 14|f|^2 \chib_{ab}-\omb f_a f_b,\\
\lambda^{-1}\om' &=  \om -\frac{1}{2}\la^{-1}e_4'(\log\la)+\frac{1}{2}f^a (\ze-\etab)_a -\frac{1}{4}|f|^2\omb - \frac{1}{4}f^a f^b \chib_{ab},\\
\lambda^{-2}\xi_a' &= \xi_a +\frac{1}{2}\la^{-1}\nab_4'f_a+\frac 12 f_b \chi_{ba}+\om f_a +\frac 14 |f|^2\eta_a-\frac 14|f|^2 \etab_a+\frac 12 f_a f_b \zeta_b-\frac 18 |f|^2 f_b \chib_{ba},\\
\ze_a' &= \ze_a -e_a'(\log\lambda)-\omb f_a -\frac{1}{2}f_b \chib_{ab},\\
\etab_a' &= \etab_a + \frac{1}{2}f_b \chib_{ba},\\
\lambda\, \chib_{ab}'&=\chib_{ab},\\
\la^{-1}\b_a'&=\b_a+\frac 32(f\rho+\dual f\dual\rho)+\frac 14 |f|^2 \bb_a-f_a f_b \bb_b- f_c \dual f_a \dual\bb_c-\frac 12 |f|^2 f_c \aa_{ac}+f^a f^b f^c\aa_{bc}.
\end{split}
\eea
\item Given  an $H$-horizontal  tensor $\psi$ we have, schematically,
\begin{equation}\lab{eq:nab'i-on-H-tensor}
    \nab'_{a_1}\cdots\nab_{a_i}'\psi=\nab_{e_{a_1}'}\cdots \nab_{e_{a_i}'}\psi+f\cdot\chib\cdot(\nab_{e_{a}'})^{\leq i-1}\psi+\nab'^{\leq i-1}(f\cdot\chib\cdot\psi).
\end{equation}
where the differentiation $\nab'$ is defined in Definition \ref{def:psit}.
\item Similar   formulas   hold true \footnote{By the usual $e_3-e_4$ duality.}    for  frame transformations of the form
\begin{equation*}
    e_4'=e_4,\quad e_a'=e_a+\frac 12 \fb_a e_4,\quad e_3'=e_3+\fb^a e_a+\frac 14 |\fb|^2 e_4,
\end{equation*}
for which  $\xi=0$. 

\end{enumerate}
\end{lemma}
\begin{proof}
The first part of the Lemma is a special case of  Proposition 2.2.3 in \cite{KS}. To prove the second part we proceed as follows:
\begin{equation}\label{eq:nab'-on-H-tensors}
\begin{split}
    \nab_b'\psi_{a_1\cdots a_k}&=\nab_b'\tilde\psi_{a_1\cdots a_k}=e_b'(\tilde\psi_{a_1\cdots a_k})-\tilde\psi(D_{e_b'}e_{a_1}',\cdots,e_{a_k}')-\cdots-\tilde\psi (e_{a_1}',\cdots,D_{e_b'}e_{a_k}')\\
    &=e_b'(\psi_{a_1\cdots a_k})-\tilde\psi(D_{e_b'}e_{a_1}',\cdots,e_{a_k}')-\cdots-\tilde\psi (e_{a_1}',\cdots,D_{e_b'}e_{a_k}')\\
    &=(D_{e_b'}\psi)(e_{a_1},\cdots,e_{a_k})+\psi(D_{e_b'}e_{a_1},\cdots,e_{a_k})+\cdots+\psi(e_{a_1},\cdots,D_{e_b'}e_{a_k})\\
    &\quad -\tilde\psi(D_{e_b'}e_{a_1}',\cdots,e_{a_k}')-\cdots-\tilde\psi (e_{a_1}',\cdots,D_{e_b'}e_{a_k}')\\
    &=\nab_{e_b'}\psi_{a_1\cdots a_k}+\sum_{i=1}^k\left(g(D_{e_b'}e_{a_i},e_c)-g(D_{e_b'}e_{a_i'},e_c')\right)\psi_{a_1\cdots c\cdots a_k}\\
    &=\nab_{e_b'}\psi_{a_1\cdots a_k}+\sum_{i=1}^k\left(-\frac 12 f_{a_i}\chib_{bc}\psi_{a_1\cdots c\cdots a_k}+\frac 12 f_c \chib_{a_i b} \psi_{a_1\cdots c \cdots a_k}\right).
    \end{split}
\end{equation}
Note that all terms on the last line, including $\nab_{e_b'}\psi_{a_1\cdots a_k}=(\nab_b+\frac 12 f_b\nab_3)\psi_{a_1\cdots a_k}$, are $H$-horizontal $(k+1)$-tensors. This serves as a general formula when we apply $\nab'$ to an $H$-horizontal tensor $\psi$.

To write $\nab_a'\nab_b'\psi$ using $\nab_{e_a'}\nab_{e_b'}\psi$, we repeat the calculation in \eqref{eq:nab'-on-H-tensors} with $\psi$ replaced by the last line of \eqref{eq:nab'-on-H-tensors}, and obtain schematically
\begin{equation*}
    \nab_a'\nab_b'\psi=\nab_{e_a'}\nab_{e_b'}\psi+f\cdot \chib \cdot \nab_{e_c'}\psi+\nab'(f\cdot\chib\cdot\psi).
\end{equation*}
Repeating this we obtain, inductively,
\beaa
    \nab'_{a_1}\cdots\nab_{a_i}'\psi=\nab_{e_{a_1}'}\cdots \nab_{e_{a_i}'}\psi+f\cdot\chib\cdot(\nab_{e_{a}'})^{\leq i-1}\psi+\nab'^{\leq i-1}(f\cdot\chib\cdot\psi)
\eeaa
as stated.
\end{proof}
\begin{remark}
    When $\lambda=1$ in \eqref{eq:trasportation.formulas-1}, one can also derive similarly as \eqref{eq:nab'-on-H-tensors} the relation
    \begin{equation*}
        \nab_4'\psi_{a_1\cdots a_k}=\nab_{e_4'}\psi_{a_1\cdots a_k}+\sum_{i=1}^k\left(-\frac 12 f_{a_i}\left(\etab_c+f_b\chib_{bc}\right)\psi_{a_1\cdots c\cdots a_k}+\frac 12 f_c \left(\etab_{a_i}+f_b\chib_{ba_i}\right) \psi_{a_1\cdots c \cdots a_k}\right),
    \end{equation*}
    and \eqref{eq:nab'i-on-H-tensor} can be generalized to
    \begin{equation}\label{eq:nab'andnab_4'on-H-tensors}
        (\nab_4',\nab')^{i}\psi=(\nab_{e_4'},\nab_{e_a'})^{i}\psi+f\cdot \left(\etab,f\cdot \chib\right)\cdot (\nab_{e_4'},\nab_{e_a'})^{\leq i-1}\psi+(\nab_4',\nab')^{\leq i-1} \left(f\cdot (\etab,f\cdot\chib)\cdot \psi\right).
    \end{equation}
\end{remark}

\subsection{Basic  facts about the  Kerr  spacetime \texorpdfstring{$\KK(a,m)$}{K(a,m)}}
\subsubsection{Boyer-Lindquist (BL) coordinates}


Relative to the BL coordinates, the  Kerr  metric  $\g=\g_{a,m}$  takes  the form
 \beaa
\g &=& \,    -\frac{\left(\Delta-a^2\sin^2\theta\right)}{|q|^2}dt^2-
\frac{ 4 amr }{|q|^2}   \sin^2\theta  \,   dt  d\phi+\frac{|q|^2}{\Delta}dr^2+ |q|^2 d\theta^2+
\frac{ \Si^2}{|q|^2}\sin^2\theta
\, d\phi^2.
\eeaa
where $q=r+ i a \cos\th$ and
\beaa
\left\{\ba{lll}
\Delta &=& r^2-2mr+a^2,\\
|q|^2 &=& r^2+a^2(\cos\theta)^2,\\
\Sigma^2 &=& (r^2+a^2)|q|^2+2mra^2(\sin\theta)^2=(r^2+a^2)^2-a^2(\sin\theta)^2\Delta.
\ea\right.
\eeaa
The function $\Delta=r^2-2mr+a^2$ has two zeros  $ r_\pm=m\pm \sqrt{m^2-a^2}$  with 
 $\{r=r_+\}$  the event horizon of  $\KK(a, m)$.
The ingoing principal null (PN) frame, regular towards the future  for all $r>0$,  is given by
\bea
\lab{eq:null-pair-in}
\bsplit
e_4=\frac{r^2+a^2}{|q|^2} \pr_t +\frac{\De}{|q|^2} \pr_r +\frac{a}{|q|^2} \pr_\phi, \qquad
e_3   =\frac{r^2+a^2}{\De} \pr_t -\pr_r +\frac{a}{\De} \pr_\phi.
\end{split}
\eea
The canonical horizontal basis  is given by 
\bea
\lab{eq:canonicalHorizBasisKerr}
e_1=\frac{1}{|q|}\pr_\th,\qquad e_2=\frac{a\sin\th}{|q|}\pr_t+\frac{1}{|q|\sin\th}\pr_\phi.
\eea

\subsubsection{Ingoing Eddington-Finkelstein coordinates}

The ingoing Eddington-Finkelstein coordinates 
$(v , r, \th, \vphi)$  are given by 
\beaa
v:= t+ f(r), \quad f'(r)=\frac{r^2+ a^2}{ \De}, \qquad \vphi:= \phi +h(r), \quad h'(r)=\frac{a}{\Delta},
\eeaa
such that,
\beaa
e_3(r)=-1, \quad  e_3(v)=e_3(\th)=e_3(\vphi)=0, \qquad e_1(r)=e_2(r)=0, \quad e_1(v)=0, \quad e_2(v)=\frac{a\sin\th}{|q|}.
\eeaa
Thus,
\beaa
    e_4=2\frac{r^2+a^2}{|q|^2}\pa_v+\frac{\Delta}{|q|^2}\pa_r+2\frac{a}{|q|^2}\pa_\vphi,\qquad  e_3=-\pa_r,\qquad e_1=\frac{1}{|q|}\pa_\theta,\qquad e_2=\frac{a\sin\theta}{|q|}\pa_v+\frac{1}{|q|\sin\theta}\pa_\vphi.
\eeaa
The metric expression in the ingoing Eddington-Finkelstein coordinates is given by
\begin{equation}\label{eq:metric-expression-exact-Kerr-ingoingEF}
    \g_{a,m}=
    -\left(1-\frac{2mr}{|q|^2}\right)(dv-a\sin^2 \theta\, d\varphi)^2+2(dv-a\sin^2\theta\, d\varphi)(dr-a\sin^2\theta\, d\varphi)+|q|^2(d\theta^2+\sin^2\theta\, d\varphi^2).
\end{equation}

 In the language of  \cite{KS}  (chapters 2,3),  the nonvanishing  Ricci and curvature  coefficients are given by
\begin{equation*}
    \tr X=\frac{2\overline{q}\Delta}{|q|^4}, \quad \tr \underline X=-\frac{2}{\overline q},\quad Z=H=\frac{aq}{|q|^2}\mathfrak J,\quad \underline H=-\frac{a\overline q}{|q|^2}\mathfrak J,\quad \om=-\frac{1}2 \pa_r\left (\frac{\Delta}{|q|^2}\right),\quad P=-\frac{2m}{q^3},
\end{equation*}
where $X=\chi+i\dual\chi$, $\underline X=\chib+i\dual\chib$, $H=\eta+i\dual\eta$, $\underline H=\etab+i\dual\etab$, $Z=\zeta+i\dual\zeta$, $P=\rho+i\dual\rho$ are complexified Ricci coefficients and curvature components, and $\Jk$ is a horizontal $1$-form satisfying $\Jk_1=i\sin\theta\, |q|^{-1}$, $\Jk_2=\sin\theta\, |q|^{-1}$ in the canonical horizontal basis.

\subsection{Perturbation of Kerr according to \cite{KS}}
\lab{subsect;Kerr-perturbations}

The global spacetime  constructed in \cite{KS}    is of the form $\MM=  \Mint\cup\Mext  $ where  $\Mint=\{r\leq r_0\}$ is covered by an ingoing principal geodesic\footnote{See Section 2.3.1 in \cite{KS} for the definition of PG structures. Note that  Kerr has a canonical  ingoing PG structure   given by the Eddington-Finkelstein coordinates and associated frame. }     (PG) structure, and $\Mext=\{r\geq r_0\}$ by an outgoing PG structure, with $r_0$  a sufficiently  large constant.  The perturbed spacetime has a final mass $m_f$ and $a_f$ which we denote  $m$ and $a$ for simplicity. We define,  using the same expression as in the stationary case:
\begin{equation}
\lab{equation; r_+}
    r_+=m+\sqrt{m^2-a^2},\quad \Delta=r^2-2mr+a^2.
\end{equation}

\subsubsection{Quantitative bounds in $\Mint$}
\lab{section:Quantbonds-Mint}

According to \cite{KS}  $\Mint$ is equipped with an ingoing PG structure, consisting of coordinates $(v,r,\theta,\varphi)$, $(v,r,x^1,x^2)$, and a null frame $\{e_3,e_4,e_a\}$, satisfying
\begin{equation*}
    D_3 e_3=0,\quad e_3(r)=-1, \quad e_a(r)=0,\quad e_3(v)=e_3(\theta)=e_3(\varphi)=e_3(x^1)=e_3(x^2)=0.
\end{equation*} 
The intersecting spheres of level sets of $v$ and $r$ are denoted  by     $S(v,r)$.
The two coordinate systems $S(v,r)$ are  related by
\begin{equation*}
    x^1=\sin\theta \cos\varphi,\quad x^2=\sin\theta \sin\varphi.
\end{equation*}
In \cite{KS},  it was shown that in $\Mint$ the        metric $g$ behaves as  follows \footnote{In this paper, by $A\lesssim B$ we mean $A\leq CB$ for some $C>0$ independent of $A,B$, and the location in the spacetime. If the implicit constant is dependent on a parameter $s$, we denote it by $A\lesssim_s B$. We also use $O(B)$ to denote a quantity $A$ satisfying $|A|\lesssim |B|$.}:
\bea
\label{eq:perturbed-metric-expression}
\begin{split}
    g&=\g_{a,m}+(dv,dr,rd\theta,r\sin\theta\, d\varphi)^2\,  O(\epsilon_0\, v^{-1-\delta_{dec}}),\quad \frac\pi 4<\theta<\frac{3\pi}4,
\\
    g&=\g_{a,m}+(dv,dr,rdx^1,rdx^2)^2\, O(\epsilon_0 \, v^{-1-\delta_{dec}}),\qquad\,\,\, 0\leq \theta<\frac \pi 3 \text{ or }\frac{2\pi}3<\theta\leq\pi.
    \end{split}
\eea
\begin{remark}[Main constants] Here $\epsilon_0,\delta_{dec}>0$ are small constants. Later we will also  introduce  the constants $\delta_{\HH}$ and $\delta_*:=1-|a|/m$. They satisfy the relation $\epsilon_0\ll \min\{\delta_{dec},\delta_{\HH},\delta_*\}$. We choose  a small constant $\delta$ such that $\epsilon_0\ll \delta \leq \min\{\delta_{\HH}/2,\delta_*\}$.
\end{remark}

To state the estimates for the Ricci and curvature  coefficients in \cite{KS}, we need to  recall the definition of  the complex horizontal $1$-form $\Jk$. In $\Mint$,  according to  \cite{KS}, $\Jk$ satisfies 
\beaa
\nab_3\Jk=\frac 1{\overline{q}}\Jk,\qquad  \dual \Jk=-i\Jk,\qquad  \Jk\cdot\overline{\Jk}=\frac{2(\sin\theta)^2}{|q|^2}.
\eeaa
We have $\Jk=j+i\dual j$ where $j$ is a real horizontal $1$-form, which satisfies
\begin{equation}\label{eq:nabla3onj}
    \nab_3 j=|q|^{-2}(rj-a\cos\theta\dual j),\quad {\nab_3 \dual j}=|q|^{-2}(r\dual j+a\cos\theta j),\quad |j|^2=\frac{(\sin\theta)^2}{|q|^2}.
\end{equation} 
\begin{definition}
 \lab{def:linearizedq}
 We define $\Gac$ to be the set of the following quantities\footnote{ In our situation, we only need the real  form of these coefficients,  which can be easily  translated  from the corresponding  complexified quantities in \cite{KS}.} 
 \beaa
&& {\xi},\, \widecheck{{\om}},\, \widecheck{{\trch}},\, \widecheck{^{(a)}{\trch}},\, {\chih},\, {\chibh},\, \widecheck{{\zeta}},\, \widecheck{{\etab}},\, \widecheck{{\trchb}},\, \widecheck{^{(a)}{\trchb}},\, r\widecheck {{\rho}},\, r\dual\rho,\, r{\beta},\, r{\a},\, r{\underline \b},\, {\underline \a},\, r\widecheck{e_4(r)},\, r\widecheck{e_4(v)}, \, \widecheck{\nabb v},\\
&&  r{\mathrm{div}} j,\, r\widecheck{{\mathrm{curl}} j},\, \widecheck{\nabb_4\, j},\, \widecheck{\nabb_4 \dual j},\, r\nabb\hot j,\, \widecheck{\nabb(\cos\theta)},\, re_4(\cos\theta)\
 \eeaa
 where 
 the  checked, linearized quantities are defined by
 \bea
 \lab{eq:linearization}
 \begin{split}
 &  \widecheck{{\om}}:={\om}+\frac 12\pa_r(\frac{\Delta}{|q|^2}),\quad \widecheck{{\trch}}:={\trch}-\frac{2r\Delta}{|q|^4},\quad \widecheck{\anti{{\trch}}}:=\anti{{\trch}}-\frac{2a\cos\theta \Delta}{|q|^4}, \quad   \widecheck{{\zeta}}={\zeta}-a\frac{rj-a\cos\theta \dual j}{|q|^2},\\
 &  \widecheck{{\etab}}={\etab}+a\frac{rj+a\cos\theta\dual j}{|q|^2}, \quad \widecheck{{\trchb}}:={\trchb}+\frac{2r}{|q|^2},\quad \widecheck{^{(a)}{\trchb}}:=\anti{{\trchb}}-\frac{2a\cos\theta}{|q|^2},\quad \widecheck {{\rho}}={\rho}+\Re(\frac{2m}{q^3})\\
 &   \widecheck{e _4(r)}=e _4(r)-\frac{\Delta}{|q|^2},\quad \widecheck{e _4(v)}=e _4(v)-\frac{2(r^2+a^2)}{|q|^2},\quad \widecheck{\nabb v}=\nabb v-aj,\quad \widecheck{\nabb (\cos\theta)}:=\nabb (\cos\theta)+\dual j\\
 &\widecheck{\mathrm{curl}\, j}=\mathrm{curl}\, j-\frac{2(r^2+a^2)\cos\theta}{|q|^4},\quad 
    \widecheck{\nabb_4 j}=\nabb_4 j-\frac{\Delta(rj+a\cos\theta ^*j)}{|q|^4},\quad \widecheck{\nabb_4\,  ^*j}=\nabb_4 \dual j+\frac{\Delta(r\dual j-a\cos\theta\, j)}{|q|^4}.
 \end{split}
 \eea
 \end{definition}
Note that under an ingoing PG structure, we have $\xib=0$, $\omb=0$, $\eta=\zeta$.
 
The estimates in $ \Mint$ in \cite[Section 3.4.3]{KS} provide  the  bounds
\bea
    |(\nabb_3,\nabb_4,\nabb)^{\leq N}\widecheck{{\Gamma}}|\leq C\epsilon_0 v^{-1-\delta_{dec}}.
\eea
Here $N=k_{small}$ in \cite{KS}. 
\begin{remark}
\lab{remark:v-decay}
The validity of  our theorems     depend only on  uniform  bounds   with respect to the perturbation parameter  $\ep_0$   and not 
   on   decay    with respect to  the advanced time $v$, i.e., we only use the estimate
   \bea
\label{eq:estim.Gac}
    |(\nabb_3,\nabb_4,\nabb)^{\leq N}\widecheck{{\Gamma}}|\leq C\epsilon_0.
\eea
\end{remark}

We will need to deal with derivatives of the real $1$-form $j$ when changing frames.

\begin{lemma}
{We have
  \bea
  \label{eq:Bonds-j}
   |(\nabb,\nabb_3,\nabb_4)^{\leq N} (j,\dual j)|\les   1 .
   \eea }
\end{lemma}
\begin{proof}
In view of the linearization of $\nab_4 j$, $\curl j$, see \eqref{eq:linearization}, and \eqref{eq:nabla3onj}, along with \eqref{eq:estim.Gac}, in order to control $(\nab,\nab_3,\nab_4)^{\leq k}(j,\dual j)$, it suffices to control $(\nab,\nab_3,\nab_4)^{\leq k}\cos \theta$ and $(\nab,\nab_3,\nab_4)^{\leq k-1}(j,\dual j)$. The boundedness of $(j,\dual j)$ has been established in \eqref{eq:nabla3onj}, so  by  induction, we assume $|(\nab,\nab_3,\nab_4)^{\leq k-1}(j,\dual j,\cos\theta)|\lesssim 1$. Then we have
\begin{equation*}
    |(\nab,\nab_3,\nab_4)^{k}\cos \theta|\lesssim |(\nab,\nab_3,\nab_4)^{\leq k-1} (\nab_4 \cos\theta,\nab_3 \cos\theta,\nab\cos\theta)|\lesssim |(\nab,\nab_3,\nab_4)^{\leq k-1}(\Gac,\dual j)|\lesssim 1,
\end{equation*} 
where we used $e_3 (\cos\theta)=0$, $re_4(\cos\theta)\in \Gac$, and $\widecheck{\nabb (\cos\theta)}:=\nabb (\cos\theta)+\dual j $.
\end{proof}

\subsubsection{Full subextremal case}
The result in \cite{KS} deals with the perturbation of Kerr with small angular momentum $a$. Therefore, the estimates above are only known to be true when $|a|/m\ll 1$. Nevertheless, in this work, we consider  the  general case   by assuming similar estimates hold true\footnote{We  implicitly assume that any  proof of stability will produce such estimates.}    in the  full sub-extremal case  $|a|/ m<1$, with the size of perturbation $\epsilon_0$   much less than $\delta_*=1-|a|/m$. 

\begin{lemma}
\lab{le:nab-omega-near-extremal}
Assuming that    \eqref{eq:estim.Gac} holds for the full sub-extremal range  $|a|/m <1$ we have, in the region $\{|r-r_+|\leq \de_*\}$ with $\delta$ a small constant satisfying $\delta\leq \delta_*$, we have
\bea
\label{eq:nab-omega-near-extremal}
\om\leq -\frac{1}{32 m^2}\delta_*^\frac 12,\qquad |\nab^{\leq N}  \om| \les O(\delta_*^\frac 12).
\eea

\end{lemma}  
\begin{proof}

In  the exact Kerr $\KK(a,m)$,   under the principal ingoing PG frame, we have
 \beaa
 \om=-\frac{1}2 \pa_r\left (\frac{\Delta}{|q|^2}\right)=-\frac{a^2\cos^2\theta(r-m)+mr^2-a^2r}{|q|^4}.
 \eeaa
 Recall that $r_+=m+\sqrt{m^2-a^2}= m+ m\sqrt{1-(a/m)^2}= m +m\sqrt{\de_* (2-\de_*)}=  m+O(\de_*^{1/2} )$. 
 For $r=r_+$,
 \beaa
 a^2\cos^2\theta(r-m)+mr^2-a^2r\ge &= & a^2\cos^2\theta \sqrt{\de_* (2-\de_*)}+ r_+\big( \big(m+     \sqrt{\de_* (2-\de_*)}\big) m-a^2\big)\\
 &=&  \sqrt{\de_* (2-\de_*)} \big(a^2 \cos^2 \th+ r_+ m\big)+ r_+\big( m^2 - a^2)\\
 &=& \sqrt{\de_* (2-\de_*)} \big(a^2 \cos^2 \th+ r_+ m\big)+  r_+m^2\de_*(2-\de_*)
 \\
 &=&   m r_+ \sqrt{\de_* (2-\de_*)}  \big(1+   \frac{a^2}{ m r_+} \cos^2 \th+\sqrt{\de_* (2-\de_*)} \big)\\
 &\ge &   m r_+ \sqrt{\de_* (2-\de_*)}. 
 \eeaa
  We  deduce,
 \beaa
 -\om|_{r=r_+}\ge  \frac{ m r_+ \sqrt{\de_* (1+\de_*)} }{|q|^4}\ge    \frac{\de_*^{1/2} }{16 m^2}.
 \eeaa
Therefore, with a perturbation of size $O(\epsilon_0)$,  $\epsilon_0 \ll \de_*$,  in the region $\{|r-r_+| \le \de\}$, we have
 \beaa
 \om\le -   \frac{\de_*^{1/2} }{32 m^2}.
 \eeaa
Also, since $\nab(r)=0$, for  $|r-r_+|\leq \de\leq \de_*$, 
\beaa
|    \nab^{\leq N} \om|=O(|r-m|,|mr^2-a^2 r|)=O\left(\sqrt{1-|a|/m}\right)=O(\delta_*^\frac 12).
\eeaa
With the assumption that similar estimates in the last subsection hold, we see that the estimates regarding $\om$ here also hold under an $O(\epsilon_0)$ perturbation.
\end{proof}

\subsection{Adapted integrable frame in $\Mint$}
\lab{subsect:adapted-integr}
In addition to the non-integrable ingoing PG frame in $\Mint$, we introduce a related integrable frame adapted to the spheres $S(v,r)$.

{\subsubsection{Adapted integrable frame in Kerr}}
\begin{lemma}
\lab{Lemma:Mots-Kerr}For the ingoing PG structure in Kerr,
there exists a canonical  integrable  horizontal structure    $(\eS_3, \eS_4,  \, ^{(S)}H)$   compatible\footnote{ i.e. $\, ^{(S)}H $ tangent to $S(v, r)$.} with   the sphere $S(v, r)$ with transformation parameters $(F, \Fb, \la=1)$ such that
\begin{itemize}
\item We have
\bea\label{F-FbKerr}
F=-\frac{4e_4(r)\nab v}{e_4(v)+\sqrt{|e_4(v)|^2-4e_4(r)|\nab v|^2}}, \qquad \Fb=-\frac{2\nab v}{\sqrt{|e_4(v)|^2-4e_4(r)|\nab v|^2}}.
\eea
where $e_4(r)=\frac{\De}{|q|^2}$ and $e_4(v)=2\frac{r^2+a^2}{|q|^2}$.
\item The spheres $S(v,r_+)$ are marginally outer trapped surfaces (MOTS), i.e. $\trchS=0.$
\end{itemize}
\end{lemma}
\begin{proof}
 Denote the frame transformation from $\{e_3, e_4, e_a\}$ to $\{\eS_3,\eS_4,\eS_a\}$ by $(F,\underline F)$. 
Recall that $e_3(r)=-1$, $e _3(v)=0$, $e_a(r)=0$.  According to \eqref{eq:general-null-frame-transformations}  we look for 
a horizontal structure  spanned by 
\beaa
    \eS_a&=e_a+\frac 12 \Fb_a F^b e_b+\frac 12\Fb_a e_4+\left(\frac 12 F_a+\frac 18|F|^2 \Fb_a\right)e_3,
   \eeaa
such that  $  \eS_a(v)=0,\, \eS_a(r)=0.$  
  Therefore
\beaa
0&=& e_a(v) + \frac 12 \Fb_a F^b e_b(v) +\frac 12\Fb_a e_4(v)\\
0&=&  \frac 12\Fb_a e_4(r)-\left(\frac 12 F_a+\frac 18|F|^2 \Fb_a\right).
\eeaa
From the second  equation we deduce
\bea
\lab{eq:F-Fb}
    \Big(e_4(r)-\frac 14|F|^2\Big)\Fb=F, \qquad \text{with }\, e_4(r)=\frac{\De}{|q|^2}.
\eea
The      first  equation takes the form
\beaa
\nab v +\frac 12\Fb (F\cdot\nab v)+\frac 12 \Fb e _4(v)=0.
\eeaa
We look for an $F$ of the form  $F=h\nab v$. Thus,
\beaa
\nab  v +\frac 12 h |\nab v|^2  \Fb  +\frac 12 \Fb e _4(v)=0.
\eeaa
Multiplying by  $e _4(r)-\frac 14|F|^2$ and using \eqref{eq:F-Fb} we deduce
\begin{equation*}
    \Big(e _4(r)-\frac 14 h |\nab v|^2\Big) h\nab v+\frac 12 h |\nab v|^2 h\nab v  +\frac 12(h\nab v)e _4(v)=0.
\end{equation*}
We deduce
\begin{equation*}
    \frac 14|\nab v|^2 h^2 \nab v+ \Big(\frac 12 e _4(v) \nab v      + he _4(r)\Big) \nab v=0
\end{equation*}
and therefore, {as long as $|\nab v|\neq 0$},
\bea
\lab{eq:h^2}
   |\nab v|^2  h^2  + 2    e_4(v) h     + 4 e _4(r) =0
\eea
or, taking the smaller root\footnote{It is straightforward to verify that the discriminant is strictly positive in $\Mint$.},
\beaa
   h=\frac{-e_4(v)+\sqrt{|e_4(v)|^2-4e_4(r)|\nab v|^2}}{|\nab v|^2}=-\frac{4e_4(r)}{e_4(v)+\sqrt{|e_4(v)|^2-4e_4(r)|\nab v|^2}}.
\eeaa
Recall that $e_4(r)=\frac{\De}{|q|^2}, \, e_4(v)=2\frac{r^2+a^2}{|q|^2}$. Also, whenever $e_4(r)\neq 0$, we have
\beaa
    \Fb=\frac{1}{e_4(r)-\frac 14|F|^2}F=\frac{h\nab v}{e_4(r)-\frac 14 h^2|\nab v|^2}   =\frac{h\nab v}{2e_4(r)+\frac 12e_4(v)h}=\frac{-2\nab v}{\sqrt{|e_4(v)|^2-4e_4(r)|\nab v|^2}}
    \eeaa
    where we used \eqref{eq:h^2}. When $e_4(r)=0$, there is in fact no constraint on $\Fb$, so we can simply use the same expression for $\Fb$.

  Finally  in view of the transformation formula \eqref{eq:trasportation.formulas},  from the standard frame  to the $S$-frame, using the  fact that  $F$ and $\trch$ vanish
  everywhere on $\{r=r_+\}$, we have, on $S(v,r_+)$,
  \begin{equation*}
    \begin{split}\trchS&=\delta^{ab} g(D_{\eS_a} \eS_4,\eS_b)=\delta^{ab} g\Big(D_{e_a+\frac 12\Fb_a e_4}\big(e_4+F^b e_b+\frac 14 |F|^2 e_3\big),e_a+\frac 12\Fb_a e_4\Big)\Big|_{r=r_+}\\
    &=\trch +\frac 12 \Fb \cdot \xi=\trch=0
    \end{split}
\end{equation*}
as stated.
    \end{proof}

\subsubsection{Adapted integrable frame in $\Mint$}
Similarly, we can obtain an adapted integrable frame from the PG frame in the perturbed spacetime.
\begin{lemma}
\lab{Lemma:Mots-Kerr(ep_0)}
Consider the perturbation of Kerr spacetime $\MM=\Mint\cup \Mext$  constructed in \cite{KS}, with the ingoing PG structure in $\Mint$. Then there exists a null frame $(\eS_3,\eS_4,\eS_a)$ on $\Mint$, with an integrable horizontal structure adapted to the spheres $S(v,r)$. The transformation parameters $(F,\Fb,\lambda=1)$ from the ingoing PG frame to the new integrable frame can be written as
 \bea
 \lab{eq:FFb-perturb}
F=-\frac{4e_4(r)\nab v}{e_4(v)+\sqrt{|e_4(v)|^2-4e_4(r)|\nab v|^2}}, \qquad \Fb=-\frac{2\nab v}{\sqrt{|e_4(v)|^2-4e_4(r)|\nab v|^2}}.
\eea
\end{lemma}
\begin{proof}
The procedure is identical to the one in Lemma \ref{Lemma:Mots-Kerr}.
\end{proof}

\def\D{\mathcal{D}}

\subsection{Estimates of the integrable frame  near the event horizon}
\lab{section:estim-integrable}

 We restrict our discussion in a neighborhood of  $\HH^+$, i.e.   the region
 \bea
 \lab{def:RR_de}
 \RR_\delta:=\{|r-r_+|\leq \delta\}, \quad \mbox{where}\quad  \epsilon_0\ll \delta\leq \min\{\delta_*,\delta_{\HH}/2\}.
 \eea
 with     $ \delta_*=1-\frac{|a|}{m} $. Here $\delta_{\HH}$ is a small constant  used  in \cite{KS}, with $\{r=r_+-\delta_{\HH}\}$ being the future spacelike boundary of the spacetime constructed there.
   
We have the following bounds  for   the transformation  parameters  $(F,\underline F)$.
\begin{proposition}
\label{Prop:boundsF-Fb}
    In the region $\RR_\delta$, we have
    \begin{equation*}
        |(\nabb_4,\nabb)^{\leq N} F|\leq C_N\min\{\delta,a\},\quad |(\nabb_3,\nabb_4,\nabb)^{\leq N} F|\leq C_N a,\qquad        |(\nabb_3,\nabb_4,\nabb)^{\leq N}\underline F|\leq C_N a.
    \end{equation*}
  
\end{proposition}
We note that if we only consider the case of the small angular momentum, i.e. the context of \cite{KS}, then all constants on the right-hand side can be replaced by $\epsilon_0$.
\begin{proof}
The bounds  proportional to $a$ are direct corollaries  of the estimates \eqref{eq:estim.Gac},  \eqref{eq:Bonds-j}  and  the formulas  for $F, \Fb$ in 
 \eqref{eq:FFb-perturb}, which are   both    proportional to  $\nabb v=\widecheck{\nabb v}+aj$. To check the  $O(\de)$ bound for $F$ it suffices to notice that
 its expression is proportional  to $e_4(r)=\frac{\Delta}{|q|^2}+\widecheck{e_4(r)}=O(\de)$. Moreover  this bound is preserved  by taking higher derivatives in $e_4, e_a$.
 Indeed,
\beaa
    e _4\left(e _4(r)
    \right)&=&e_4\left(\widecheck{e _4(r)}\right)+e _4\left(\frac{\Delta}{|q|^2}\right)=e _4\left(\widecheck{e _4(r)}\right)+e _4(r)\left(\pa_r\Delta\right)\left(\frac{1}{|q|^2}\right)+\Delta\,  e _4\left(\frac{1}{|q|^2}\right),\\
     e _a\left(e _4(r)\right)&=&e _a\left(\widecheck{e _4(r)}\right)+\Delta\, e _a\left(\frac{1}{|q|^2}\right),
\eeaa
so there is always a $\Delta$ or $e _4(r)$ factor, ensuring the $\delta$-smallness. Note that $e_4 (\Delta)$ gives $e_4(r)$, so along with $e_a(r)=0$, we see that a similar estimate holds with higher order differentiations in $\nabb_4$ and $\nabb$.
\end{proof}

\begin{corollary}\label{Cor:HigherboundF}
    We have $|(\nab_{\eS_{a}},\nab_{\eS_4})^{\leq N} F|\leq \min\{\delta,a\}$. 
\end{corollary}
Note that we are still using the $H$-horizontal operator $\nab$. Recall that $\eS_a=e_a+\frac 12 \Fb_a F^b e_b+\frac 12 \Fb_a e_4+(\frac 12 F_a+\frac 18|F|^2 \Fb_a)e_3$, so $\nab_{\eS_a}$ is defined by $\nab_{\eS_a}=\nab_a+\frac 12 \Fb_a F^b \nab_b+\frac 12 \Fb_a \nab_4+(\frac 12 F_a+\frac 18|F|^2 \Fb_a)\nab_3$. Similarly $\nab_{\eS_4}=\nab_4+F^a \nab_a+\frac 14 |F|^2 \nab_3$.
\begin{proof}
    The only derivative that does not give $\delta$-smallness is the $e_3$ direction, but one can see from the expression
    of $\nab_{\eS_a}$ and $\nab_{\eS_4}$ that $\nab_3$ is always paired with a factor in $F$, which provides a $\delta$-smallness factor. Therefore, the estimate is a direct corollary of Proposition \ref{Prop:boundsF-Fb}.
\end{proof}

\subsection{The time function $\tau$}\label{subsect:timefunction}
We  endow     the region $\{|r-r_+|<2\delta\}$ with a spacelike foliation $\Si_\tau$,   given    by  the level hypersurfaces of the  function  $\tau:=v-r$.
  One may extend $\Sigma_\tau$ to the whole spacetime, but here it suffices to focus on $\{|r-r_+|<2\delta\}$.
To see that $\tau$ is indeed a time function and  therefore $\Si_\tau$  are  spacelike, we compute (for $|r-r_+|\leq 2\delta$),
\begin{equation*}
\begin{split}
    g(\grad\, \tau,\grad\, \tau)&=\pa^\mu \tau\,  \pa_\mu \tau=-e_3(\tau) e_4(\tau)+e_a(\tau)e_a(\tau)=e_3(r) e_4(v-r)+e_a(v) e_a(v)\\
    &=-(e_4(v)+O(\delta))+|\nab v|^2\\
    &=-2\frac{r^2+a^2}{|q|^2}+\frac{a^2\sin^2\theta}{|q|^2}+O(\delta)+O(\epsilon_0)\\
    &=-\frac{2r^2+a^2\cos^2\theta}{|q|^2}+O(\delta)\leq -1<0,
\end{split}
\end{equation*}
so $\tau$ is indeed a time function.

\def\D{\mathcal{D}}
\def\const{\mathrm{const}}

{\bf The coordinate system on $\Sigma_\tau$.} 
From the spacetime coordinates $(v,r,\theta,\varphi)$ (and $(v,r,x^1,x^2)$), we see that $(r,\theta,\varphi)$ (and $(r,x^1,x^2)$) become coordinate systems on $\Sigma_\tau=\{v-r=\tau\}$ for each $\tau$. The coordinate basis can be expressed by the spacetime coordinate basis by $(\overline \pa_r,\overline \pa_\theta,\overline \pa_\vphi)=(\pa_r+\pa_v,\pa_\theta,\pa_\varphi)$ (similarly for the $(r,x^1,x^2)$ coordinate). 

In fact, by \eqref{eq:perturbed-metric-expression}, there exists a diffeomorphism
\begin{equation}\label{eq:def-diffeomorphism-Sigma0}
    \Phi_{\Sigma_\tau}\colon (r_+-2\delta,r_+ +2\delta)\times \mathbb{S}^2 \to \Sigma_\tau \cap \{|r-r_+|<2\delta\}
\end{equation}
where $(\theta,\varphi)$ (and $(x^1=\sin\theta\cos\vphi,x^2=\sin\theta\sin\vphi)$) are the spherical coordinates on $\mathbb{S}^2$, where the pullback metric, denoted by $\Phi_{\Sigma_\tau}^* \overline g$, or simply $\overline g$, satisfies
\begin{equation*}
    \overline g=\overline{\g}_{a,m}+O(\epsilon_0)
\end{equation*}
where $
\overline{\g}_{a,m}$ is the correponding pullback metric induced on $\{\tau=v-r=\const\}$ in exact Kerr. One can compute from the expression \eqref{eq:metric-expression-exact-Kerr-ingoingEF} that
\begin{equation*}
\begin{split}
    \overline\g_{a,m}&=\left(1+\frac{2mr}{|q|^2}\right) dr^2-2\left(1+\frac{2mr}{|q|^2} \right)a\sin^2\theta drd\vphi+|q|^2 d\theta^2+\left(\Big(1+\frac{2mr}{|q|^2}\Big) a^2 \sin^4\theta+|q|^2\sin^2\theta\right) d\vphi^2.
    \end{split}
\end{equation*}  
One can also derive the expression in $(r,x^1,x^2)$ coordinates by straightforward calculation:
\begin{equation*}
\begin{split}
    \overline \g_{a,m}&=\left(1+\frac{2mr}{|q|^2}\right) dr^2+2\left(1+\frac{2mr}{|q|^2}\right)\left(a x^2 dr dx^1-a x^1 dr dx^2\right) \\
    &+\left(a^2 \left(1+\frac{2mr}{|q|^2}\right) (x^2)^2 +|q|^2\frac{1-(x^2)^2}{1-|x|^2}\right)(dx^1)^2\\
    &+2\left(-a^2 \left(1+\frac{2mr}{|q|^2}\right) x^1 x^2+|q|^2\frac{x^1 x^2}{1-|x|^2}\right)dx^1 dx^2\\
    &\quad +\left(a^2 \left(1+\frac{2mr}{|q|^2}\right) (x^1)^2+|q|^2\frac{1-(x^1)^2}{1-|x|^2}\right)(dx^2)^2,
    \end{split}
\end{equation*}
which is regular away from $\theta=\pi/2$.

In particular, we see that
\begin{equation}\label{eq:grr}
    1\leq g(\overline\pa_r,\overline\pa_r)=1+\frac{2mr}{|q|^2}+O(\delta)\leq 3.
\end{equation}
It is also clear from the metric expression that there exists a constant only dependent on $a,m,\delta$ such that
\begin{equation}\label{eq:equilvalence-metrics-on-Sigma}
    C^{-1}  \overline g_0\leq \overline g\leq C\overline g_0,\quad |r-r_+|\leq 2\delta.
\end{equation}
where $\overline g_0$ denotes the induced Euclidean metric when we embedd $(r_+ -2\delta,r_+ +2\delta)\times \mathbb{S}^2$ into $\mathbb{R}^3$ using the standard spherical coordinates.

We will frequently use this diffeomorphism $\Phi_{\Sigma_\tau}$ from \eqref{eq:def-diffeomorphism-Sigma0}. Denote the natural projection $(r_+ -2\delta,r_+ +2\delta)\times \mathbb{S}^2\to \mathbb{S}^2$ by $P_{\mathbb{S}^2}$. Then we also define
\begin{equation}\label{eq:def-projection-S^2}
    P_{\Sigma_\tau;\mathbb{S}^2}:=P_{\mathbb{S}^2}\circ \Phi_{\Sigma_\tau}^{-1}\colon \Sigma_\tau \cap \{|r-r_+|<2\delta\}\to \mathbb{S}^2.
\end{equation}
This assigns to each point on $\Sigma_\tau \cap\{|r-r_+|<2\delta\}$ a coordinate value on $\mathbb{S}^2$.

\subsection{Causal relations}
\begin{definition}
    Given a set $S$ in the spacetime, we define its chronological future and causal future by
    \begin{equation*}
        I^+(S):=\{q\colon \text{There exists a point $p\in S$ and a future-directed timelike curve from $p$ to $q$}\}.
    \end{equation*}
    \begin{equation*}
        \JJ^+(S):=\{q\colon \text{There exists a point $p\in S$ and a future-directed causal curve from $p$ to $q$}\}.
    \end{equation*}
\end{definition}
One can also define its chronological past $I^-(S)$ and causal past $\JJ^-(S)$ by replacing ``future" with ``past".

By a causal curve we mean a $C^1$ curve with velocity always timelike or null. We have the transitivity properties \cite[Proposition 2.5]{Penrose1972}\footnote{Strictly speaking, the set $I^+(S)$ and $\JJ^+(S)$ is defined differently by piecewise $C^1$ geodesics in \cite{Penrose1972}, but the equivalence was also discussed in the same chapter.}:
\begin{proposition}
    If $a\in \JJ^+(b)$, $b\in \JJ^+(c)$, then $a\in \JJ^+(c)$.
\end{proposition}

The spacetime we study here is globally hyperbolic. In this case, we have the following properties (\cite[Chapter 8.3]{Wald})
\begin{proposition}
    In any globally hyperbolic spacetime, $\JJ^+(K)$ is closed if $K$ is compact. Moreover, we have $\overline{I^+(K)}=\JJ^+(K)$.
\end{proposition}

\begin{definition}\label{def:achronal-set}
    A set is called achronal, if any timelike curve cannot intersect it more than once.\footnote{  A typical example  is     the light cone from a   point in Minkowski space.}
\end{definition}


\subsubsection{Null cones over spheres}\label{subsubsection:null-cones-over-spheres}


 Consider a spacetime $\MM$ foliated by a family of spacelike hypersurfaces $\{\Sigma_\tau\}_{\tau\in \mathbb{R}}$, and each $\Sigma_\tau$ is diffeomorphic to $\mathbb{R}^3\backslash B$, where $B$ is the unit ball in $\mathbb{R}^3$. Any embedded sphere on $\Sigma_\tau$ corresponds to a sphere on $\mathbb{R}^3\backslash B$ through a  diffeomorphism map and  divides \footnote{The Jordan-Brouwer separation theorem applies, for example.}  $\Sigma_\tau\backslash B$ into two regions, the interior, and the exterior (the latter containing the spatial infinity).  This then defines the outgoing and incoming null direction perpendicular to $S$.

\begin{definition}
\lab{Def:C(S)} The  future  outgoing (incoming) null cone of a sphere $S$,  denoted by $C^+(S)$ ($C^{-}(S)$),  is the  null hypersurface generated by the  congruence of future  outgoing (incoming)   null geodesics  perpendicular to $S$.  We can define the past  null  cones  in the same way. 
\end{definition}
\begin{definition}
\lab{Def:RegularNC}
A (portion of) null cone $C(S)=C^{\pm}(S)$, generated by an embedded sphere $S$, is said to be regular if it  is smooth and embedded, i.e. it is the image of a smooth embedding from $[0,T]\times \mathbb{S}^2$ to the spacetime $\MM$.   
 
\end{definition}

The following   definition is   important in determining the regularity of $C(S)$, see Remark \ref{Remark:immersioncriterion}.
\begin{definition}
\lab{Def:trch-C(S)}
 Let  $L$  be  a  null geodesic vectorfield   of  $C(S)$, i.e.      tangent to the  null   generators of $C(S)$  and $D_L L=0$, and 
      define  the null expansion along each null generator   to be     $\trch=\delta^{ab}  g(D_{E_a}L, E_b) $  where $E_a$  is an arbitrary   choice  of spacelike orthonormal frame  perpendicular to  $L$.   Note that $\trch$  is invariant, i.e.,  it does not dependent  on the choice of $E_a$\footnote{Note that when $L$ is fixed,   any two null frames with horizontal structure perpendicular to $L$  can be related by a null frame transformation $(f=0,\fb,\la=1)$ (with $L$ viewed as $e_4$). In view of the transformation formulas, see Section \ref{subsect:NullFrameTransformation},   we  easily see that $\trch$ is invariant.}. 
\end{definition}
We now give some causal properties when the future outgoing  null cone $C^+(S)$ is regular. Clearly similar properties also hold true for a regular past incoming cone of $S$ (i.e. $C^+(S)$ when extended towards the past).
\begin{definition}
Assume that  $C=C^+(S)$   is regular, so $S_\tau=C\cap \Si_{\tau}$ is an embedded sphere on $\Sigma_\tau$. In this case, $S_\tau$ defines the interior and exterior region of $\Sigma_\tau$,  denoted by $\Sigma_\tau^i$ and $\Sigma_\tau^e$. Unless specifically mentioned  otherwise, we  assume  that they     include  the boundaries.
\end{definition}

\begin{proposition}
    Assume $C=C^+(S)$ is a regular null cone with $\Sigma_\tau^i$, $\Sigma_\tau^e$ defined as above. Then there is no future-directed timelike curve\footnote{By convention, a single point is not a timelike curve; it is however often counted as a causal curve.} emanating from $\Sigma_{\tau}^i$ that intersects  $C$.
\end{proposition}
\begin{proof}
  Suppose there is such a curve $\ga$ initiating on $\Si_\tau$, and denote  the first intersecting point  $p\in S_{\tau'}\subset C$ for some $\tau'> \tau$.   Take a null frame $\{L,\Lb,E_1,E_2\}$ so that $E_1$, $E_2$ are tangent to $S_{\tau'}$ (hence $\Sigma_{\tau'}$), $L$ is tangent to the null generators of $C$, and both $L$, $\Lb$ are pointing to the future side of $\Sigma_{\tau'}$. Then the velocity vector $V$  along $\ga$ can be written as $V=aL+b\Lb+V_s$, where $V_s\in \mathrm{span}\{E_1,E_2\}$ is spacelike. Since $p$ is the first intersecting point and $V$ is future  directed, we must have $b\leq 0$. Note that
    \begin{equation*}
        g(V,V)=-4ab+g(V_s,V_s)\geq -4ab.
    \end{equation*}
    Since $V$ is timelike,    we must have    $ab>0$. Therefore  $a,b<0$, which  contradicts the assumption that $V$ is pointing to the future of $\Sigma_{\tau'}$.
\end{proof}
Similarly, one can show that the past incoming null cone, if regular, does not intersect  the chronological past of $\Sigma_{\tau}^e$. 
\begin{definition}\label{def:C-Sigma-achronal}
    For any $\tau_1,\tau_2$ with $\tau_1<\tau_2$, denote the portion of $C=C^+(S)$ between $\Sigma_{\tau_1}$ and $\Sigma_{\tau_2}$ by $C({\tau_1},{\tau_2})$. Assume that $C(\tau_1,\tau_2)$ is regular, then from above, the set $\Sigma_{\tau_1}^i\cup C({\tau_1},{\tau_2})\cup \Sigma_{\tau_2}^e$ is an achronal set, i.e., no timelike curve can intersect with it more than once. We denote it by $\widehat C(\tau_1,\tau_2)$.
\end{definition}

Recall that in a globally hyperbolic spacetime, the closure of $I^+(S)$ is $\JJ^+(S)$. Therefore we have
\begin{corollary}\label{Cor:int-ext}
    Under the assumption of Definition \ref{def:C-Sigma-achronal}, we have $\JJ^+(\Sigma_{\tau_1}^i)\cap \Sigma_{\tau_2}^e=S_{\tau_2}$, $\JJ^-(\Sigma_{\tau_2}^e)\cap \Sigma_{\tau_1}^i=S_{\tau_1}$.
\end{corollary}


\section{Strategy of the  proof of  Theorem \ref{thm:regularity}}\lab{Section:main}

The goal  of the section  is to  describe the main ideas of the proof of  Theorem \ref{thm:regularity},  restated below:
\begin{theorem}
For the spacetime $\MM$, described in Section \ref{Sect:preliminaries}, satisfying the smallness condition \eqref{eq:estim.Gac}, the     event horizon $\hp$, defined as the boundary of the region $\mathcal J^-(\ip)$, is a regular null hypersurface in the future of the spacelike hypersurface $\Sigma_0$\footnote{In the case of the stability result of \cite{KS}  this 
would be the   initial Cauchy hypersurface,  which  lies in the initial data layer.}.
\end{theorem}

\begin{remark}
    In the proof, we only need to assume the smallness estimate \eqref{eq:estim.Gac} in $\RR_\delta=\{|r-r_+|\leq \delta\}$ and the characterization \eqref{eq:pastofscri} of the past of null infinity.
\end{remark}


\subsection{Regular  null cones over spheres}\label{subsection:Causalstructure-nullcone}

We start with  the following definition.

\begin{definition}    An immersion $i\colon X\to Y$ is a smooth  map  whose  tangent map is  injective  at all points  of $X$. An embedding is an injective immersion where the map is a homeomorphism onto its image, with the image $i(X)$ viewed in the subspace topology of $Y$.
\end{definition}
In the case when $X$ is compact,   an  injective  immersion map is  an  embedding.

We now assume that $S$ is an embedded sphere on $\Sigma_{0}$ in a spacelike foliation $\{\Sigma_\tau\}$ of $\MM$. Consider the outgoing (or incoming) cone of $S$, denoted by $C(S)$. Denote the null generator emanating from $p\in S$ by $\ga_p$, and its unique intersection with $\Sigma_\tau$ by $\ga_{p,\tau}$. Define the following map:
\bea
i_\tau\colon S\to \Sigma_\tau,\quad p\mapsto \ga_{p,\tau}.
\eea
The regularity of the null cone $C(S)$,  see Definition \ref{Def:C(S)},   is crucially dependent on whether $i_\tau$ is an embedding into $\Sigma_\tau$. Since $S$ is diffeomorphic to $\mathbb{S}^2$, we  can, by abuse of language,  write this map as $i_\tau\colon \mathbb{S}^2\to \Si_\tau$.

\begin{remark}
\lab{Remark:immersioncriterion}
A null cone  $C(S)$ generated by a smooth embedded sphere $S$,  in a smooth spacetime $\MM$,   is   regular,  if and only if:
\begin{itemize}
\item      Nearby null geodesic generators do not intersect.   This is equivalent to saying that $i_\tau$ is an immersion for each $\tau$.
\item No null geodesic  generators   of $C(S)$ intersect. This is   equivalent to saying that $i_\tau$ is in fact an embedding for each $\tau$.
\end{itemize} 
\end{remark}

\subsubsection{Immersion Criterion}
 
Given a smooth choice of the  vector field $L$ on $S$, one can extend it  as a geodesic vector along the null generators of $C(S)$. Then the null expansion $\trch$ is well-defined, see Definition \ref{Def:trch-C(S)}.

\begin{proposition}[Immersion Criterion]\label{Prop:Immersion-Criterion}
 The map $i_\tau$ is an immersion for all $\tau\in [0,T]$ if and only if   the null expansion  $\trch$  is   bounded away from $-\infty$ for $\tau\in [0,T]$. This also means that the map $[0,T]\times S\to \MM$, $(\tau,p)\mapsto \gamma_{p,\tau}$ is an immersion.
\end{proposition}
\begin{proof}
See Appendix \ref{Appendix:Null-hypersurfaces}.
\end{proof}


\subsubsection{Embedding Criterion}
\begin{lemma}\label{lem:local-in-time-regular}
 Suppose $i_{\tau_0} :\SSS^2 \to  \Si_{\tau_0} $ is an embedding for some $\tau_0$.   Then there exists   a sufficiently small  $\ep$  such that
  the maps  $i_{\tau} :\SSS^2 \to      \Si_{\tau} $ remain embeddings for all $\tau$, $|\tau-\tau_0|\le \ep$.
\end{lemma} 
\begin{proof}
It is clear that we only need to prove it towards the future, as the other side would be the same. Also, since $i_{\tau_1}\circ i_{\tau_2}=i_{\tau_1+\tau_2}$, it suffices to prove the statement  for $\tau_0=0$. We thus  assume $\tau_0=0$  and denote 
  by $S$ the image  of   $i_{0}$, i.e.  $i_{0}(\mathbb{S}^2)=S$. If the lemma does not hold true,  then there exists a sequence $\tau_n\searrow 0$, and there exist  $p_1^n$, $p_2^n\in \SSS^2$, distinct for each $n$, satisfying $i_{\tau_n}(p_1^n)=i_{\tau_n}(p_2^n)$. By compactness,  we can assume\footnote{Or just extract a subsequence.}   that $p_1^n\to p_1$, $p_2^n\to p_2$. Since $S$ is compact, $i_{\tau_n}(p)\to i_{0}(p)$ uniformly\footnote{Evaluated e.g.  with  respect  to the  Riemannian metric  associated  to  $g$ under the spacelike foliation $\Sigma_\tau$.} in $p\in \mathbb{S}^2$ as $n\to\infty$, and therefore, $i_{\tau_n}(p_1^n)\to i_{0}(p_1)$, $i_{\tau_n}(p_2^n)\to i_{0}(p_2)$, so $i_{0}(p_1)=i_{0}(p_2)$, i.e. $p_1=p_2$. Since $S$ is embedded, the null expansion is clearly bounded near $S$, so by Proposition \ref{Prop:Immersion-Criterion}, we know that for each $p$, there exist a $\ep_p>0$ and a neighborhood $\OO_p$ on $\mathbb{S}^2$ such that $[0,\ep_p)\times \OO_p\to \MM$, $(\tau,p')\mapsto i_\tau(p')$ is injective. Taking $n$ large such that $i_{\tau_n}(p_1^n)$, $i_{\tau_n}(p_2^n)$ lie in this injective neighborhood for $p_1=p_2$ gives a contradiction.
\end{proof}
Clearly  for $\tau $ far away  from $\tau_0$,  $i_\tau$ may fail to be an embedding, even  if it    remains an immersion.  See Figure \ref{Figure:example-non-embedding} for an illustration of the situation.

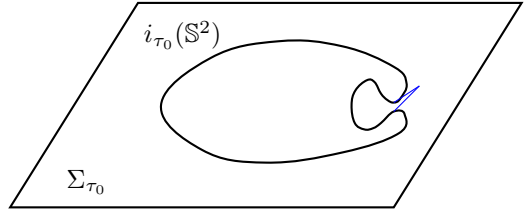
\begin{figure}[ht]
\centering
\begin{tikzpicture}[scale=1.7]
    \draw[thick] (-2,-0.8) -- (-1,0.8) -- (2,0.8) -- (1,-0.8) -- cycle;

    \draw[thick,smooth cycle,tension=0.7] plot coordinates {
        (-0.82,0) 
        (-0.43,0.37)
        (0.1,0.5) 
        (0.6,0.46)
        (1,0.3)
        (1.1,0.17)
        (1.05,0.05)
        (1,0.02)
        (0.95,0.04)
        (0.9,0.09)
        (0.85,0.18)
        (0.79,0.195)
        (0.7,0.115)
        (0.67,0) 
        (0.695,-0.14)
        (0.77,-0.19)
        (0.84,-0.18)
        (0.9,-0.14)
        (1,-0.05)
        (1.084,-0.07)
        (1.04,-0.22)
        (0.6,-0.37)
        (0,-0.45)
        (-0.5,-0.35)
    };
    \draw[blue] (1,-0.05) -- (1.2,0.15);
    \draw[blue] (1,0.02) -- (1.2,0.15);
    \node at (-1.4,-0.6) {$\Si_{\tau_0}$}; 
    \node at (-0.65,0.55) {$i_{\tau_0}(\mathbb{S}^2)$};
\end{tikzpicture}
\captionsetup{width=.7\linewidth, justification=centering}
\caption{A picture in $2+1$ dimension illustrating the case where there will be $\tau>\tau_0$ such that $i_\tau$ is an immersion but not an embedding. The short blue lines represent intersecting null generators.}\label{Figure:example-non-embedding}
\end{figure}
The following gives a useful global criterion. In the following, the spacelike foliation $\{\Sigma_\tau\}$ is the one defined in Section \ref{subsect:timefunction}, and recall $r$ is a coordinate function on $\Si_\tau$.

\vspace{2ex}
\begin{proposition}[Embedding Criterion]  
\lab{lem:im-em1}
    Fix any $\tau\geq 0$. Let           $i:\mathbb{S}^2\to \Sigma_\tau$     be   an immersion  that verifies  the following properties:
    \begin{enumerate}
   \item  For each point $p\in\mathbb{S}^2$, there is a neighborhood $\OO_p$ such that\footnote{ That is  $i$ is a local  embedding. }  $i|_{\OO_p}$ is an embedding to $\Sigma_\tau$.

   \item 
     There exists a small number $\mathring \delta$  such that  every  $\OO_p$   admits  an orthonormal frame $\{\ezero_a\}$,  tangent to the embedded submanifold $i(\OO_p)$,  for which  $|\ezero_a(r)|\leq \mathring \delta$.
         \end{enumerate}
      Then the immersion $i:\mathbb{S}^2\mapsto \Sigma_\tau$ is an embedding. Moreover, there exists a smooth map
    \begin{equation*}
        R\colon \mathbb{S}^2\to [r_+-\delta,r_+ +\delta], \quad p\mapsto R(p)
    \end{equation*}
    such that $\Phi_{\Sigma}(\{(R(p),p)\colon p\in \mathbb{S}^2\})=i(\mathbb{S}^2)$ (with  $\Phi_{\Sigma}=\Phi_{\Sigma_\tau}$ defined in \eqref{eq:def-diffeomorphism-Sigma0}).
\end{proposition}
\begin{proof}
    See Appendix \ref{subsubsect:proof-im-em1}.
\end{proof}

\begin{corollary}
For any $T>0$, if $i_\tau$ are embeddings for all $\tau\in [0,T]$, then the null cone $C(S)$ is regular between $\Si_0$ and $\Si_T$.
\end{corollary}
\begin{proof}
    Note that by causal structure, each null generator intersects a spacelike hypersurface only once. Therefore the injectivity  w.r.t.  the time variable is verified and the regularity of the cone follows.
\end{proof}

\subsection{Main steps  in  the proof of  the regularity theorem  \ref{thm:regularity}}
\lab{section:strategy}

\begin{definition}\label{def:past-of-scri}
    We say that a point $p\in \mathcal M$ is in the causal past of the future null infinity $\ip$, denoted by $\mathcal J^-(\ip)$, if the causal future of $p$, $\JJ^+(p)$, satisfies
    \begin{equation*}
        \sup_{\JJ^+(p)} r=+\infty.
    \end{equation*}
\end{definition}

\begin{definition}
    The event horizon of $\mathcal M$, denoted by $\HH^+$, is defined as the boundary of $\mathcal J^-(\ip)$.
\end{definition}


We now give a characterization of $\mathcal J^{-1}(\ip)$. 
\begin{proposition}[Characterization of $\mathcal J^-(\ip)$]
    We have 
\begin{equation}\label{eq:pastofscri}
    \mathcal J^{-}(\ip)=\{p: \mathcal J^+(p)\cap \{r\geq r_++\delta\}\neq \varnothing\}.
\end{equation}
In particular,  any point  in the set $  \{r\geq r_++\delta\}$  belongs to   $ J^{-}(\ip)$.
\end{proposition}
\begin{proof}
Clearly, the left is included in the right. 
For the converse, if $p$ is in the set on the right hand side, then there is a point $p'\in \mathcal J^+(p)$ with $r$-value of $p'$ no less than $r_++\delta$. Consider the integral curve of $e_4$ starting from $p'$. This is a causal curve, and we know that, for all  $r\geq r_++\delta$,
\beaa
e_4(r)=\frac{\Delta}{|q|^2}+\widecheck{e_4(r)} >c\delta
\eeaa
 for some constant $c>0$. This implies that along this curve one can reach arbitrarily large values of $r$. Therefore, $\sup_{\JJ^+(p')}r=+\infty$, and the result follows from the transitivity of causal relations.
\end{proof}

To prove Theorem  \ref{thm:regularity},   we consider a spacelike hypersurface  $\Si_0$     given as  the level set  $\{\tau=0\}$, where      $\tau$  is the time function defined in Section \ref{subsect:timefunction}, the outgoing (past incoming) null cones $C_v$  generated  backward from a sphere $S(v,r^*_\delta)\subset \{r=r_\delta^*:=r_++\delta\}$,   and   the intersections $S_v^0=\Si_0\cap C_v$.
To prove the theorem we need to implement the following steps (see Figure \ref{Figure:Strategy-proof}  for a graphic representation of the proof):

\begin{figure}
\begin{tikzpicture}
\useasboundingbox (-8,-0.5) rectangle (4,4);
\def\size{5}
    \def\height{5}

    \draw[dashed] (-\size, 0) -- (0,\height) node[midway,above] {${\cal H}^+$};
    \draw[dashed]  (0,\height) -- (\size, 0) node[midway,above=0.2cm] {${\cal I}^+$};

    \draw[thick] (0,\height) -- (-\size,\height/2) node[midway,left=0.2cm] {$r=r_+-\delta$};
    \draw[thick] (0,\height) -- (-\size/3,0) node[midway,right] {$r=r^*_\delta:=r_++\delta$} node[pos=0.6, inner sep=1pt](point){}
    node[pos=0.4, inner sep=1pt](point2){}
    node[pos=0.2,inner sep=1pt](point3){};
    \filldraw [blue] (point3) circle (1.5pt) node[anchor=west] {$S(v,r^*_\delta)$};
    \filldraw [blue] (point2) circle (1.5pt) node[anchor=south] {};
    \filldraw [blue] (point) circle (1.5pt) node[anchor=south] {};

    \node[draw, circle, fill=white, inner sep=1pt, label=above:{$i^+$}] at (0,\height) {};
    \node[draw, circle, fill=white, inner sep=1pt, label=right:{$i^0$}] at (\size,0) {};

    \draw[thick, name path=curve] (\size/2,0) node[right] {$\Sigma_0$} .. controls (0,0) and (-\size/2,
    \size/8) .. (-\size,\height/2) ;

    \draw[thick] (2.5,1) node[right] {$\Sigma_\tau$} .. controls (0.8,1.03) and (-\size/2+\size/5,
    \size/8+\size/5) .. (-4,3) ;
    \node at (3,0.68) {($\tau=v-r$)};

    \path[name path=tangent] (point) -- +(-2,-2); 
    \path[name path=horizon] (0,\height) -- +(-7,-7);

    \path[name intersections={of=tangent and curve, by=E}];
    \path[name intersections={of=horizon and curve, by=Estar}];

    \draw[blue] (point) -- (E);

    \filldraw [red] (E) circle (1.5pt) node[anchor=south] {};
    \filldraw [] (Estar) circle (1.5pt) node[above=1mm] {$S^*$};

    \path[name path=tangent2] (point2) -- +(-5,-5); 

    \path[name intersections={of=tangent2 and curve, by=E2}];
    
    \draw[blue] (point2) -- (E2);

    \filldraw [red] (E2) circle (1.5pt) node[anchor=south] {$S^0_v$};

    \path[name path=tangent3] (point3) -- +(-8,-8); 
    
    \path[name intersections={of=tangent3 and curve, by=E3}];
    
    \draw[blue] (point3) -- (E3);

    \filldraw [red] (E3) circle (1.5pt) node[anchor=south] {};
\end{tikzpicture}
\caption{Idea of the proof of Theorem \ref{thm:regularity}}\label{Figure:Strategy-proof}
\end{figure}
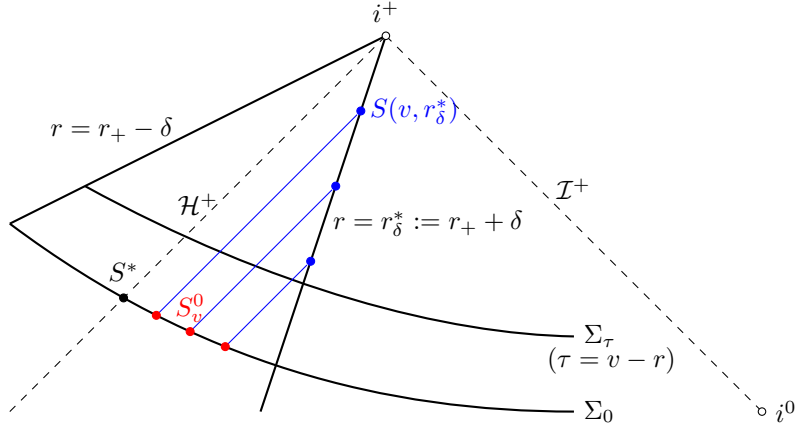

\begin{enumerate}

\item[\bf{Step 1.}]  \textit{Monotonicity.}    Assuming that the hypersurfaces $C_v$ are regular, we  show  that for   each  $v_1< v_2$,  the sphere 
 $S_{v_1}^0= C_{v_1}\cap \Si_0$  is included in the  exterior of the sphere  $S^0_{v_2}=C_{v_2} \cap \Si_0$  on $\Sigma_0$.

\item[{\bf Step 2.}]     \textit{Regularity of $C_v$.}      This,  the heart of the proof,  is given in  Section  \ref{subsection:C_v-frame}.   The  argument   depends in an essential way 
 on   the bounds  \eqref{eq:estim.Gac},  Lemma \ref{le:nab-omega-near-extremal} and the transformation formulas of Section   \ref{subsect:NullFrameTransformation}.
The proof  shows moreover  that for each $\tau$, $S_v^\tau$ (the $\Sigma_\tau$-section of $C_v$)  is an embedded sphere that can be written as $r=R_v(\theta,\varphi)$ (switched to $r=R_v(x^1,x^2)$ when needed) on $\Sigma_\tau$. This  fact  relies  heavily  on the Embedding Criterion (Proposition \ref{lem:im-em1}). Note  that the  main  assumption 2.  of the Embedding Criterion is verified in view of Lemma \ref{lem:ezero-on-r}.

\item[{\bf Step 3.}]     \textit{Regularity of  the spheres  $S_v^0=C_v\cap\Si_0$.}  Using the estimates in Step 2 we derive,  in Section \ref{subsubsection:Higher-order-sphere}, geometric   estimates for the spheres    $S_v^0=C_v\cap\Si_0$, in particular uniform bounds on the derivatives of the extrinsic curvuture of $S_v^0$ on $\Sigma_0$.

\item[{\bf Step 4.}]      \textit{Limiting sphere $S^*$.}    Using  Step 1 and  the  trivial lower bound  for  the functions $R_v$,  $|R_v|\geq r_+-\de$, it is easy to show   that the induced surfaces  $S_v^0$ on $\Si_0$ have a   limit  $S^*$ as $v\to \infty$.   In  Section  \ref{section:Convergence of S^0_v}  we  rely on  the bounds obtained in Step 3 to  show   that  $S^*$ is   a smooth sphere.

\item[{\bf Step 5.}]     \textit{Future Horizon.}       In Section \ref{section:EvenHorizon}  we   show  that the future outgoing null  cone generated from $ S^*$ coincides  with the  future horizon   $\HH^+\cap \JJ^+(\Si_0)$.
\end{enumerate}

\subsection{Monotonicity property}
We first prove a monotonicity property  under   the assumption that each $C_v$ is regular. 
For $v_2>v_1$, each point on $S(v_1,r^*_\delta)$ can be connected by a point on $S(v_2,r^*_\delta)$ through  the integral curve of $\pa_v$, which is always timelike on $\{r=r^*_\de\}$ (recall that $  r^*_\de=r_+ +\delta$), i.e.,   $S(v_1,r^*_\delta)\subset I^-(S(v_2,r^*_\delta))$. Therefore, the past of $S(v_1,r_\delta^*)$ is contained in the chronological past of the achronal set $\widehat{C}_{v_2}(0,v_2-r_\delta^*)$ (see Definition \ref{def:achronal-set} and \ref{def:C-Sigma-achronal}). Therefore, we have $\JJ^-(S(v_1,r_\delta^*))\cap \widehat{C}_{v_2}(0,v_2-r_\delta^*)=\varnothing$. In particular, $S_{v_1}^0\subset \JJ^-(S(v_1,r_\delta^*))$ does not contain any point in the interior of $S_{v_2}^0$ on $\Si_0$, hence must be in the exterior of $S^0_{v_2}$ on $\Sigma_0$.


\section{ Proof of Theorem \ref{thm:regularity}}
\lab{section:RegularityC_vS_v}

      In this section,   we  follow Steps 2-5     and prove   Theorem   \ref{thm:regularity}.

\subsection{Geometry of the null cones $C_v$}
\lab{subsection:C_v-frame}

We now study the geometry of  $C_v$.  We only care about the part of $C_v$ that is in the past of $S(v,r^*_\delta)$, so  in what  follows     $C_v$ only refers to this part.
\subsubsection{Null generators}
Recall that $C_v$ is spanned by the  null geodesics emanating from points on $S(v,r^*_\delta)$  orthogonal to the sphere at these points, called the   null generators of $C_v$.  To determine the tangent vector field $e_4'$  corresponding to these  null generators\footnote{For simplicity, we do not label the vector with $v$, but clearly this is dependent on $v$.}  we consider a frame  transformation  of the form
\bea
    \lab{eq:e-e'}
    e'_4=\lambda \left(e_4+f^b e_b+\frac 14|f|^2  e_3\right), \quad e'_a=e_a+\frac 12 f_a e_3, \quad e_3'=\lambda^{-1} e_3, 
\eea
and impose the condition (Recall that $\eS_4$ is orthogonal to the spheres $S(v,r)$)
\bea
\lab{eq:prime-frame-r=r+delta}
    \om'=0,\quad \xi'=0,\quad e_4'|_{S(v,r^*_\delta)}=\eS_4|_{S(v,r^*_\delta)},
\eea
which ensure  that $e_4'$ is geodesic.  The third condition in \eqref{eq:prime-frame-r=r+delta}  can be rewritten as
\begin{equation*}
    f|_{S(v,r^*_\delta)}=F|_{S(v,r^*_\delta)},\quad \lambda|_{S(v,r^*_\delta)}=1,
\end{equation*}
where $F$ is defined in \eqref{eq:FFb-perturb}.
Recall that $|F|\lesssim \min\{\delta,a\}$.
\begin{remark}
Note that  the parameters $(f, \la)$ also depend on $v$, so they should in fact be denoted by $(f_v, \la_v)$. We omit however the dependence on $v$ as long as there is no danger of confusion.
\end{remark}

In view of the transformation formulas \eqref{eq:trasportation.formulas},       we have
\begin{equation*}
0=    \lambda^{-2} \xi'=\xi+\frac 12 \la^{-1}\nab_{4}' f+\frac 14\trch f+\om f+\frac 12 f\cdot\chih-\frac 14 \atrch \dual f+O(|f|^2),
\end{equation*}
\bea
\lab{eq:transf-om}
    0=\lambda^{-1}\om'=\om-\frac 12 \lambda^{-1} e_4'(\log\lambda)+\frac 12 f\cdot (\zeta-\etab)+O(|f|^2).
\eea
The first equation can be  written  as
\bea
\label{eq:nab_4'f}
    \nab'_{e_4+f^b e_b+\frac 14|f|^2 e_3}f+\frac 12\trch f+2\om f=-2\xi-\chih\cdot f+\frac 12 \atrch \dual f+O(|f|^2),\quad f|_{S(v,r^*_\delta)}= F,
\eea
and it is manifestly independent of $\la$ as the horizontal structure $\{e_a'\}$ is also independent of $\la$.  

 Let $s$ denote a time parameter  along the null geodesics spanning $C_v$ such that  
\bea
\lab{eq:def-s}
\la^{-1}e_4'(s)=\Big(e_4+f^b e_b+\frac 14 |f|^2 e_3\Big)(s)=1, \qquad s|_{S(v,r^*_\delta)}=v.
\eea
\begin{proposition}
\lab{Prop:Boundsfor-f}
      For each null generator $\gamma$ of $C_v$,  assume  that  $\ga$    belongs  to $\RR_\delta=\{r\colon |r-r_+|\leq \delta\},$ (see \eqref{def:RR_de}) for $s_*\leq s\leq v$. Then, on this portion of $\ga$, 
    \bea\label{eq:zero-order-bound-f}
    |f|\lesssim \delta, \qquad e^{C_1 \delta_*^\frac 12(v-s)} \leq \lambda\leq  e^{C_2 (v-s)}. 
    \eea
    \end{proposition}
\begin{proof}
On $\RR_\delta$ we have   $ |\trch|\leq C\delta$ and,        in view of  Lemma \ref{le:nab-omega-near-extremal},  $ \om \leq -{C\delta_*^\frac 12}$.
    Therefore,  for $\de\leq \de_*$ sufficiently small,
\beaa
\phi:=4\om+\trch<-C \delta_*^\frac 12.
\eeaa
In view of \eqref{eq:nab_4'f}, along each null geodesic, we have
\begin{equation*}
    \frac{d}{ds}\left(|f|^2\right)+\trch |f|^2+4\om |f|^2=-4\xi\cdot f+O(|\chih,\atrch||f|^2)+O(|f|^3).
\end{equation*}
Therefore,
\begin{equation*}
    \frac d{ds}\left(e^{\int_{s_*}^s \phi ds'} |f|^2\right)=-4e^{\int_{s_*}^s \phi ds'} \xi\cdot f+e^{\int_{s_*}^s \phi ds'}\left(O(|\chih,\atrch||f|^2)+O(|f|^3)\right).
\end{equation*}
We now make  the bootstrap assumption 
\bea\label{eq:bootstrap-zero-order}
|f|\leq C_b\delta, \quad \mbox{for } s\in [s_1,v]
\eea
 which holds when $s=v$ with $C_b$ replaced by some $C_0$, using the bound of $F$. Now, for any $s_0\in [s_1,v]$,   we integrate the equation over $s\in [s_0,v]$ to get
\begin{equation*}
    e^{\int_{s_*}^{v} \phi\, ds'}|f|^2\Big|_{s=v}-e^{\int_{s_*}^{s_0} \phi\, ds'}|f|^2\Big|_{s=s_0}=-\int_{s_0}^{v} 4e^{\int_{s_*}^{s'} \phi ds''} \xi\cdot f ds'+\int_{s_0}^v e^{\int_{s_*}^{s'} \phi ds''}\Big(O(|\chih,\atrch||f|^2)+O(|f|^3)\Big) ds'.
\end{equation*}
We now make use of the bound $\phi\leq -C\de^{1/2}_*$,  $|\xi|,|\chih|,|\atrch|\leq C\epsilon_0$ in $\RR_\delta$, and the bound  for $f=F$ on $S(v, r^*_\delta)$ 
    provided by   Proposition \ref{Prop:boundsF-Fb} to deduce, in   $\RR_\delta$,
\begin{equation*}
\begin{split}
    |f|^2\Big|_{s=s_0}&=e^{\int_{s_0}^{v} \phi ds}|f|^2\Big|_{s=v}+\int_{s_0}^{v} 4e^{\int_{s_0}^{s'} \phi ds''} \xi\cdot f ds+\int_{s_0}^v e^{\int_{s_0}^{s'} \phi ds''}\Big(O(|\chih,\atrch||f|^2)+O(|f|^3)\Big) ds'\\
    &\leq C_0\delta^2+
    \int_{s_0}^{v} e^{-C\delta_*^\frac 12(s'-s_0)} \Big(C\epsilon_0 |f|+C\epsilon_0 |f|^2+|f|^3\Big)\, ds'.
    \end{split}
\end{equation*}
Using the bootstrap bound,
  we deduce
\begin{equation*}
\begin{split}
    |f|^2\Big|_{s=s_0}&\leq C_0\delta^2+C_b\int_{s_0}^v e^{-C\delta_*^\frac 12(s'-s_0)} (\epsilon_0+\delta^2)|f|\, ds' \leq C_0 \delta^2+C_b\int_0^\infty e^{-C\delta_*^\frac 12 s}\Big(\delta^\frac 32\cdot \delta^\frac 12\sup_{s\in [s_1,v]} |f|\Big) ds\\
    &\leq C_0\delta^2+\frac{2C_b}{C\delta_*^\frac 12} \Big(\delta^3+\delta\sup_{s\in [s_1,v]}|f|^2\Big),
    \end{split}
\end{equation*}
so as long as $\delta\leq  \delta_*$, taking the supremum over $s_0\in [s_1,v]$ we obtain the bound
\begin{equation*}
    \sup_{s\in[s_1,v]}|f|^2\leq 2C_0 \delta^2,
\end{equation*}
which improves the bootstrap assumption \eqref{eq:bootstrap-zero-order} once we pick up appropriate $C_b$. Therefore $|f|\lesssim \delta$ for all $s\in [s_*,v]$. 

Next we  estimate $\lambda$. By \eqref{eq:transf-om} we have
\begin{equation*}
    \frac{d}{ds}(\log\lambda)=2\om+f\cdot (\zeta-\etab)+O(|f|^2).
\end{equation*}
Using the  above  estimate  for $f$, the estimates for $\etab, \ze$,   and the upper bound for $\om$,  we have  for the right-hand side
\begin{equation*}
    -C_2 \leq 2\om+f\cdot (\zeta-\etab)+O(|f|^2) \leq -C_1 \delta_*^\frac 12,\quad \text{when }|r-r_+|\leq \delta.
\end{equation*}
Therefore, along each null generator, when $s_*\leq s\leq v$, we have
\bea
\label{eq:bound-la}
    e^{C_1 \delta_*^\frac 12(v-s)} \leq \lambda\leq  e^{C_2 (v-s)}
\eea
as stated.
\end{proof}


\subsubsection{Regularity of the null cones $C_v$}
\lab{subsection:HigherOrder}

We now start to prove that $C_v$ are regular null cones in $\MM$. 
According to  the transformation formula 
\eqref{eq:trasportation.formulas} for $\trch$, we have
\bea\label{transformation-trch'}
    \lambda^{-1}\trch'=\trch+\div'f+f\cdot(\eta+\zeta)+O(|f|^2).
\eea
We already have estimates for  $\la$ and $f$, provided by Proposition \ref{Prop:Boundsfor-f}, but we  still  need to  estimate $\div' f$.
To do this  we  need to commute the equation \eqref{eq:nab_4'f}   for  $f$ with $\nab'$, using the following standard commutation lemma, see Section 2.2.7 in \cite{GKS}.

\begin{lemma}[Commutation formula]
Given a null frame and a horizontal covariant tensor $\psi_A=\psi_{a_1\cdots a_k}$, we have
\begin{equation}\label{eq:commutation4a}
    \begin{split}
        [\nab_4,\nab_b]\psi_A&=-\chi_{bc}\nab_c \psi_A+(\etab_b+\zeta_b)\nab_4 \psi_A+\sum_{i=1}^k \left(\chi_{b a_i}\etab_c-\chi_{bc}\etab_{a_i}\right)\psi_{a_1\cdots c\cdots a_k}\\&\quad +\sum_{i=1}^k\left(\chib_{b a_i}\xi_c-\chib_{bc}\xi_{a_i}+\in_{a_i c}\dual\b_b\right)\psi_{a_1\cdots c \cdots a_k}+\xi_b\nab_3\psi_A.
    \end{split}
\end{equation}
\end{lemma}
We will apply this to the frame $\{e_3',e_4',e_a'\}$.

We now state the main result in this part.
\begin{proposition}
\label{prop:estimate-trch'}
   For every $v$, the null geodesic congruence $C_v$ is a regular null cone in $\MM\cap \{0\leq \tau \leq v-r^*_\delta\}$. Moreover, we have the estimate
   \begin{equation}
       |\nab'^{\leq i} f|\leq C_N\delta,\quad 0\leq i\leq N.
   \end{equation}
\end{proposition}

\begin{proof}
The  strategy of the proof is as follows:

{\bf Step 1.}  Fix a value of $v$.  We  assume  that there exists  a minimal value $\tau_*\in [0,v-r^*_\delta]$ such that
\bea
\lab{eq:bootstrap-trch'}
|\lambda^{-1} \trch'|\leq \delta^\frac 12 ,\qquad  \tau_*\leq \tau\leq v - r_\de^*.
\eea
By the Immersion Criterion (Proposition \ref{Prop:Immersion-Criterion}), this shows that the intersection $S^\tau_v:=C_v \cap \Sigma_\tau$, for   $\tau_*\leq \tau\leq v - r_\de^*$, is the image of an immersion $i_v^\tau\colon \mathbb{S}^2\to \Sigma_\tau$. Thus every point $p\in \mathbb{S}^2$ has a neighborhood $\OO_p\subset \SSS^2$ such that $i^\tau_v|_{\OO_p}$ is an embedding into $\Sigma_\tau$. It then makes sense to talk about a local orthonormal frame w.r.t. the induced metric on $\Sigma_\tau$.
We have the following lemma.

\def\etau{\, ^{(\tau)}e}
\begin{lemma}\label{lem:ezero-on-r}
 Under the assumption \eqref{eq:bootstrap-trch'}, consider a neighborhood $\OO_p$ for which $i_v^\tau|_{\OO_p}$ is an embedding. If\footnote{Recall that we have defined $f$ everywhere along each null generators, see \eqref{eq:e-e'} and \eqref{eq:prime-frame-r=r+delta}.} $|f|\leq C\delta$ on $i_v^\tau(\OO_p)$, then there exists an orthonormal basis $\{\ezero_a\}$ on $i_v^\tau(\OO_p)$ w.r.t the induced metric on $\Sigma_\tau$ satisfying \bea\label{eq:e(r)-delta} |\ezero_a(r)|\leq C' \delta,\eea with $C'$ depedent on $C$ but independent of $v$, $\tau$, $\tau^*$, $p$.
\end{lemma}
\begin{proof}
    See Section \ref{subsection:RegularityS_v^0}.
\end{proof}

The bound \eqref{eq:e(r)-delta} provided by the lemma verifies the condition in the Embedding Criterion (Proposition \ref{lem:im-em1}). Therefore we have the following:
\begin{proposition}
\lab{Proposition:C_v is regular}
    Under the assumption \eqref{eq:bootstrap-trch'}, $C_v$ is a regular null cone in $\MM\cap \{\tau_*\leq \tau\leq v-r^*_\delta\}$.
\end{proposition}
\begin{proof} Recall that $\tau=v-r$  and that   the  cones $C_v$ initiate on  $r=r^*_\de= r_+ +\de$.
According to Proposition \ref{Prop:Immersion-Criterion} and the assumption  \eqref{eq:bootstrap-trch'}  on $\trch'$, it suffices to show that the immersion $i_v^\tau\colon \mathbb{S}^2\to \Sigma_\tau$ is an embedding for any $v$  and $\tau \in [\tau_*,v-r^*_\delta]$. By compactness of $\mathbb{S}^2$, for it not to be an embedding,  $i^\tau_v$ must fail to be injective. Denote $\tau_*'\geq \tau_*$ as the first (backward) value of $\tau$ such that $i^\tau_v$ fails to be injective. This means that the part of $C_v$  strictly above (in the future of) $\Sigma_{\tau_*'}$ is regular, so all null generators are in $\RR_\delta$ \footnote{Indeed,  $\{r= r_+ +\delta\}$ is a timelike  hypersurface,  spanned by timelike curves (the integral curves of $\pa_v$) from $S(v,r^*_\delta)$.  The  regular part of  $C_v$   lies on the boundary of $\mathcal J^-(S(v,r^*_\delta))$ and  cannot intersect  either  $\{r= r_+ +\delta\}$ or 
   $\{r=r_+-\delta\}$.}. By continuity, they remain in $\RR_{\delta}$ on $\Sigma_{\tau_*'}$, so the estimate of $|f|$ in Proposition \ref{Prop:Boundsfor-f} applies, and we have $|f|\lesssim \delta$. Therefore, for $\delta$ small enough, applying Lemma \ref{lem:ezero-on-r} and the Embedding criterion (Proposition \ref{lem:im-em1}), we deduce that the intersection $C_v\cap \Si_{\tau'_*}$  must also be an embedded sphere on $\Sigma_{\tau_*'}$. This leads to  a contradiction as we can now extend the cone     from this embedded sphere on $\Sigma_{\tau_*'}$ regularly,  for  a small neighborhood of $\tau_*'$ in $\tau$, by Lemma \ref{lem:local-in-time-regular}.
\end{proof}
Therefore, we  conclude   that for $\tau_*\leq \tau\leq v-r^*_\delta$, we have $C_v$ regular  and $|f|\lesssim \delta$.

{\bf Step 2.}    We  prove bounds for    derivatives of $f$. We derive a higher order version which we will use later. Recall that from Step 1, we know that $C_v$ is regular for $\tau\in [\tau_*,v-r_\delta^*]$.
\begin{lemma}
\label{Lemma:estimate-trch'}
   Suppose the bootstrap assumption \eqref{eq:bootstrap-trch'} holds. Then, on $C_v$, whenever $\tau\in [\tau_*,v-r^*_\delta]$, we have $|\nab'^{\leq i} f|\leq C_N\delta$ for all $i\leq N$.
\end{lemma}
\begin{proof}
    See Section \ref{subsubsect:proof-higher-order-f}.
\end{proof}

{\bf Step 3.}   We improve  the bootstrap assumption \eqref{eq:bootstrap-trch'}  by  showing  that  in fact 
     $|\lambda^{-1}\trch'|\leq C\delta$ when $\tau\in [\tau_*,v-r^*_\delta]$, independent of $\tau_*$ and $v$.   This implies that $C_v$ remains regular in $\MM$ for all $\tau\in [0,v-r^*_\delta]$.
\begin{proof}[Proof of Step 3]
    In view of  the transformation formula \eqref{transformation-trch'},  the estimate of $|f|$ in \eqref{eq:zero-order-bound-f}, and the estimates of $\eta, \ze$ in \eqref{eq:estim.Gac},
     it suffices to show the bound for $|\div' f |$. This directly follows from the bound of $|\nab'f|$ in Lemma \ref{Lemma:estimate-trch'}.
\end{proof}
This,     modulo  the   proofs  of  the  Lemmas \ref{lem:ezero-on-r} and  \ref{Lemma:estimate-trch'},  ends the proof of Proposition  \ref{Proposition:C_v is regular}.
 \end{proof}
 It remains to prove Lemma \ref{lem:ezero-on-r} and Lemma \ref{Lemma:estimate-trch'}.


\subsection{Proof of Lemma  \ref{lem:ezero-on-r} } 
\label{subsection:RegularityS_v^0}
In this subsection, we prove Lemma \ref{lem:ezero-on-r}. We prove in fact  the following more detailed version below,  in Lemma \ref{Lemma:ezero}. Note that we only make use of things that have been proved before the statement of Proposition \ref{prop:estimate-trch'} and the bootstrap assumption \eqref{eq:bootstrap-trch'}.

Recall that for each fixed $v$ and $\tau\in [\tau^*,v-r_\delta^*]$, $S^\tau_v$ is the image of an immersion $i_v^\tau\colon \mathbb{S}^2\to \Sigma_\tau$, and for each $p\in \mathbb{S}^2$, there is an open neighborhood $\OO_p$ such that $i_v^\tau|_{\OO_p}$ is an embedding. In the following, we drop the dependence on $\tau$ and $v$ and denote $S=S_v^\tau$, $i=i_v^\tau$.

\begin{lemma}
\lab{Lemma:ezero}
    With the notation above, for each $p\in \mathbb{S}^2$, there exists  a frame transformation $(\, ^{(0)}f, \fbzero$, $\lazero)$ from $\{e_3,e_4,e_a\}$ on the embedded submanifold $i(\OO_p)$ such that  \begin{enumerate}
    \item  The new frame $\{\ezero_3,\ezero_4,\ezero_a\}$ is adapted to $i(\OO_p)$, in the sense that $\{\ezero_a\}$ is tangent to $i|_{\OO_p}$.
    \item  $\Nv=\frac 12 \left(\ezero_4-{\ezero}_3\right)$ is the unit normal vector of $i(\OO_p)$ on $\Sigma_\tau$.
  
    \item   The transformation coefficients can be written as 
    \bea\label{eq:expression-f0-fb0-la0}
 \, ^{(0)} f_a=f_a,\qquad \fbzero_a=-\frac{f_a+2e_a(v)}{e_4(\tau) +f^b e_b(v)+\frac 1 4 |f|^2},
 \eea
 \begin{equation*}
    \lazero^2=\frac{1+\frac 12 f\cdot \fbzero+\frac 1{16} |f|^2 |\fbzero|^2+(\fbzero^b+\frac 14|\fbzero|^2 f^b) e_b(v)+\frac 14|\fbzero|^2 e_4(\tau)}{e_4(\tau)+f^b e_b(v)+\frac 14|f|^2},
\end{equation*}
where $f$ is the one defined in Section \ref{Section:main}. Note that one also has
\begin{equation*}
    \ezero_4=\lazero e_4',\quad \ezero_a=e_a'+\frac 12 \fbzero e_4',\quad e_3'=\lazero^{-1} \left(e_3+\fbzero^a e_a'+\frac 14 |\fbzero|^2 e_4'\right),
\end{equation*}
that is, the new frame can be obtained from the $e'$-frame through the transformation $(0,\fbzero,\lazero)$.
\item We have the estimates
\begin{equation*}
    |\fbzero|\leq C\max\{\delta,a\},\quad C^{-1}\leq \lazero\leq C,
\end{equation*}
where $C>0$ is independent of $v$, $\tau$, $\tau^*$, $p$.
   \end{enumerate}
\end{lemma}
\begin{remark}
    As a result, we have
\begin{equation*}
    \ezero_a(r)=e_a(r)+\frac 12\fb_a f^b e_b(r)+\frac 12\fb_a e_4(r)+\left(\frac 12 f_a+\frac 18|f|^2 \fb_a\right)e_3(r)=\frac 12 \fb_a e_4(r)+O(|f|)=O(\delta).
\end{equation*}
This proves Lemma \ref{lem:ezero-on-r}.
\end{remark}

\begin{proof}
We first take $\, ^{(0)}f_a=f_a$, which ensures that $\ezero_4$ is tangent to the null generators.
To have $\{\ezero_a\}$ tangent to $i(\OO_p)\subset \Sigma_\tau$, we impose (using $e_3(\tau)=1$, $e_a(r)=0$)
  \begin{equation*}
    \begin{split}
    0={\ezero}_a(\tau)  &=  e_a(\tau)+\frac 12\fbzero_a f^b e_b(\tau)+\frac 1 2 \fbzero_a e_4(\tau) +\Big(\frac 12 f_a+\frac 18|f|^2 \fbzero_a\Big)e_3(\tau)\\
    &=e_a(v)+\frac 12\fbzero_a f^b e_b(v)+\frac 1 2 \fbzero_a e_4(\tau)+\Big(\frac 12 f_a+\frac 18|f|^2 \fbzero_a\Big).
    \end{split}
\end{equation*}
Hence,
 \bea\label{eq:def-fb}
 \fbzero_a=-\frac{f_a+2e_a(v)}{e_4(\tau) +f^b e_b(v)+\frac 1 4 |f|^2}.
 \eea
  This determines $\fbzero$.  Using the bound for  $f$ provided by Proposition \ref{Prop:Boundsfor-f}  and the bounds  for $\nab v= \widecheck{\nabb v}+a j$   in Section \ref{section:Quantbonds-Mint}  we deduce  that
  \beaa
  |\fbzero|\lesssim \max \{\de,a\}.
  \eeaa

  To determine $\lazero$, we impose the condition $\ezero_4(\tau)= {\ezero}_3(\tau)$. This gives
\begin{equation*}
    \lazero^2=\frac{1+\frac 12 f\cdot \fbzero+\frac 1{16} |f|^2 |\fbzero|^2+(\fbzero^b+\frac 14|\fbzero|^2 f^b) e_b(v)+\frac 14|\fbzero|^2 e_4(\tau)}{e_4(\tau)+f^b e_b(v)+\frac 14|f|^2}.
\end{equation*}
The denominator is away from zero since $e_4(\tau)$ is away from zero. To see the numerator is also away from zero, we notice that, similar to Section \ref{subsect:timefunction},
\begin{equation*}
    e_4(\tau)=\frac{2(r^2+a^2)}{r^2+a^2\cos\theta}+O(\delta)\geq 2+O(\delta),\quad \fbzero_a=-\frac{2e_a(v)}{e_4(\tau)}+O(\delta), \quad |\nab v|^2=\frac{a^2\sin^2\theta}{|q|^2}+O(\delta)\leq 1+O(\delta),
\end{equation*}
so the numerator is of the size
\begin{equation*}
    1+\fbzero^b e_b(v)+\frac 14|\fbzero|^2 e_4(\tau)+O(\delta)=1-\frac{|\nab v|^2}{e_4(\tau)}+O(\delta)\geq \frac 12+O(\delta).
\end{equation*}
Therefore $\lazero$ is well-defined, regular, and bounded from above and below (by a constant dependent on $\delta$ and $a$).

Note that $\ezero_a=e_a'+\frac 12 \fbzero_a e_4'$, so $\ezero_a$ is indeed tangent to $i(\OO_p)$. It is also straightforward to verify that $\Nv=\frac 12\left(\ezero_4-{\ezero}_3\right)$ is the unit normal of $i(\OO_p)$ on $\Sigma_\tau$.    
\end{proof}

\subsection{ Higher  order estimates and  proof of Lemma \ref{Lemma:estimate-trch'} }
\lab{subsubsect:proof-higher-order-f}

\subsubsection{Higher order estimates on $C_v$}
From now, since it is not relevant anymore, we consider  transformations of type   \eqref{eq:e-e'} with  $\lambda=1$, i.e.,
\begin{equation*}
    e_4'=e_4+f^a e_a+\frac 14 |f|^2 e_3,\quad e_a'=e_a+\frac 12 f_a e_3,\quad e_3'=e_3.
\end{equation*}
In particular, this does not change the result of Lemma \ref{Lemma:estimate-trch'} since the horizontal structure $H'$ (and hence $\nab'$) remains unchanged.

We want to estimate higher-order derivatives of $f$. Note that $f$ is only defined on each null generator, hence on the null cone $C_v$, so we should commute the equation of $f$ by $e_a'$ derivatives, which are tangent to the null cone. 

We first derive the bounds of $H'$-horizontal derivatives on spacetime geometric quantities in the original frame.
\begin{proposition}
    Let $\Gamma$ denote all Ricci coefficients defined in the original $e$-frame. We have, for each $\psi\in \Gamma$,
    \begin{equation}\label{eq:nab'nab'4onGamma}
        |(\nab',\nab'_4)^i \psi|\lesssim \left(1+|(\nab',\nab_4')^{\leq i-1}f|\right)|(\nab,\nab_4,f\nab_3)^{\leq i}\psi|.
    \end{equation}
    
\end{proposition}
\begin{proof}
    Recall the formula \eqref{eq:nab'i-on-H-tensor}
    \begin{equation}\label{eq:repeat-nab'kformula-main-text}
        \nab'_{a_1}\cdots\nab_{a_i}'\psi=\nab_{e_{a_1}'}\cdots \nab_{e_{a_i}'}\psi+f\cdot\chib\cdot(\nab_{e_{a}'})^{\leq i-1}\psi+\nab'^{\leq i-1}(f\cdot\chib\cdot\psi).
    \end{equation}
    Note that schematically
    \begin{equation*}
        \nab_{e_{a_1}'}\cdots \nab_{e_{a_i}'}\psi=\l 1,(\nab_{e_a'})^{\leq i-1}f \r\cdot (f\nab_3,\nab)^{\leq i}\psi.
    \end{equation*}
    We first apply \eqref{eq:repeat-nab'kformula-main-text} to $\chib=(\chibh,\trchb,\atrchb)$ to get
    \begin{equation*}
        \nab'_{a_1}\cdots\nab_{a_i}'\chib=\nab_{e_{a_1}'}\cdots \nab_{e_{a_i}'}\chib+f\cdot\chib\cdot(\nab_{e_{a}'})^{\leq i-1}\chib+\nab'^{\leq i-1}(f\cdot\chib\cdot\chib),
    \end{equation*}
    then by induction, we see that
    \begin{equation*}
        |\nab'_{a_1}\cdots\nab_{a_i}'\chib|\lesssim \l 1+O(|\nab'^{\leq i-1}f|)\r|(\nab,f\nab_3,\nab_4)^{\leq i}\chib|\lesssim 1+|\nab'^{\leq i-1}f|.
    \end{equation*}
    Then applying \eqref{eq:nab'i-on-H-tensor} to all other quantitiesin $\Gamma$ we get similar estimates. The case with $\nab_4'$ can be obtained similarly using \eqref{eq:nab'andnab_4'on-H-tensors}.
\end{proof}
We now commute the equation  \eqref{eq:nab_4'f} with $\nab'$ for $i$ times to get 
    \begin{equation*}
    \nab'_4 \nab'^i f+\nab'^i\Big(\big(2\om+\frac {1}2\trch\big) f\Big)=[\nab'_4,\nab'^i]f-2\nab'^i\xi+\nab'^{i}\big(-\chih\cdot f+\frac 12 \atrch \dual f+O(|f|^2)\big)
\end{equation*}
where $O(|f|^2)$ is an expression quadratic in $f$.
\begin{lemma}
\lab{Lemma.commutator-primes}
    We have 
    \begin{equation*}
        [\nab_4',\nab'^i]U=-i\chi\cdot \nab'^i U+\nab'^{\leq i} f\cdot \nab'^{\leq i}U+ \nab'^{\leq i-1}(\Gamma,f)\cdot(\nab'^{\leq i-1}\nab_4' U+\nab'^{\leq i-1}U).
    \end{equation*}
\end{lemma}
\begin{proof}
    Recall the commutation formula \eqref{eq:commutation4a} which can be written schematically as 
    \begin{equation*}
        [\nab_4',\nab']U=-\chi'\cdot \nab'U+(\etab'+\zeta')\cdot \nab_4' U+(\chi',\etab',\b')\cdot U.
    \end{equation*}
    We also have (recall we are now transforming with $\la=1$), schematically,
    \begin{equation*}
        \chi'=\chi+\nab'f+f\cdot(\eta,\zeta,f\cdot\chih),\quad \etab'=\etab+\frac 12 \chib\cdot f,\quad \zeta'=\zeta+\frac 12 \chib\cdot f,\quad \b'=f\cdot (\rho,\dual\rho,\b)
    \end{equation*}
    so we obtain
    \begin{equation*}
        [\nab_4',\nab']U=-\nab'f\cdot\nab'U-\chi\cdot \nab'U+(\Gamma,f)\cdot \nab_4' U+(\Gamma,\nab'f,f)\cdot U.
    \end{equation*}
    Then, applying this formula recursively,  we derive  the desired  formula in the lemma.
\end{proof}
Applying Lemma \ref{Lemma.commutator-primes} to  $U=f$\footnote{Here $f$ is understood as an $H'$-horizontal tensor as explained in Definition \ref{def:psit}.}, and replacing $\nab_4'f$ by the right hand side of \eqref{eq:nab_4'f}, which is schematically $\xi+f\cdot\Gamma$, we obtain
\begin{equation}\label{eq:nab'4nab'if}
\begin{split}
    \nab'_4 \nab'^i f+\left(2\om+\frac {i+1}2\trch\right)\nab'^i f&=\nab'^{\leq i}(\chih,\atrch)\nab'^{\leq i} f+\nab'^{\leq i}\xi+\nab'^{\leq i-1}\big((\nab'\om,\nab'\trch,\Gamma)\cdot f\big)\\
    &\quad +O(|\nab'^{\leq i} f|^2).
    \end{split}
\end{equation}

\subsubsection{Proof of Lemma \ref{Lemma:estimate-trch'}}

The following proposition concludes the proof of Lemma \ref{Lemma:estimate-trch'} and provides estimates of various quantities in the new frame.
\begin{proposition}
\label{Prop:intermediary}
    Suppose $f$ satisfies the equation \eqref{eq:nab_4'f} and is equal to $F$ on $S(v,r^*_\delta)$. Then:
    \begin{enumerate}
    \item  On  the initial sphere  $S(v,r^*_\delta)$ we have, {for all $0\le i \le N$},
    \beaa
        |\nab'^i f|\lesssim_N \delta.
    \eeaa
   
    \item The following bound holds true along $C_v$, for all $0\leq i\leq N$,
    \bea
    \sup_{C_v} |\nab'^{i}f|\lesssim_N \de.
\eea    
    \item  We have, along $C_v$,
    \bea\label{eq:nab'Riccibounds}
  \sup_{C_v}    |(\nab',\nab_4')^{i}\Gamma_1'|\lesssim_N 1, \quad i\leq N-1
\eea
where $\Gamma_1'=\{\chi',\chib',\etab',\zeta',\om'\}$ are  defined w.r.t. the frame $\{e_3',e_4',e_a'\}$ (transformed with $\la=1$).
    \end{enumerate}
\end{proposition}
\begin{proof}[Proof of 1.]
   For the initial condition, note that $f|_{S(v,r^*_\delta)}= F|_{S(v,r^*_\delta)}$ means that the derivatives of $f$ in $\eS_a$ directions are the same as  those of $F$. We define the vector field
    \def\eSf{\, ^{(S,f)}e}
\begin{equation*}
\begin{split}
    \eSf_a&:=e_a+\frac 12 \Fb_a f^b e_b+\frac 12 \Fb_a e_4+\Big(\frac 12 f_a+\frac 18|f|^2 \Fb_a\Big)e_3=\Big(e_a+\frac 12 f_a e_3\Big)+\frac 12 \Fb_a \Big(e_4+f^b e_b+\frac 14|f|^2 e_3\Big)\\
    &=e_a'+\frac 12 \Fb_a e_4'.
    \end{split}
\end{equation*}
This vector field is the same as $\eS_a$ everywhere on $S(v,r^*_\delta)$ but has the advantage of  being  always tangent to the cone $C_v$.

We have, by \eqref{eq:nab'-on-H-tensors}, $\nab_a'f=\nab_{e_a'}f+f\cdot\chib=\nab_{\eSf_a}f-\frac 12 \Fb_a \nab_{e_4'}f+f\cdot \chib$ (the last term is written schematically), and 
\begin{equation*}
\begin{split}
    \nab_a'\nab_b' f&=\nab_{e_a'}\nab_{e_b'}f+{\nab'}^{\leq 1}(f\cdot \chib)=\nab_{\eSf_a}\nab_{e_b'}f-\frac 12 \Fb_a \nab_{e_4'}\nab_{e_b'}f+{\nab'}^{\leq 1}(f\cdot \chib)\\
    &=\nab_{\eSf_a}\nab_{\eSf_b}f-\nab_{\eSf_a}\Big(\frac 12 \Fb_b \nab_{e_4'}f\Big)-\frac 12 \Fb_a \nab_{e_4'}\nab_{e_b'}f+{\nab'}^{\leq 1}(f\cdot \chib)\\
    &=\nab_{\eSf_a}\nab_{\eSf_b}f+\nab_{\eSf_a}\big(\Fb\cdot (\nab_4'f,\Gamma\cdot f)\big)+\Fb\cdot \big(\nab_4'\nab'f+(\nab_4',\nab')^{\leq 1}(f\cdot\Gamma)\big)+\nab'^{\leq 1}(f\cdot\chib)\\
    &=\nab_{\eSf_a}\nab_{\eSf_b}f+\sum_{i+j\leq 1}(\nab,\nab_4)^{\leq i}\Fb\cdot \nab'^{\leq j}(f,\xi)
\end{split}
\end{equation*}
where we use many times \eqref{eq:nab'andnab_4'on-H-tensors}, the boundedness of Ricci coefficients \eqref{eq:nab'nab'4onGamma}, and the fact that in view of \eqref{eq:nab_4'f} and \eqref{eq:nab'4nab'if}, the expression ${\nab_4'}^{i} {\nab'}^j f$ can be written schematically as ${\nab'}^{\leq j}(f,\xi)$.

Therefore, repeating this, we obtain the schematic relation
\begin{equation*}
    \nab'_{a_1}\cdots\nab'_{a_i} f=\nab_{\eSf_{a_1}}\cdots \nab_{\eSf_{a_i}}f+ \sum_{j_1+j_2\leq i-1}(\nab,\nab_4)^{\leq j_1}\Fb \cdot  {\nab'}^{\leq j_2}(f,\xi).
\end{equation*}
When we evaluate this on $S(v,r_\delta^*)$, we have $\eSf_a=\eS_a$ tangent to the sphere, so the first term on the right hand side can be replaced by $\nab_{\eS_{a_1}}\cdots \nab_{\eS_{a_i}}F$. Therefore, using Corollary \ref{Cor:HigherboundF}, we obtain, on $S(v,r_\de^*)$,
\begin{equation*}
    |\nab'^i f|\lesssim \delta+|\nab'^{\leq i-1}f|,
\end{equation*}
so by induction we see that $|\nab'^i f|\lesssim_N \delta$ for all $i\leq N$.
\end{proof}

\begin{proof}[Proof of 2.]
We  derive the estimate of $\nab'^i f$ along the null cone.          Using \eqref{eq:nab-omega-near-extremal}, we have
\begin{equation*}
    \nab'_a\om=\nab_a \om+\frac 12 f_a \nab_3\om=O(\delta_*^\frac 12),
\end{equation*}
and by induction as before,
\begin{equation*}
    |\nab'^i\om|\lesssim \left(1+|\nab'^{\leq i-1} f|\right) \delta_*^\frac 12.
\end{equation*}
Also, since $(\nab,\nab_4)^{\leq i}\trch=(\nab,\nab_4)^{\leq i}(-\frac{2r\Delta}{|q|^2})+O(\epsilon_0)\lesssim \min\{\Delta,e_4(r),\epsilon_0\}\lesssim \delta$ in $\RR_\delta$,
using \eqref{eq:nab'nab'4onGamma} we have
\begin{equation*}
\begin{split}
    |\nab'^i \trch|&\lesssim (1+|\nab'^{\leq i-1 }f|)\Big( |(\nab,\nab_4)^{\leq i}\trch|+O(|\nab'^{\leq i-1}f|)|(\nab,\nab_4,\nab_3)^{\leq i}\trch|\Big) \\
    &\lesssim (1+|\nab'^{\leq i-1 }f|)\left(\delta+O(|\nab'^{\leq i-1}f|)\right).
    \end{split}
\end{equation*}
Therefore, we have
\begin{equation*}
    \nab'_4 \nab'^i f+\left(2\om+\frac {i+1}2\trch\right)\nab'^i f=O(\epsilon_0)\nab'^{\leq i} f+O(\epsilon_0)+O(\delta_*^\frac 12)\nab'^{\leq i-1}f+O(|\nab'^{\leq i} f|^2).
\end{equation*}
Then similar to the zero-order case, we have ($|\nab'^i f|^2:=(\nab'^i f)_{a_1\cdots a_i b}(\nab'^i f)^{a_1\cdots a_i b}$)
\begin{equation*}
    \frac{d}{ds}\left(|\nab'^i f|^2\right)+\left(2\om+\frac{i+1}2\trch\right)|\nab'^i f|^2=O(\epsilon_0)|\nab'^{\leq i} f|^2+O(\delta_*^\frac 12)|\nab'^{\leq i-1} f||\nab'^{\leq i} f|+O(|\nab'^{\leq i}f|^3).
\end{equation*}
We now assume by  induction  that we have proved the estimate for $i\leq j-1$. Then for $i=j$, we integrate the equation as the zero-order case to get \footnote{For simplicity, we still use the parameter $s$; one can also use $\tau$ as the parameter. Note that along each null generator we have $ds/d\tau=1/(e_4(v)+O(\delta))$ which is comparable to $1$.}
\begin{equation*}
\begin{split}
    |\nab'^{j} f|^2&\leq C_0 \delta^2+C\int_{s}^v e^{-C\delta_*^\frac 12 (s'-s)}\left(O(\epsilon_0)|\nab'^{\leq j} f|^2+O(\delta_*^\frac 12)|\nab'^{\leq j-1} f||\nab'^{\leq j} f|+|\nab'^{\leq j} f|^3\right)\, ds'\\
    &\leq  C_0\delta^2+C\Big(\int_0^\infty e^{-C\delta_*^\frac 12} s'\, ds'\Big)
    \cdot C \delta_*^\frac 12 \sup_{s\in [s_*,v]} \left(M|\nab'^{j-1} f|^2+\frac{1}{M}|\nab'^{\leq j}f|^2+ |\nab'^{\leq j}f|^3\right).
    \end{split}
\end{equation*}
This is now similar to the zeroth-order estimate. One can make a bootstrap assumption that
\begin{equation*}
    |\nab'^{\leq N} f|\leq C_b \delta, \quad \text{for }s\in [s_1,v],
\end{equation*}
and derive
\begin{equation*}
\begin{split}
    \sup_{s\in [s_1,v]} |\nab'^{\leq j}f|^2 &\leq C_0\delta^2+C \delta_*^{-\frac 12}\cdot \delta_*^\frac 12 \Big(M\delta^2+\Big(\frac{1}{M}+C_b\delta\Big)\sup_{s\in[s_1,v]}|\nab'^{\leq j}f|^2\Big)\\
    &\leq C_0 \delta^2+CM\delta^2+\left(\frac{C}{M}+C_b \delta\right) \sup_{s\in[s_1,v]} |\nab'^{\leq j}f|^2,
    \end{split}
\end{equation*}
so picking suitably large constants $M$, $C_b$, we obtain the estimate for $j=N$ and improve the bootstrap assumption. Therefore, the estimate $|\nab'^{\leq N}f|\lesssim \delta$ is valid along each null generator until $s$ reaches the value corresponding to the point on $\Sigma_{\tau_*}$ (recall the definition of $\tau_*$ from the bootstrap argument of Proposition \ref{prop:estimate-trch'}). This concludes the proof of the Lemma \ref{Lemma:estimate-trch'}. 
\end{proof}
\begin{proof}[Proof of 3.]
As a consequence of the proof of Part 2, we also obtain the boundedness of  the  quantities $\Gamma_1'=\{\chi',\chib',\etab',\zeta',\om'\}$  in the new frame $\{e_3',e_4',e_a'\}$:
\begin{equation*}
    |(\nab',\nab_4')^{\leq i}\Gamma_1'|\lesssim 1, \quad i\leq N-1.
\end{equation*}
These estimates  follow by  applying $\nab'$ and $\nab_4'$ to  the corresponding transformation formulas and  using the estimates in Part 2. This ends the proof of Proposition \ref{Prop:intermediary}.
\end{proof}

\subsubsection{Higher order estimates on \texorpdfstring{$S_v^\tau$}{Svtau}}
\label{subsubsection:Higher-order-sphere}

Recall that we have shown that for fixed $v,\tau$, $S_v^\tau=C_v\cap \Si_\tau$ is an embedded submanifold, so Lemma \ref{lem:ezero-on-r} becomes a global statement on $S_v^\tau$, i.e., 
there exists  a frame transformation $(\, ^{(0)}f, \fbzero$, $\lazero)$ from $\{e_3,e_4,e_a\}$ on $S_v^\tau$ such that $\{\ezero_a\}$ is tangent to $S_v^\tau$, and $\Nv=\frac 12 (\ezero_4-{\ezero}_3)$ is the outward unit normal vector of $S_v^\tau$ on $\Sigma_\tau$.
In this subsection, since there is no danger of confusion, we denote for simplicity $\fb=\fbzero$, $\la=\lazero$, $\Sigma=\Sigma^\tau_v$, $S=S_v^\tau$. 

We need to prove higher-order bounds of $(\fb,\lambda)$ under the projected covariant derivatives adapted to the sphere $S$. Recall that the adapted frame $\{\ezero_3,\ezero_4,\ezero_a\}$ is transformed from $\{e_3,e_4,e_a\}$ through $(f,\fb=\fbzero,\lambda=\lazero)$. Note that
\begin{equation*}
    \ezero_a=e_a'+\frac 12 \fb_a e_4',
\end{equation*}
and we have estimates of various quantities in the $e'$-frame from Proposition \ref{Prop:intermediary}.
Using this expression of $\ezero_a$, and that $\fb$ is defined by $f$ through \eqref{eq:def-fb}, we immediately have, by the control of $(\nab',\nab_4')^{\leq i}f$,
\begin{equation}\label{eq:nab''i-fb}
    |(\nab'_{\ezero_a})^{\leq i} \fb|\lesssim 1+|(\nab'_{\ezero_a})^{\leq i-1} \fb|, \text{ so by induction, }|(\nab'_{\ezero_a})^{\leq i} \fb|\lesssim_i 1.
\end{equation}
We now prove
\def\nabzero{\, ^{(0)}\nab}
\begin{proposition}\label{Prop:bounds-in-double-primed-nabla}
    We have, on $S$,
    \bea\label{eq:nabzero-f-fb}
    |\nabzero^{\leq i}(f,\fb)
    |\lesssim_N 1,\quad i\leq N-1.\eea
\end{proposition}

\begin{proof}
The estimate has been proved for $i=0$. We now assume that the estimate holds for $i-1$ and   prove it for $i$.
    
    As before, since the horizontal structure $^{(0)} H$ (given by the new $\ezero$ frame, see \eqref{eq:expression-f0-fb0-la0}) does not depend on $\lazero$, we prove this under the frame given by $(f,\fb)$ with the simplifying assumption  $\lazero=1$, so
    \begin{equation*}
        \ezero_4=e_4',\quad \ezero_a=e_a'+\frac 12 \fb_a e_4',\quad \ezero_3=e_3'+\fb^a e_a'+\frac 14 |\fb|^2 e_3'.
    \end{equation*}

    In this setting, by Lemma \ref{Le:trasportation.formulas}, for any $H'$-horizontal covariant tensor $\psi'$, we have
    \begin{equation*}
        \nabzero^i\psi'=(\nab'_{\ezero_a})^{\leq i}\psi+\fb\cdot\chi'\cdot(\nab'_{\ezero_{a}})^{\leq i-1}\psi'+\nabzero^{\leq i-1}(\fb\cdot\chi'\cdot\psi').
    \end{equation*}
    Let $\psi'=\chi'$ first. This gives
    \begin{equation*}
    \begin{split}
        |\nabzero^{\leq i}\chi'|&\lesssim |(\nab'_{\ezero_a})^{\leq i}\chi'|+|\nabzero^{\leq i-1} \chi'|^2|\nabzero^{\leq i-1}\fb|\\
        &\lesssim |(\nab',\nab'_4)^{\leq i}\chi'|\left(1+|(\nab_{\ezero_a})^{\leq i-1}\fb|\right)+|\nabzero^{\leq i-1} \chi'|^2|\nabzero^{\leq i-1}\fb|.
        \end{split}
    \end{equation*}
Then by \eqref{eq:nab''i-fb}, the bounds of Ricci coefficients in $e'$-frame \eqref{eq:nab'Riccibounds}, and the induction assumption for $i-1$, we obtain the boundedness of $|\nabzero^{\leq i}\chi'|$.
    Then, applying this estimate to $\psi'=(f,\fb)$ and using \eqref{eq:nab''i-fb}, we establish the required estimate for $i$ and conclude the proof.
\end{proof}
It is then also straightforward to derive
\begin{equation}\label{eq:nabzero-la-Ricci'}
    |\nabzero^i(\lambda',\lambda'^{-1},\chi',\chib',\om',\zeta',\etab')|\lesssim_N 1,\quad i\leq N-1.
\end{equation}
Note that those Ricci coefficients are still the ones in the $e'$ frame.

We are now ready to derive the main estimate we need. Denote the Levi-Civita connection on $\Sigma$ by $\overline D$. The extrinsic curvature of $S$ on $\Sigma$ is given by
\begin{equation*}
    k_{ab}=g\left(\overline D_{\ezero_a} \Nv,{\ezero}_b\right)=\frac 12 g\left(\overline D_{\ezero_a} (\ezero_4- {\ezero}_3),\ezero_b\right)=\frac 12\left(\, ^{(0)}\chi_{ab}-{^{(0)}\chib}_{ab}\right).
\end{equation*}
In view of the transformation  formulas \eqref{eq:trasportation.formulas}
\begin{equation*}
    \lambda^{-1}{^{(0)}\chi}=\chi',\quad \lambda\, {^{(0)}\chib}=\chib'+{^{(0)}\nab}\underline f+\underline f\cdot(\zeta',\etab')+\fb\cdot\fb\cdot(\om',\chi'),
\end{equation*}
and the estimates \eqref{eq:nabzero-f-fb}, \eqref{eq:nabzero-la-Ricci'}, we obtain the following estimate
\begin{proposition}
    For every fixed $v,\tau$ satisfying $0\leq \tau\leq v-r_\delta^*$, the extrinsic curvature $k$ of $S_v^\tau$ on $\Sigma_\tau$ satisfies  $|\nabzero^{\leq i}k|\lesssim_N 1$, $i\leq N-2$, where the implicit constant is independent of $v,\tau$.
\end{proposition}

\begin{remark}
It is interesting to point out the  connection between our proof below and  the well  known  Cheeger-Gromov (CG) convergence theorem (see e.g. \cite{Cheeger1970}, \cite{Gromov1998}).
 First, CG does not guarantee that the limit  sphere is embedded in $\Sigma_\tau$, while we want to study the regularity of the pointwise limit of a family of spheres $\{S^\tau_v\}$ on $\Sigma_\tau$. Second, using our estimates, we can show that
    \begin{equation*}
        K(S_v^\tau)=-\, ^{(0)}\rho+\frac 1 2\, ^{(0)}\chih\cdot\, ^{(0)}\chibh+\frac 14 \, ^{(0)}\trch \, ^{(0)}\trchb=\Re\left(\frac{2m}{q^3}\right)+O(\delta),
    \end{equation*}
    so Gauss curvature may not be positive for large $a$. This prevents direct use of the Myers theorem \cite{Myers} to obtain the diameter estimate required in the convergence theorem. 
\end{remark}

\def\Db{\overline{D}}

\subsection{Estimate of the graph function}

From above we know that $i(\mathbb{S}^2)$, also denoted by $S=S_v^\tau$, is an embedded sphere on $\Sigma=\Sigma_\tau$. Furthermore it can be written as a graph function, i.e., there is a function $R\colon \mathbb{S}^2 \to [r_+-\delta,r_+ +\delta]$ so that $\Phi_{\Sigma}\{(R(p),p)\colon p\in \mathbb{S}^2\}=S$.

Recall in \eqref{eq:equilvalence-metrics-on-Sigma} that pullback metric on $(r_+-2\delta,r_+ +2\delta)\times \mathbb{S}^2$ through $\Phi_{\Sigma}$, denoted by $\overline g$, is comparable with the natural metric $\overline g_0$ on $(r_+-2\delta,r_+ +2\delta)\times \mathbb{S}^2$. We take a spherical coordinate $(\theta,\varphi)$, switched to $(x^1,x^2)$ near the poles. In the following, the partial derivative $\pa$ should be understood in these two charts of $\mathbb{S}^2$.

\begin{lemma}
\lab{Lemma:Derivatives-R}
We have the estimates  \begin{equation*}
        |\pa R|\lesssim \delta,\quad 
        |\pr^i R|\lesssim 1, \quad i\leq N.
    \end{equation*}
    In particular the constants do not depend on $v$, $\tau$.
\end{lemma} 

\begin{proof} 

To prove the estimates, we focus on the chart $(r,\theta,\varphi)$; the argument for the other chart is the same. To start with, we derive the estimate of the first-order derivative.

{\bf First order derivatives.}
    Let $\widetilde\pa_{\theta}:=\pa_{\theta}+(\pa_{\theta}R)\overline\pa_r$, $\widetilde\pa_{\vphi}:=\pa_{\vphi}+(\pa_{\vphi}R)\overline\pa_r$. Then $\widetilde\pa_{\theta}r=\pa_{\theta}R$, $\pat_{\vphi} r=\pa_{\vphi} R$, so $\pat_{\theta},\pat_{\vphi}$ are tangent to $S$ as we get zero when applying them to $r-R(\theta,\varphi)$. 

    By Lemma \ref{lem:ezero-on-r}, for any unit vector $V$ tangent to $S$, we have $|V(r)|\lesssim \delta$, so in particular,
    \begin{equation*}
        |\pat_\theta|_{\overline g}^{-1}|\pat_\theta r|\lesssim \delta,\quad |\pat_\vphi|_{\overline g}^{-1}|\pat_\vphi r|\lesssim \delta.
    \end{equation*}
    Since the metric coefficients on $\Sigma$ under $(r,\theta,\varphi)$ are bounded, we have
    \begin{equation*}
        |\pat_\theta|_{\overline g}=|\pa_\theta+(\pa_\theta R)\overline{\pa}_r|_{\overline g}\lesssim 1+|\pa_\theta R|,\text{ and similarly, }|\pat_\varphi|_{\overline g}\lesssim 1+|\pa_\vphi R|.
    \end{equation*}
    Therefore,
    \begin{equation*}
    |\pa_{\theta}R|=|\widetilde\pa_\theta r| \leq C\delta (1+|\pa_\theta R|),
\end{equation*}
which yields $|\pa_\theta R|\lesssim \delta$. Similar for $|\pa_\varphi R|$.

{\bf Second-order derivatives.} 
We compare  $\overline g$ with the Euclidean metric in the $(r,\theta,\vphi)$ coordinates, following    \cite{SmithAAlemma}.
 Recall that in Euclidean space, the second fundamental form of a graph function  is expressed in terms of  the Hessian of the function. In the following, we use Greek letters $\a,\b,\cdots$ to denote the indices on $\Sigma$, and Latin letters $i,j,\cdots$ to denote the indices on $S$. Recall the vector field $\pat_\theta$, $\pat_\varphi$ defined in the proof of the last proposition. They are the natural coordinate vectors in the $(\theta,\varphi)$ coordinate on $S$. Nevertheless, we denote them by $\pat_i$, $\pat_j$, and reserve the notation $\pa_i$, $\pa_j$ for $\pa_\theta$, $\pa_\varphi$. We use $\mathcal N(a,b,\cdots)$ to denote schematically nonlinear functions of $a,b,\cdots$. 
    
Denote the induced metric on $\Sigma$ by $\overline g$. The vector field $N_{\bar{g}}=\mathrm{grad_{\Sigma}}(r-R(\theta,\varphi))=\overline g^{\a\b}\pa_\b (r-R(\theta,\varphi))\pa_\a$ is normal to $S$ (not necessarily of unit length). We have
\begin{equation*}
    |N_{\bar{g}}|k(\pat_i,\pat_j)=\overline g(\overline D_{\pat_i}\pat_j,N_{\overline{g}})=(\overline D_{\pat_i}\pat_j)^\a \pa_\a (r-R(\theta,\varphi))
\end{equation*}
Notice that $\overline D_{\pat_i}\pat_j$ differs from $\overline D^e_{\pat_i}\pat_j$, where $\overline D^e$ is the connection associated with the Euclidean metric with respect to $(r,\theta,\varphi)$ (i.e., $dr^2+d\theta^2+d\vphi^2$), by Christoffel symbols of $\overline g$ in the $(r,\theta,\varphi)$ coordinates. For the Euclidean connection, we have
\[\overline D^e_{\pat_i}\pat_j=\overline D^e_{\pat_i}(\pa_j+(\pa_j R)\overline\pa_r)=(\pa_i\pa_j R)\overline\pa_r.\]
Therefore 
\begin{equation*}
    (\overline D^e_{\pat_i}\pat_j)^\a \pa_\a(r-R(\theta,\varphi))=(\pa_i \pa_j R(\theta,\varphi)) \overline\pa_r (r-R(\theta,\varphi))=\pa_i\pa_j R.
\end{equation*} 
So we have shown that
\begin{equation*}
    \pa_i\pa_j R=|N_{\bar{g}}|k(\pat_i,\pat_j)+\mathcal N(\Gamma_{\bar{g}},\pa R),
\end{equation*}
where $\Gamma_{\bar{g}}$ represents the Christoffel symbol components of $\overline g$ in $(r,\theta,\varphi)$, which are independent of the sphere $S$.
Since $|N_{\bar{g}}|$ is bounded as long as $\pa R$ and $\overline g^{-1}$ are bounded, also using $\overline g(\pat_i,\pat_i)\lesssim 1+|\pa_i R|^2$ we deduce $|k(\pat_i,\pat_j)|\lesssim |k|(1+|\pa R|^2)\lesssim 1$. This  establishes  the boundedness of $|\pa^2 R|$.

\def\gb{{\bar{g}}}

{\bf Higher-order derivatives.} For higher orders, we have (recall that the covariant derivative on the sphere $S$ is denoted by $\nabzero$) 
\begin{equation}\label{eq:nabNk-expression}
\begin{split}
    \nabzero^m &k(\pat_i,\pat_j;\pat_{l_1},\cdots,\pat_{l_m})=\nabzero_{\pat_{l_m}}\nabzero^{m-1}k(\pat_i,\pat_j;\pat_{l_1},\cdots,\pat_{l_{m-1}})\\
    &
    =\pat_{l_{m}}(\nabzero^{m-1}k(\pat_i,\pat_j;\pat_{l_1},\cdots,\pat_{l_{m-1}}))-\nabzero^{m-1}k(\nabzero_{\pat_{l_m}}\pat_i,\pat_j;\pat_{l_1},\cdots,\pat_{l_{m-1}})\\
    &\qquad \qquad -\cdots-k(\pat_i,\pat_j;\pa_{l_1},\cdots,\nabzero_{\pat_{l_m}}\pat_{l_{m-1}}).  
\end{split}
\end{equation} 
To estimate $\nabzero_{\pat_i} \pat_j$, note that the induced metric on the sphere, denoted by $\slashed{g}$, satisfies
\begin{equation*}
    \slashed{g}(\pat_i,\pat_j)=\gb(\pa_i+(\pa_i R)\overline\pa_r,\pa_j+(\pa_j R)\overline\pa_r)=\mathcal N(\pa R)\cdot \gb,
\end{equation*}
where $\gb$ in the last term means coefficients of $\gb$ in $(r,\theta,\varphi)$ coordinates. Therefore, since $\nabzero_{\pat_i} \pat_j$ is given by Christoffel symbols, we have $|\nabzero_{\pat_i} \pat_j|\lesssim \mathcal N(\pa^{\leq 2}R,\pa^{\leq 1}\gb)$.

We claim that for any $m\leq N-2$ we have
\begin{equation*}
    |N_\gb|\nabzero^m k(\pat_i,\pat_j;\pat_{l_1},\cdots,\pat_{l_m})=\pa_{l_m}\cdots \pa_{l_1}\pa_i\pa_j R+\mathcal N(\pa^{\leq m+1}R,\pa^{\leq m+1}\gb).
\end{equation*}
In the proof above, we have shown that this is true when $m=0$. Now suppose that this holds with $m$ replaced by $m-1$. Then using \eqref{eq:nabNk-expression}, $|N_\gb|^{-1}=\mathcal N(\gb,\gb^{-1},\pa R)$ and the induction assumption, we have
\begin{equation*}
\begin{split}
    \nabzero^m k(\pat_i,\pat_j;&\pat_{l_1},\cdots,\pat_{l_m})
    =\pat_{l_{m}}\left(|N_\gb|^{-1}|N_\gb|\nabzero^{N-1}k(\pat_i,\pat_j;\pat_{l_1},\cdots,\pat_{l_{N-1}})\right)\\
    &\quad +\mathcal N(\nabzero^{m-1}k,\pa^{\leq 2}R,\pa^{\leq 1}\gb)\\
    &=|N_\gb|^{-1} \pat_{l_{m}}\left(\pa_{l_{m-1}}\cdots \pa_{l_1}\pa_i\pa_j R+\mathcal N(\pa^{\leq m}R,\pa^{\leq m}\gb)\right)+\mathcal N(\nabzero^{m-1}k,\pa^{\leq 2}R,\pa^{\leq 1}\gb)\\
    &=|N_\gb|^{-1} \pa_{l_m}\cdots \pa_{l_1}\pa_i\pa_j R+\mathcal N(\pa^{\leq m+1}R,\pa^{\leq m+1}\gb)
    \end{split}
\end{equation*} 
which means that it also holds for $m\leq N-2$, concluding the induction. Again noting that $\overline g(\pat_i,\pat_i)\lesssim 1$, we obtain
\begin{equation*}
    |\pa^{\leq N}R|\leq C_N,
\end{equation*}
which is independent of $v$. 
This end the proof of Lemma \ref{Lemma:Derivatives-R}.
\end{proof}


\subsection{Convergence of the spheres  \texorpdfstring{$S^\tau_v$}{Svtau} on \texorpdfstring{$\Sigma_\tau$}{Sigmatau}}
\lab{section:Convergence of S^0_v}

With the uniform bounds obtained in Lemma \ref{Lemma:Derivatives-R}, we can now apply the Arzela-Ascoli lemma to get
\begin{proposition}
\lab{Proposition:convergence}
 On each $\Sigma_\tau$, the family of spheres $S^\tau_v$ converges  pointwise             to a $C^{N-1}$ sphere $S_\tau^*$ as $v\to\infty$.  
\end{proposition}

\begin{proof}
We have shown that under the coordinate charts $(\theta,\varphi;\frac{\pi}4< \theta < \frac{3\pi}4)$ and $(x^1,x^2;0\leq \theta < \frac{\pi}3\text{ or }\frac {2\pi}3< \theta\leq \pi)$, each sphere $S_v^\tau$ can be expressed as $r=R_v(\theta,\varphi)$, $r=R_v'(x^1,x^2)$ and their $C^N$ norms are uniformly bounded, independent of $v$. 
Therefore, we can find a subsequence $\{v_n\}\rightarrow \infty$ such that $R_{v_n}$ and $R_{v_n}'$ converges in $C^{N-1}$ norm (over the corresponding region) to $C^{N-1}$ functions $R^*(\theta,\varphi)$, $R'^*(x^1,x^2)$, which, of course, coincide with the pointwise limit. This in fact gives a diffeomorphism from $\mathbb{S}^2$ to $S_\tau^*$. Therefore, $S_\tau^*$ is an embedded $C^{N-1}$ submanifold.
\end{proof}

\subsection{The event horizon}
\lab{section:EvenHorizon}
We have constructed a limit sphere $S^*_\tau$ on $\Sigma_\tau$. In this section, we show that the union $\cup_{\tau\geq 0} S^*_\tau$ coincides with the future outgoing null cone of $S^*=S^*_0$ and gives the event horizon. 

For any given $\tau_1$, we consider the future outgoing (past incoming) null cone of $S_{\tau_1}^*$, denoted by $C_{\tau_1}^*$. Since $S_{\tau_1}^*$ is smoothly embedded and the background metric is smooth, by Lemma \ref{lem:local-in-time-regular}, $C_{\tau_1}^*$ is regular between $\Sigma_{\tau_1-\delta_{\tau_1}}$ and $\Sigma_{\tau_1+\delta_{\tau_1}}$ for some $\delta_{\tau_1}>0$. For any $\tau\in (\tau_1-\delta_{\tau_1},\tau_1+\delta_{\tau_1})$, we denote the intersection of $C_{\tau_1}^*$ and $\Sigma_\tau$ by $S^*_{\tau_1;\tau}$.
We have
\begin{proposition}
    For any $\tau\in (\tau_1- \delta_{\tau_1},\tau_1+\delta_{\tau_1})$, we have $S_{\tau_1;\tau}^*=S_{\tau}^*$.
\end{proposition}
\begin{proof}
Throughout this proof, all discussions are restricted to $\tau\in (\tau_1- \delta_{\tau_1},\tau_1+\delta_{\tau_1})$,  where $C^*_{\tau_1}$ is regular.

{\bf Case  $\tau\geq \tau_1$. }

{\bf Step 1: }Show that $S_{\tau_1;\tau}^*$ is in the interior\footnote{Throughout the proof, the words ``future", ``past", ``interior" and ``exterior" of an object include the object itself. We will specify (e.g. by using the word ``strict") otherwise.} of $S_\tau^*$.

Since $S^*_{\tau}=\lim_{v\rightarrow\infty} S^\tau_v$, it suffices to show that $S_{\tau_1;\tau}^*$ is in the interior of $S_v^\tau$ for each $v$. By construction, on $\Sigma_{\tau_1}$ we know that $S_{v}^{\tau_1}$ is in the exterior of $S_{\tau_1}^*$ for any $v$, by Corollary \ref{Cor:int-ext} applied to $\widehat{C}_v(\tau_1,\tau)$, we know that all points in $\JJ^+(S_{\tau_1^*})$ cannot be {strictly}  in the exterior of $S^\tau_v$
    on $\Sigma_\tau$, i.e. $S_{\tau_1;\tau}^*$ must be in the interior of $S^\tau_v$,
    for any $v$, as required.

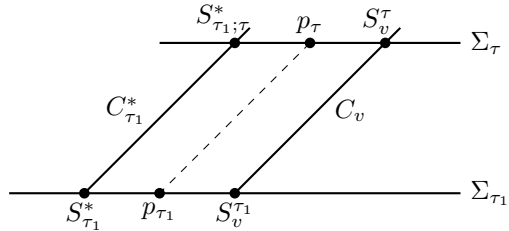
\begin{center}
\begin{tikzpicture}
    \draw[thick] (0,1) -- (6,1) node[right] {$\Sigma_{\tau_1}$};
    \draw[thick] (2,3) -- (6,3) node[right] {$\Sigma_{\tau}$};
    \draw[dashed] (4,3) -- (2,1) ;

    \draw[thick] (1,1) -- (3.2,3.2) node[midway,left=0.15cm] {$C_{\tau_1}^*$};
    \draw[thick] (3,1)  -- (5.2,3.2) node[midway,right=0.1cm] {$C_v$};

    \fill[] (3,3) circle (2pt) node[left=0.15cm,above] {$S^*_{\tau_1;\tau}$};
    \fill[] (5,3) circle (2pt) node[left=0.1cm,above] {$S^\tau_v$};
    \fill[] (1,1) circle (2pt) node[below] {$S^*_{\tau_1}$};
    \fill[] (3,1) circle (2pt) node[below] {$S^{\tau_1}_v$};
    \fill[] (4,3) circle (2pt) node[above] {$p_\tau$};
    \fill[] (2,1) circle (2pt) node[below] {$p_{\tau_1}$};
\end{tikzpicture}
\end{center}

{\bf Step 2:} Show that they must coincide.

If they do not, then there is a point $p_\tau\in \Sigma_\tau$ in between, i.e., in the strict exterior of $S_{\tau_1;\tau}^*$, and the interior of every $S_v^\tau$. For each $v'$, there exists a point $p_{\tau_1,v'}\in \mathcal J^-(p_\tau)$ on $\Sigma_{\tau_1}$ that lies in the interior of $S_{v'}^{\tau_1}$.
Then, we can pick a sequence of $p_{\tau_1,v_n}$, lying in the interior of $S_{v_n}^{\tau_1}$ for each $n$. By compactness, there is a converging subsequence of $p_{\tau_1,v_n}$ on $\Sigma_{\tau_1}$. On the other hand, by definition
\begin{equation*}
    S^*_{\tau_1}=\lim_{v\rightarrow\infty} S^{\tau_1}_v,
\end{equation*}
this subsequential limit must be on $S_{\tau_1}^*$. This means that a subsequence of $p_{\tau,v}\in \JJ^-(p_\tau)$ is converging to $S_{\tau_1}^*$.  Since the causal past of a point is closed for any globally hyperbolic spacetime (\cite[Theorem 8.3.11]{Wald}, this implies $S_{\tau_1}^*\cap \JJ^{-1}(p_\tau)\neq \varnothing$, i.e., $p_\tau\in \JJ^+(S_{\tau_1}^*)$, 
which contradicts the assumption that $p_\tau$ is in the strict exterior of $S_{\tau_1;\tau}^*$, hence not on $S_{\tau_1;\tau}^*$.

{\bf Case $\tau\leq \tau_1$.} The proof is similar. We first have that $S_{v}^\tau\subset \JJ^-(S_v^{\tau_1})$ must be in the exterior $S_{\tau_1;\tau}^*$ since $\widehat{C}_{\tau}^*(\tau,\tau_1)$ is achronal.
To prove they coincide, assume again a point $p_\tau$ in between (not on $S_{\tau_1;\tau}^*$ or $S_v^\tau$ and for all $v$). Then by a similar argument, there must be a point $p_{\tau_1}\in \JJ^+(p)\cap \Sigma_\tau$ that lies in the exterior of $S_{\tau_1}^*$ on $\Sigma_{\tau_1}$. Then since $S_{\tau_1}^*$ is the limit of $S_v^{\tau_1}$, there exists some $v'$ such that $p_{\tau_1}$ is in the exterior of $S_{v'}^{\tau_1}$, so applying Corollary \ref{Cor:int-ext} to $C_{v'}(\tau,\tau_1)$ we see that $p_\tau$ cannot be in the strict interior $C_{v'}(\tau,\tau_1)$, a contradiction.
\end{proof}

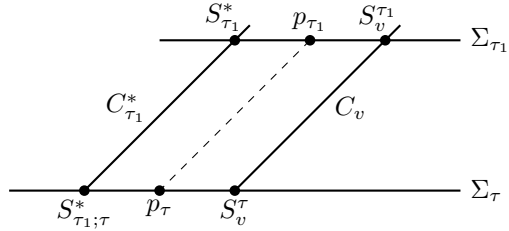
\begin{center}
\begin{tikzpicture}
    \draw[thick] (0,1) -- (6,1) node[right] {$\Sigma_{\tau}$};
    \draw[thick] (2,3) -- (6,3) node[right] {$\Sigma_{\tau_1}$};
    \draw[dashed] (4,3) -- (2,1) ;

    \draw[thick] (1,1) -- (3.2,3.2) node[midway,left=0.15cm] {$C_{\tau_1}^*$};
    \draw[thick] (3,1)  -- (5.2,3.2) node[midway,right=0.1cm] {$C_v$};

    \fill[] (3,3) circle (2pt) node[left=0.15cm,above] {$S^*_{\tau_1}$};
    \fill[] (5,3) circle (2pt) node[left=0.1cm,above] {$S^{\tau_1}_v$};
    \fill[] (1,1) circle (2pt) node[below] {$S^*_{\tau_1;\tau}$};
    \fill[] (3,1) circle (2pt) node[below] {$S^\tau_v$};
    \fill[] (4,3) circle (2pt) node[above] {$p_{\tau_1}$};
    \fill[] (2,1) circle (2pt) node[below] {$p_\tau$};
\end{tikzpicture}
\end{center}

\vspace{1ex}

We now show that the cone $C_0^*$ is globally regular towards the future and is the union of all $S^*_\tau$ with $\tau\geq 0$. In a neighborhood of $S_0^*$, $C_0^*$ is regular, and by the proposition, coincides with the union of $S_\tau^*$. Now let $\tau^*$ be the supremum of $\tau'$ so that $C_0^*$ is regular for $\tau\in [0,\tau')$.
Applying the proposition with $\tau_1=\tau^*$, we obtain a piece of null cone for $\tau\in (\tau^*-\delta_{\tau^*},\tau^*+\delta_{\tau^*})$, which coincides with the union of $S_\tau^*$. 
On the other hand, by assumption, $C_0^*$ is regular for $\tau\in [0,\tau^*-\frac 12 \delta_{\tau^*}]$, and, since $[0,\tau^*-\frac 12 \delta_{\tau^*}]$ is a compact interval, we see that $C_0^*$ equals $\cup_\tau S^*_\tau$ on this interval using the proposition.
This extends $C_0^*$ regularly to $\tau\in [0,\tau_*+\delta')$ with $C_0^*=\cup_\tau S^*_\tau$, so we must have $\tau^*=\infty$. 

We have constructed a globally regular future outgoing null cone $C_0^*$ which equals the union of $S_\tau^*$. We now show that it is the event horizon.
\begin{proposition}
    The future null cone $C_0^*$  coincides with  the event horizon (in the future of $\Sigma_0$).
\end{proposition}
\begin{proof}
 Since $S_\tau^*\subset \RR_\delta$, we see that the future of any points on $C_0^*$ is not reaching $r$-value larger than $r_+ +\delta$. Therefore, by Definition \ref{def:past-of-scri}, points on $C_0^*$ are not in $\mathcal J^{-1}(\ip)$. 

    By definition, to show that $C_0^*$ is indeed the event horizon, it remains to show that any neighborhood of a point on $C_0^*$ contains a point belonging to $\mathcal J^-(\ip)$. To see this, notice that such a neighborhood contains a point $p$ that lies in the strict exterior of $S^*_\tau$ on $\Sigma_\tau$ for some $\tau$, so it must lie in the exterior of some $S^\tau_v$ for some large $v$, hence the exterior of $C_v$. This means that $p$ is in the causal past of $S(v,r^*_\delta)$. Then by the characterization \eqref{eq:pastofscri}, we know that $p\in \mathcal J^{-1}(\ip)$.
\end{proof}

\def\const{\mathrm{const}}


\section{Proof of the uniqueness theorem}
\lab{section:ProofUniqueness}
We have  shown that        the event horizon $\hp$ of the perturbed spacetime $\MM$ is regular. According to the regularity proof, any point in the strict exterior side of $\HH^+$ lies in the causal past of $S(v,r_\delta^*)$ for $v$ large enough, and hence belongs to $\mathcal J^-(\ip)$. Hence,  there cannot be another future inextendible  null hypersurface  (future horizon)  containing a point in the exterior  side of $\hp$\footnote{The notions ``exterior side" and ``interior side" are well-defined    in view of $\HH^+=\cup_\tau S^*_\tau$, see also Section \ref{subsubsection:null-cones-over-spheres}.} that remains in $\RR_\delta$. 

However, it may still be the case that   there exists    such  a future horizon   in the {interior} side of $\HH^+$  that  does not exit   $\RR_\delta$. In this section, we eliminate this possibility by proving a stronger statement that any null geodesic in the interior side hits the spacelike hypersurface $\{r=r_+ -\delta\}$, hence entering the trapped region.

Recall that $v$, $r$ are coordinate functions under the ingoing PG structure of $\MM$, see Section \ref{section:Quantbonds-Mint}. 
Given $V>0$, we construct a foliation in the region $\RR_\delta\cap \{v\leq V\}$ as follows: We define
\[u:=-(r-r_+)\quad \text{ on }\{v=V\}\cap \RR_\delta.\]
So $u=0$ on $S(V,r_+)$, and $\pa_r u=-1$ on $\{v=V\}$.
We then construct the null cone of $S(V,r)$, denoted by $C_{V,r}$, in the past incoming direction, for $-\delta\leq r-r_+\leq \delta$.
The spheres here are also $S(v,r)$-type spheres within $\RR_\delta$, so by a similar argument as in Section \ref{section:RegularityC_vS_v}, we can show that $C_{V,r}$ are regular null cones towards the past of $S(V,r)$, and there exists a null frame transformation
\begin{equation*}
    e_4'=e_4+f^a e_a+\frac 14|f|^2 e_3,\quad e_a'=e_a+\frac 12 f_a e_4,\quad e_3'=e_3
\end{equation*}
such that $e_4'$ is tangent to the null generators of the null cone, and the frame satisfies the estimates $|f|, |\nab'f|\lesssim \delta$ (hence in particular, $|\trch'|\lesssim \delta$) uniformly in $V$ and $r$.

 We extend $u$, from $\{ v= V\}$, so that it is constant on each null cone $C_{V,r}$. We denote the level hypersurfaces   of  $u$ by $H_u$.

Since $e_a'$ is in the orthogonal complement of $e_4'$, hence tangent to the null cone $H_u$, we have, using $\nab'u=0$, $\nab_4' u=0$,
\begin{equation}\label{eq:e_4(e_3(u))}
\begin{split}
    \nab'_4(e_3'(u))&=e_4'(e_3'(u))-e_3'(e_4'(u))=(D_{e_4'}e_3')u-(D_{e_3'}e_4')u\\
    &=2(\etab'-\eta')\cdot\nab' u+2\om'\nab_3' u-2\omb'\nab_4'u=2\om' e_3'(u).
    \end{split}
\end{equation}

\begin{remark}
    To be precise,  we also need to show that $\eta'$ is finite at each point. (The finiteness of other quantities $\omb'=0$, $\eta'$, $\om'$ are clear as they only depends on $f$ itself, but not on any derivatives of $f$.)
    \end{remark}
    \begin{proof}[Proof of remark]
       One has the following transformation formula using that $e_3$ is geodesic: 
\begin{equation*}
    \eta' = \eta +\frac{1}{2}\nab_3' f. 
\end{equation*}
Therefore, we need to estimate of $\nab_3' f$. Commuting $\nab_3'$ with the equation \eqref{eq:nab_4'f} we obtain
\begin{equation}
    \nab_4' \nab_3' f+(\frac 12 \trch+2\om)\nab_3' f+\nab_3(\frac 12\trch+2\om)f=[\nab_4',\nab_3']f+\nab_3'\big(-2\xi-\chih\cdot f+\frac 12 \atrch \dual f+O(|f|^2)\big),
\end{equation}\label{eq:transport-nab3f}
with the initial condition $\nab_3'f =\nab_3 F$ on $\{v=V\}$. In view of the formula
\begin{equation*}
    [\nab_4',\nab_3']f=2(\eta'-\etab')\cdot \nab' f+2\om'\nab_3' f-2\omb'\nab_4'f+2(\etab'\cdot f)\eta'-2(\eta'\cdot f)\etab'-\dual f\cdot \dual \rho'
\end{equation*}
and the transformation formulas \eqref{eq:trasportation.formulas}, we see that with $f$ and $\nab'f$ already controlled, there are no quadratic terms of $\nab_3' f$ in \eqref{eq:transport-nab3f} (which is markedly different from the transport equation of $\nab'f$ studied in Section \ref{section:RegularityC_vS_v}), so $|\nab'_3 f|$ grows at most exponentially. Therefore we see that $\nab_3'f$ remains finite at any given point and hence so is $\eta'$.
\end{proof}

We now prove the uniqueness theorem using this construction. Suppose there is a future null geodesic $\ga$  
in the interior side of $\hp$ that stays  entirely in the region $\RR_\delta$, see Figure \ref{fig:Idea-of-proof-uniqueness} below. Fix a value $v_0$. Let $p_1$  denote a point of    intersection of $\ga$ with  $\{v=v_0\}$\footnote{Since $\{v=v_0\}$ is not achronal, it is possible that $\gamma$ intersects with it more than once. Nevertheless, the argument holds for any points of intersection.}.    Consider the integral curve of $\pa_r=-e_3$ passing through $p_1$ {along $v=v_0$}. This is a transversal null geodesic, hence intersecting $\hp$  at a unique  point $p_0$. 

We want to show that on each such integral curve of $\pa_r=-e_3$, where $r$ is of course an affine parameter, we have $r(p_1)=r(p_0)$, and hence $p_1=p_0$. Then by the arbitrariness of $v_0$, we  deduce  that $\gamma\subset \HH^+$,
contradicting  the above assumption.

To show this, we assume that for $p_1\in\{v=v_0\}$, we have $p_1\neq p_0$, i.e. $d(v_0):=|r(p_1)-r(p_0)|\big|_{v=v_0}\neq 0$.
The proposition below shows that in this case, the difference $u(p_1)-u(p_0)$ would be exponentially amplified compared with $r(p_1)-r(p_0)$.
\begin{proposition}\label{prop:u-difference-estimate-uniqueness}
    For $d(v_0)=|r(p_1)-r(p_0)|\big|_{v=v_0}$ and the function $u$ constructed from $\{v=V\}$ above (for some $V>v_0$), we have $u(p_1)-u(p_0)\geq d(v_0)\, e^{\frac 12 C\delta_*^\frac 12 (V-v_0)}$.
\end{proposition}
\begin{proof}

For $\om'$, from the transformation formula in \eqref{eq:trasportation.formulas}, the bounds \eqref{eq:nab-omega-near-extremal}, \eqref{eq:estim.Gac}, and $|f|\lesssim \delta$, we have
\begin{equation*}
    \om'={\om}+\frac 12 f\cdot({\zeta}-{\etab})+O(|f|^2)\leq -C\delta_*^\frac 12+O(\delta),
\end{equation*}
so we also have $\om'\leq -C\delta_*^\frac 12$ for some constant $C>0$. Then, since $e_3'(u)=1$ on $\{v=V\}$, and $e_4'(v)=e_4(v)+O(\delta)\sim 1$, using \eqref{eq:e_4(e_3(u))} we get $e_3'(u)\geq e^{\frac 12 C\delta_*^\frac 12(V-v_0)}$ on $\{v=v_0\}$. Therefore,
\begin{equation*}
    u(p_1)-u(p_0)=\int_{r(p_1)}^{r(p_0)} (-\pa_r u)\, dr\geq |r(p_0)-r(p_1))|\big|_{v=v_0}\cdot\inf_{v=v_0}|e_3'(u)|\geq d(v_0)\, e^{C\delta_*^\frac 12(V-v_0)}
\end{equation*}
as stated. In particular,
 we can pick suitably large $V$ (the above process works for any $V$) to make $u(p_1)-u(p_0)$ large (say, larger than $\delta^{-1}$). 
\end{proof}

\begin{figure}[htbp]
    \centering 
    \begin{tikzpicture}
    \draw[thick, name path=Hu] (0,0.5) -- (3.3,3.8) node[above] {$H_\delta$};
    \draw[dashed, name path=HH] (0.8,0.2) -- (4.8,4.2) node[right] {$\HH^+$};
    \draw[thick, name path=Hmd] (1.6,-0.1) -- (4.6,2.9) node[above] {$\quad H_{-\delta}$};

    \draw[thick, name path=v0] plot[smooth, tension=1] coordinates {(0.3,1.8) (1.3,0.8) (2.2,-0.35)}
    node[below] {$v=v_0$}
    ;
    \draw[thick, name path=V] plot[smooth, tension=1] coordinates {(2.7,3.95) (3.6,3.1) (4.5,2)}
    ;
    \node at (4.5,1.82) {$v=V$}; 

    \path[name intersections={of=Hu and V, by=uV}];
    \path[name intersections={of=Hmd and V, by=mdV}];
    \path[name intersections={of=HH and v0, by=Hv0}];

    \filldraw [] (uV) circle (1.5pt) node[left=1mm] {\scriptsize $S(V,r_+-\delta)$};
    \filldraw [] (mdV) circle (1.5pt) node[right=1mm] {\scriptsize $S(V,r_++\delta)$};
    \filldraw [] (Hv0) circle (1.5pt) node[right=1mm] {\small $p_0$};

    \path[name path=line1] (0.1,0) -- (4.1,4);
    \path[name intersections={of=line1 and v0, by=p1}];
    \filldraw [] (p1) circle (1.5pt) node[below=0.6mm,left=0.2mm] {\small $p_1$};

    \path[name path=line2] (0,0.15) -- (4,4.15);
    \path[name intersections={of=line2 and V, by=p1V}];
    \filldraw [] (p1V) circle (1.5pt) node[right=0.25mm] {\small $p_{1,V}$};
    \draw[blue] plot[smooth, tension=1] coordinates {(p1) (2,2.1) (p1V)}node[right] {}
    ;
    \node[text=blue] at (2.35,2.15) {$\gamma$};

\end{tikzpicture}
\captionsetup{width=.7\linewidth, justification=centering}
    \caption{Idea of the proof of uniqueness. Note that the event horizon $\HH^+$ does not correspond to any $H_u$, but it lies between $H_{-\delta}$ and $H_\delta$.} 
    \label{fig:Idea-of-proof-uniqueness} 
\end{figure}
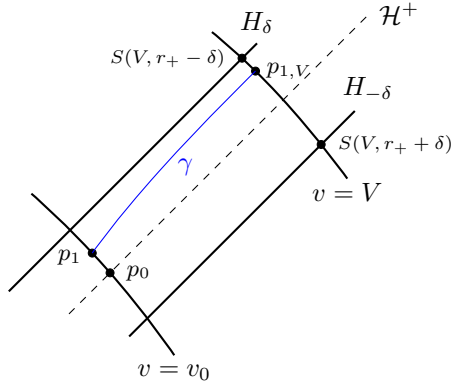

On the other hand, let $p_{1,V}$ be the intersection of the null geodesic $\gamma$ and $\{v=V\}$. By assumption, $\gamma$ lies in $\RR_\delta$, so we have $p_{1,V}\in \RR_\delta$, and hence $-\delta \leq u(p_{1,V})\leq \delta$. Since $u$ is an optical function, it is non-increasing backward along any null geodesic, which means that $u(p_1)\leq \delta$. Moreover, since $p_0\in \hp$, 
by the characterization \eqref{eq:pastofscri}, we know that any point in the causal future of $p_0$ lies in $\{r<r_+ +\delta\}$, so in particular we must have $u(p_0)> -\delta$.
Therefore,
\begin{equation*}
    0\leq u(p_1)-u(p_0)\leq 2\delta,
\end{equation*}
where the lower bound is due to the fact that $p_1$ and $p_0$ are connected by the integral curve of $-\pa_r=e_3$, which is a causal curve. Therefore, we get a contradiction on the size of $u(p_1)-u(p_0)$, so we must have $r(p_1)=r(p_0)$, hence $p_1=p_0$.

\appendix
\section{Proof of  the Embedding Criterion}\label{subsubsect:proof-im-em1}
In this appendix, we prove the Embedding Criterion (Proposition \ref{lem:im-em1}). For simplicity, we denote $S=S_v^\tau$, $i=i_v^\tau$, $\Sigma=\Sigma_\tau$. We restate the proposition:
\begin{proposition}\label{prop:immersed-embedded}
    Suppose that $S$ is the image of an immersion $i:\mathbb{S}^2\to \Sigma$ with $i(\mathbb{S}^2)\subset \Si\cap \RR_\delta$, so for each point $p\in\mathbb{S}^2$,  there exists  a neighborhood $\OO_p$ such that $i|_{\OO_p}$ is an embedding into $\Sigma$. Then there exists a small number $\mathring \delta>0$ such that  the following statement holds:
    
    If for each $p\in \mathbb{S}^2$,  there exists a tangent orthonormal  frame $\{\ezero_a\}$ (w.r.t. the induced metric from $\Sigma$) on $i(\OO_p)$ verifying
$|\ezero_a(r)|\leq \mathring\delta$, then the immersion is  in fact  an embedding. Moreover, there exists a smooth map
    \begin{equation*}
        R\colon \mathbb{S}^2\to [r_+-\delta,r_+ +\delta], \quad p\mapsto R(p)
    \end{equation*}
    such that $\Phi_{\Sigma}(\{(R(p),p)\colon p\in \mathbb{S}^2\})=S$ (with  $\Phi_{\Sigma}=\Phi_{\Sigma_\tau}$ defined by \eqref{eq:def-diffeomorphism-Sigma0}).
\end{proposition}
\begin{proof}  We briefly explain the ideas. The smallness of $\ezero_a(r)$ in the lemma will imply that locally the image of $i$ on $\Sigma$ can be written as a graph $r=R(\theta,\varphi)$ (or $R(x^1,x^2)$; we may omit this below)\footnote{Recall that $(\theta,\varphi)$ and $(x^1,x^2)$ are coordinates on $\Sigma$ given by the diffeomorphism $\Phi_{\Sigma}$.}. If the image $i$ fails to be injective, the pre-image of some point on $i(\mathbb{S}^2)$ will give isolated points, hence disconnected.     Using the local graph property, we make a continuity argument with respect to a polar angle on $\Sigma$ (constructed such that the point on $\Sigma$ where the injectivity fails correspond to the direction of the north pole), and reach a conclusion essentially saying that $\mathbb{S}^2$ becomes disconnected with finite points removed, which is a topological contradiction.

We proceed in steps as follows.

{\bf Step 1.} 
We show that for each $p\in \mathbb{S}^2$, there is a smaller neighborhood $\tilde \OO_p\subset \OO_p$, such that $i(\tilde\OO_p)$ can be written as $r=R(\theta,\varphi)$ on $\Sigma$. In other words, the map $P_{\mathbb{S}^2}\circ \Phi_{\Si}^{-1}\circ i|_{\tilde \OO_p}$ is a diffeomorphism onto its image on $\mathbb{S}^2$ (See Figure \ref{fig:diagram}).  We will refer to this as the  \textit{local graph  property at the point $p$}.

By assumption, any unit tangent vector $V$ of $i(\OO_p)$ satisfies $|V(r)|\leq 2\mathring\delta$. Clearly  the vector field $\overline{\pa}_r$ is transversal as can be  easily seen  from \eqref{eq:grr} and  $\overline{\pa}_r(r)=1$.
 In other words, assuming that   $G(r,\theta,\varphi)=0$ (or $G(r,x^1,x^2)=0$) is a local description of the embedded submanifold $i(\OO_p)$, then
\begin{equation*}
    \frac{\pa G}{\pa r}(r,\theta,\varphi)\neq 0.
\end{equation*}
Therefore, by the implicit function theorem,  there exists a neighborhood $\tilde \OO_p$, such that $i(\tilde \OO_p)$ has the form     $r=R(\theta,\varphi)$. In particular, the map $P_{\Si;\SSS^2}\circ i|_{\tilde \OO_p}$\footnote{Recall the definition of $P_{\Sigma;\mathbb{S}^2}$ in \eqref{eq:def-projection-S^2}.} is an embedding to $\SSS^2$.

{\bf Step 2.} 
 The goal of this step is to contradict the following assumption.

{\bf A.} Assume  that there are two distinct points $p_1,p_2\in\mathbb{S}^2$ such that $P_{\Sigma;\mathbb{S}^2}(i(p_1))=P_{\Sigma;\mathbb{S}^2}(i(p_2)) :=P_0$ are the same on $\mathbb{S}^2$.

Assuming {\bf A},  we        introduce some  convenient notations. On $\mathbb{S}^2$ consider the spherical coordinates $(\theta_{P_0},\varphi_{P_0})$ with $P_0$ being the north pole and $\theta_{P_0}$ being the polar angle. Denote the corresponding south pole by $P_0'$, so $\theta_{P_0}(P_0)=0$, $\theta_{P_0}(P_0')=\pi$. For each point $p$ on $\Sigma$, we assign a $\vartheta$-value to $p$, denoted by $\vartheta(p)$,        defined  to be   the the pull back through  $P_{\Si;\SSS^2} $  of the     $\theta_{P_0}$-coordinate value, i.e. the $\theta_{P_0}$ value of $P_{\Si;\SSS^2}(p)$. 
Clearly, $\vartheta$ is a continuous function on $\Sigma$.

\begin{figure}
    \centering
    \begin{tikzpicture}
  \node (A) at (-5.5,0) {$\mathbb{S}^2$};
  \node (B) at (-2.5,0) {$\Sigma\cap \{|r-r_+|<2\delta\}$};
  \node (C) at (2.5,0) {$(r_+-2\delta,r_+ +2\delta) \times \mathbb{S}^2$};
  \node (D) at (6,0) {$\mathbb{S}^2$};
  \node (E) at (-5.5,-0.6) {($G_\vartheta$, $H_{\vartheta}$)};
  \node (F) at (6,-0.6) {($P_0$, $P_0'$)};

  \draw[->] (A) -- (B) node[midway,above] {$i$};
  \draw [->] (B) -- (C) node[midway,above] {$\Phi_{\Sigma}^{-1}$};
  \draw[->] (C) -- (D) node[midway,above] {$P_{\SSS^2}$};
\end{tikzpicture}
\captionsetup{width=.7\linewidth, justification=centering}
    \caption{A diagram illustrating the maps. Note also the notation $P_{\Si;\mathbb{S}^2}=P_{\mathbb{S}^2}\circ \Phi_{\Si}^{-1}$. The local property says that the map from the left $\mathbb{S}^2$ to the right $\mathbb{S}^2$ is a local diffeomorphism.}\label{fig:diagram}
\end{figure}



Define $G_\kappa:=i^{-1} (\{p\in \Si\colon \vartheta(p)\leq \kappa\})$, $0\le \ka\le \pi$. 
Also denote 
    $H_{\kappa}:=i^{-1} (\{p\in \Si\colon \vartheta(p)= \kappa\})$. 
They are closed subsets of $\mathbb{S}^2$, hence compact. Moreover, for each point $p$ in $H_0$, using the local graph  property, we see that $\tilde\OO_p \cap H_0=\{p\}$, where $\tilde\OO_p$ is defined in  Step 1. By compactness we deduce that $H_0$ only contains finite points. The same thing holds for $H_\pi$.

We first prove a lemma.
 \begin{lemma}\label{lem:Near-H_theta}
    Given $\vartheta_0\in [0,\pi)$, and an open neighborhood $\Om$ of $G_{\vartheta_0}$ on $\mathbb{S}^2$, there exists a small $\delta'>0$, such that  $G_{\vartheta_0+\delta'}\subset \Om$. Similarly, for any open neighborhood of $H_\pi$, there exists a small $\delta''>0$ such that $H_{\pi-\delta'}\in \Om$ for all $\delta'\leq\delta''$.
\end{lemma}
\begin{proof}[Proof of lemma]
If not, then there is a sequence $\vartheta_n\searrow \vartheta_0$ such that there there are $\tilde p_n\in G_{\vartheta_n}\backslash \Om$. Since $\mathbb{S}^2$ is compact, $\{\tilde p_n\}$ has a convergent subsequence, still denoted by $\tilde p_n$, to a point $\tilde p\in \SSS^2$. Since the immersion $i$ is continuous, $i(\tilde p_n)$ is also convergent to $i(\tilde p)$. Also notice that $i(\tilde p_n)\in i(G_{\vartheta_n})$, i.e., $\vartheta(i(\tilde p_n))\leq \vartheta_n$, so we have $\vartheta(i(\tilde p))\leq \vartheta_0$, and by definition, $\tilde p\in G_{\vartheta_0}$. On the other hand, since $\tilde p_n\notin \Om$, we have $\tilde p\notin \Om$ since the complement of $\Om$ is closed, a contradiction. The proof of the second statement is almost identical.
\end{proof}
\begin{remark}
\lab{remark-Th_*}
By assumption {\bf A}, $H_0=G_0$ contains at least two points. Then since $H_0=G_0$ is a set of finite distinct points on $\mathbb{S}^2$, it is disconnected. Using $\tilde \OO_p\cap H_0=\{p\}$ for each $p\in H_0$ and the lemma with $\vartheta_0=0$, it is easy to see that there exists a $\delta_0>0$ such that $G_{\vartheta}$ remains disconnected
 for all $\vartheta\leq\delta_0$. 
 \end{remark}

Now define $\vartheta^*:=\sup\{\vartheta\in [0,\pi]\colon G_{\kappa} \text{ is disconnected for all }\kappa\in [0,\vartheta]\}$. In view of Remark \ref{remark-Th_*}  we have  $\vartheta^*\geq \delta_0$. 

With this setup, we show the following statements, which lead to a contradiction with the assumption {\bf A}:
\begin{enumerate}
    \item[{\bf A1}] It is impossible to have $\vartheta^*=\pi$;
    \item[{\bf A2}] If $\vartheta^*<\pi$, then $G_{\vartheta^*}$ is connected;
    \item[{\bf A3}] If $\vartheta^*<\pi$, then $G_{\vartheta^*}$ is disconnected.
\end{enumerate}

Once  we  show   that  {\bf A1}-{\bf A3}  must all  hold  true, we obtain a contradiction with assumption {\bf A}. 
Therefore,  for different points $p_1$, $p_2$ on $\mathbb{S}^2$, $P_{\Si;\SSS^2}(i(p_1))$ and $P_{\Si;\SSS^2}(i(p_2))$ cannot be the same, 
i.e, $P_{\Si;\SSS^2}\circ i\colon \mathbb{S}^2\to \mathbb{S}^2$ is injective. Since $P_{\Si;\SSS^2}\circ i\colon \mathbb{S}^2\to \mathbb{S}^2$ is also a local diffeomorphism between $\mathbb{S}^2$, it is a global diffeomorphism from $\SSS^2$ onto its image. Now noticing that no proper subset of $\mathbb{S}^2$ has the same topology as $\mathbb{S}^2$,\footnote{Recall that $\mathbb{S}^2$ with a point removed is homeomorphic to $\mathbb{R}^2$, so any proper subset of $\mathbb{S}^2$ is topologically a subset of $\mathbb{R}^2$.} we see that it must also be surjective, and hence, $S=i(\mathbb{S}^2)$ can be globally written as a graph of a function $R$ on $\mathbb{S}^2$ stated in Proposition \ref{prop:immersed-embedded}.
\end{proof}

We now show that {\bf A} implies {\bf A1}-{\bf A3}.

{\bf Proof of  {\bf A}$\implies$A1.} 
Indeed,  if $\vartheta^*=\pi$, then  $G_\vartheta$  is disconnected for all $\vartheta<\pi$. We have argued above that the set $H_{\pi}=i^{-1}(\Phi_{\Sigma}(\{(r,p)\colon p=P_0'\}))$ consists of finite number of points, denoted by $\{p_1',p_2',\cdots,p_m'\}$, with neighborhoods $\tilde \OO_{p_j'}$ satisfying the local graph property. Now take a sequence $\vartheta_n\nearrow \pi$. By Lemma \ref{lem:Near-H_theta} above, 
    for sufficiently large $n$, we have $i^{-1}\big(\{\vartheta>\vartheta_n\}\big)\subset \cup_{j=1}^k \tilde \OO_{p_j'}$ and becomes a union of small disks on $\mathbb{S}^2$ with radius arbitraily small. Precisely, for any $\epsilon>0$, there exists $N>0$ such that for all $n>N$, $\OO_{p_j'}\cap i^{-1}(\{\vartheta>\vartheta_n\})$ lies inside an open disk on $\mathbb{S}^2$ with radius\footnote{Evaluated under the standard metric on $\mathbb{S}^2$.} less than $\epsilon$ for all $j=1,\cdots,k$, and
    \[G_{\vartheta_n}=\mathbb{S}^2\backslash \big(\cup_{j=1}^k (\tilde \OO_{p_j'}\cap i^{-1}(\{\vartheta>\vartheta_n\})\big),\] 
    is a set obtained by removing $k$ disjoint open disks with radius less than $\epsilon$ from $\mathbb{S}^2$.
  
    By assumption, $G_{\vartheta_n}$ is disconnected for all $n$. On the other hand, for any given positive integer $k$, there exists $\epsilon'>0$ such that it is impossible to make $\mathbb{S}^2$ disconnected by taking away $k$ geodesic disks with radius less than $\epsilon'$. This gives a contradiction and proves {\bf A1}.
    
    
{\bf Proof of  {\bf A}$\implies$A2.}  
  If $G_{\vartheta^*}$  were   disconnected, since it is a compact subset of $\mathbb{S}^2$, we can write $G_{\vartheta^*}=G_{\vartheta^*}^1\cup G_{\vartheta^*}^2$ with $G_{\vartheta^*}^1$, $G_{\vartheta^*}^2$ being two non-empty, disjoint compact subsets of $\mathbb{S}^2$.
  We can then find two disjoint open sets $\Om_1$, $\Om_2$ such that $G_{\vartheta^*}^1\subset \Om_1$, $G_{\vartheta^*}^2\subset \Om_2$.
  Let $\Om=\Om_1 \cup \Om_2$. By Lemma \ref{lem:Near-H_theta}, for some $\delta'>0$, $G_{\vartheta^*+\delta'}\subset \Om$,  and so is $G_{\vartheta}$ for all $ \vartheta\leq \vartheta^*+\delta'$. Also, for any $\vartheta^*\le \vartheta\leq \vartheta^*+\delta'$, $G_{\vartheta}$ has a point in each connected component of $\Om=\Om_1\cup \Om_2$, so it is itself  disconnected. By the definition of $\vartheta^*$, we also know that $G_{\vartheta}$ is disconnected for $0\leq \vartheta <\vartheta^*$, hence for $0\leq \vartheta\leq \vartheta^* +\delta'$. This contradicts  the maximality of $\vartheta^*$. 

    

{\bf Proof of  {\bf A}$\implies$A3.}
This is shown by the following proposition.
\begin{proposition}
\lab{Prop:A3}
    If $0<\vartheta^*<\pi$, and $G_{\vartheta^*}$ is connected, then there exists $\vartheta''<\vartheta^*$ such that $G_{\vartheta''}$ is also connected.
  
\end{proposition}  

\begin{proof}
    We construct the following cover of $G_{\vartheta^*}$. Since  $\vartheta^*$ is a given number in $(0,\pi)$, for  each $p\in G_{\vartheta^*}$, we define
    \begin{itemize}
        \item If $\frac{|\vartheta(p)-\vartheta^*|}{\vartheta^*}\leq \frac 12$, 
        (i.e, $\vartheta(p)$ is not so close to $0$), we pick a smaller neighborhood $\OO'_p\subset \tilde \OO_p$ satisfying the property that $P_{\Si;\SSS^2}\circ i(\OO'_p)$ is  a  disk   $D_p(\ep_p) $,  of radius $\ep_p$  sufficiently small,     in the coordinate chart  of  $(\theta_{P_0},\vphi_{P_0})$\footnote{Note that $P_{\Sigma;\mathbb{S}^2}\circ i$ is already a diffeomorphism onto the image when restricted to $\OO_p$, and $(\theta_{P_0},\vphi_{P_0})$ is a regular coordinate chart away from $P_0$, $P_0'$.}.
        This means that for each $\kappa\in (0,\pi)$, the set $\OO_p'\cap i^{-1}(\{\vartheta\leq \kappa\})$ is diffeomorphic to $D_p(\ep_p)\cap \{\theta_{P_0}\leq \kappa\}$ through the map $P_{\Si;\SSS^2}\circ i$, so it is topologically the same as an upper half disk.

        \item If $\frac{|\vartheta(p)-\vartheta^*|}{\vartheta^*}> \frac 12$, we let $\OO'_p\subset \tilde \OO_p$ be a connected open set that lies in $G_{\frac 34 \vartheta^*}$ and contains $p$.
    \end{itemize}
    Since $\{\OO'_p\}_{p\in \mathbb{S}^2}$ is a cover of $G_{\vartheta^*}$, we can extract a finite subcover $\{\OO'_{p_1},\cdots,\OO'_{p_l}\}$, or simply $\{\OO_1',\cdots,\OO_l'\}$. Among them, we can further extract those that intersect with $H_{\vartheta^*}$, and denote them by $\{\UU_1,\cdots,\UU_k\}$. This is a finite cover of $H_{\vartheta^*}$.


We need a lemma that allows us to modify curves.
\begin{lemma}
\lab{Lemma-paths}
    Suppose that there is a curve $\gamma\colon [0,1]\to  G_{\vartheta^*}$, with $\vartheta\circ i(\gamma(0)),\vartheta\circ i(\gamma(1))<\vartheta^*$. Then there exists another curve $\tilde\gamma\colon [0,1]\to G_{\vartheta^*}$ with $\tilde \gamma(0)=\gamma(0)$, $\tilde \gamma(1)=\gamma(1)$, and it satisfies $\sup_{\gamma} \vartheta<\vartheta^*$.
\end{lemma}

Assuming the lemma to be true, we continue as follows:

For each $j\in \{1,\cdots, l\}$, we pick a representative $\mathring p_j\in \OO'_j$ such that $\vartheta(\mathring p_j)<\vartheta^*$. We deduce  that   there is a $\vartheta'<\vartheta$ such that $\vartheta(\mathring p_j)<\vartheta'$ for all $j=1,\cdots,l$. By assumption, for $\mathring p_j$, $\mathring p_k$, there is a curve in $G_{\vartheta^*}$ connecting them. Modifying the curve as above, we see that we can connect $\mathring p_j$, $\mathring p_k$ by a curve in $G_{\vartheta_{j,k}}$ for some $\vartheta_{j,k}<\vartheta^*$. Do this for all pairs of $j,k=1,\cdots,l$, and define $\vartheta''=\max\{\frac 34 \vartheta^*,\vartheta',\vartheta_{j,k}\}<\vartheta^*$.  This ensures that all   deformed curves connecting   $\mathring p_j$ are in $G_{\vartheta''}$.

  We claim that $G_{\vartheta''}$  has to be path connected, hence   connected.  To see this, take two points in $G_{\vartheta''}$. By our way of choosing $\OO'_j$, we see that $\OO'_j\cap i^{-1}(\{\vartheta\leq \vartheta''\})$, if non-empty, is connected, so these two points can be connected to their representatives in $\OO'_j\cap i^{-1}(\{\vartheta\leq \vartheta''\})$ (using $\vartheta'\leq \vartheta''$). Then since we have shown that for any $j,k=1,\cdots,l$, the points $\mathring p_j$, $\mathring p_k$ can be connected in $G_{\vartheta_{j,k}}\subset G_{\vartheta''}$, we see that $G_{\vartheta''}$ is indeed path-connected with $\vartheta''<\vartheta$. This  ends the proof of Proposition \ref{Prop:A3}.
\end{proof}

\subsection{Proof of Lemma \ref{Lemma-paths}}
We first illustrate the idea by assuming $k=2$, so $H_{\vartheta^*}\subset \UU_1\cup \UU_2$. Without loss of generality, we assume $\gamma$ enters $\UU_1$ first, 
and define $p_1\in \pa\UU_1$ as the point of this first entry\footnote{Strictly speaking, this should be described on the time interval $[0,1]$. We use the image point on $G_{\vartheta^*}$ for simplicity.}. This implies that $p_1\notin H_{\vartheta^*}$, because otherwise it would be in $\UU_2$ first. In the case when $\gamma(0)\in \UU_1$, we define instead $p_1:=\gamma(0)$. We then define $q_1\in \pa \UU_1$ to be the point where $\gamma$ leaves $\UU_1$ for the final time\footnote{This means that it is not entering $\UU_1$ after $q_1$; we allow the possibility for it to intersect $\pa\UU_1$.} (similarly, if $\gamma(1)\in \UU_1$, we let $q_1:=\gamma(1)$ instead; in this case the process ends and there are no $p_2$ and $q_2$ defined below). While $p_1\notin H_{\vartheta^*}$, it is possible that $q_1\in H_{\vartheta^*}$. We now define $p_2$ as follows:
\begin{itemize}
    \item If $q_1\in \UU_2$ (which must be true if $q_1\in H_{\vartheta^*}$), we define $p_2=q_1$;
    \item If $q_1\notin \UU_2$ (in this case $q_1\notin H_{\vartheta^*}$), we define $p_2\in \pa \UU_2$ to be point where $\gamma$ enters $\UU_2$ for the first time after $q_1$ (if there is no such point, then end the process and there are no $p_2$ and $q_2$); by assumption, at $p_2$, $\gamma$ has already left $\UU_1$ for the final time, so $p_2\notin \UU_1$, and hence, $p_2\notin H_{\vartheta^*}$.
\end{itemize}
We then define $q_2$ as the point where $\gamma$ leaves $\UU_2$ for the final time (similarly, if $\gamma(1)\in \UU_2$, we define $q_2=\gamma(1)$). By the assumption that $\gamma$ has already left $\UU_1$ for the final time, we see that $q_2\notin H_{\vartheta^*}$. Also, $\gamma$ will never be in $\UU_1\cap \UU_2$, hence never intersect $H_{\vartheta^*}$, after $q_2$. 

Now by the local structure of $\UU_1$ and $\UU_2$, there exists a curve connecting $p_1$ and $q_1$ such that its only possible intersection with $H_{\vartheta^*}$ is $q_1$, and a curve connecting $p_2$ and $q_2$ that does not intersect $\UU_2$. Therefore, if we replace the portion of $\gamma$ between $p_j$ and $q_j$ ($j=1,2$) by such segments, we obtain a new curve which only possibly intersects $H_{\vartheta^*}$ at $q_1$. Since in this case $q_1\in \UU_2$, it is clear that one can further modify near $q_1$ so that the curve does not intersect $H_{\vartheta^*}$. Therefore, there exists a curve having the same endpoints as $\gamma$ but never hits $H_{\vartheta^*}$. This finishes the proof with $k=2$ in view of the compactness of the new curve.

For general $k$, we prove by induction. Assume that we have found $p_1,q_1,\cdots,p_j,q_j$ in time order on $\gamma$, and a subcollection $\{\UU_1,\cdots,\UU_{j}\}$ from $\{\UU_1,\cdots,\UU_k\}$ (with possibly a rearrangement of the indices), such that
\begin{itemize}
    \item We have $p_1\notin H_{\vartheta^*}$; it is possible that $q_{j'}\in H_{\vartheta^*}$ for some $j'$ among $j'=1,\cdots,j$, but for $j'\geq 2$ (when $j\geq 2$), if $p_{j'}\neq q_{j'-1}$, then $q_{j'-1}\notin H_{\vartheta^*}$ and $p_{j'}\notin H_{\vartheta^*}$;
    \item For all $j'\leq j$, the curve $\gamma$ never intersects $\UU_{j'}$ after $q_{j'}$; 
    \item For all $j'\leq j$, there exists a curve $\gamma_j$, having the same two endpoints as $\gamma$, only different from $\gamma_{j'-1}$ ($\gamma_0:=\gamma$) on the segments between $p_{j'}$ and $q_{j'}$. This new segment between $p_{j'}$ and $q_{j'}$ on $\gamma_{j'}$ is denoted by $\gamma_{j'}[\, p_{j'};q_{j'}]$;
    \item The intersection $\gamma_j[\, p_{j'};q_{j'}]\cap H_{\vartheta^*} \subset \{q_{j'},q_{j'-1}\}$ (when $j'=1$, only $\{q_1\}$) for all $j'=1,\cdots,j$.
\end{itemize}
Following the same way of defining $p_1$ and $q_1$ above, we see that this holds for $j=1$.
We now define $p_{j+1}$ in a similar way as we defined $p_2$ above:
\begin{itemize}
    \item If $q_j$ lies in some element in $\{\UU_{j+1},\cdots,\UU_k\}$ (which must be true if $q_j\in H_{\vartheta^*}$; choose any one if there are multiple choices), without loss of generality denoted by $\UU_{j+1}$, we define $p_{j+1}=q_j$;
    \item If $q_j$ is not in any of $\{\UU_{j+1},\cdots,\UU_k\}$, without loss of generality, we denote the next one $\gamma$ enters by $\UU_{j+1}$, and define $p_{j+1}\in \pa \UU_{j+1}$ to be point of entry; this only happens when $q_j\notin H_{\vartheta^*}$. Also, this means that $p_{j+1}$ is not in any among $\{\UU_1,\cdots,\UU_k\}$, so $p_{j+1}\notin H_{\vartheta^*}$.
\end{itemize}
Then, we define $q_{j+1}$ as the point where $\gamma$ (or equivalently, $\gamma_j$) leaves $\UU_{j+1}$ for the final time (again, if $\gamma(1)\in \UU_{j+1}$, define instead $p_{j+1}:=\gamma(1)$ and end the process).

Note that it is possible that $q_{j+1}\in H_{\vartheta^*}$. As above, using the local structure of $\UU_{j+1}$ we can find a curve segment connecting $p_{j+1}$ and $q_{j+1}$, denoted by $\gamma_{j+1}[\, p_{j+1};q_{j+1}]$ such that its only possible intersections with $H_{\vartheta^*}$ are $q_{j+1}$ and, in the case $p_{j+1}=q_j$, the point $q_j$. We then replace the segment of $\gamma_j$ between $p_{j+1}$ and $q_{j+1}$ by this curve segment, and denote the new curve by $\gamma_{j+1}$. One can verify that the induction assumptions now hold true with $j$ replaced by $j+1$. This concludes the induction. The procedure will stop at $j=k'$ for some $k'\leq k$, with $\gamma$ does not intersect with any of $\{\UU_1,\cdots,\UU_k\}$ after $q_{k'}$, or $q_{k'}=\gamma(1)$. The curve $\gamma_{k'}$ satisfies
\begin{itemize}
    \item It has the same two endpoints as $\gamma$, and is only different from $\gamma$ on the segments between $p_{j'}$ and $q_{j'}$ for $j'\leq k'$;
    \item The intersection $\gamma_{k'}\cap H_{\vartheta^*} \subset \{q_1,\cdots,q_{k'-1}\}$.
\end{itemize}
This shows that $\gamma(0)$ and $\gamma(1)$ can be connected by a curve $\gamma_{k'}$ that only possibly intersects with $H_{\vartheta^*}$ at $\{q_1,\cdots,q_{k'-1}\}$. Recall that if $q_{j'}\in H_{\vartheta^*}$, by our construction we have $q_{j'}\in \UU_{j'+1}$, so clearly one can further modify the curve near $q_{j'}$ so that it does not intersect $H_{\vartheta^*}$ nearby. Since we have at most $k'-1$ such points, we can do this for finite times and obtain a curve $\tilde\gamma_{k'}$ connecting $\gamma(0)$ and $\gamma(1)$ that never intersects $H_{\vartheta^*}$. By compactness we see that $\sup_{\tilde\gamma_{k'}} \vartheta<\vartheta^*$, which finishes the proof.




\section{Proof of the Immersion Criterion}
\label{Appendix:Null-hypersurfaces}
For the benefit of the reader  we provide below a  short  outline on the   theory of focal points of null hypersurfaces; for a more complete treatment see  \cite{AretakisNotes}.

In this part, it would be convenient to work with the foliation induced by the affine parameter of the null generators. Fix a smooth choice of $L$ on $S$ and consider the corresponding null cone $C(S)$. Define $s$ by $L(s)=1$, $s=0$ on $S$.
For a local coordinate $(\theta_p,\varphi_p)$ near $p$ on $S$, one can then extend them near $\gamma(p)$ by $L(\theta_p)=L(\vphi_p)=0$, which also defines $\pa_{\theta_p}$, $\pa_{\vphi_p}$ tangent to $\{s=\const\}$ near each point, until the linear span of them becomes degenerate. The vector fields $\pa_{\theta_p}$, $\pa_{\vphi_p}$ are normal Jacobi fields in the sense that $[L,J]=0$ for $J=\pa_{\theta_p}, \pa_{\vphi_p}$ ($L=\pa_s$), and they satisfy the Jacobi equation
($D_s:=D_L$)
\begin{equation*}
    D_s^2 J+R(J,L)L=0,\quad J|_S=\pa_{\theta_p},\pa_{\vphi_p},\quad D_s J|_{S}=D_{J} L.
\end{equation*}
The equation is linear when $L$ is given, and the solution will be smooth if initially smooth.

\begin{proposition}\label{Prop; radius-conjugacy}
Fix a point $p\in S$. Let $s^*$ be the first (minimal)\footnote{This is towards the future; the treatment of the case towards the past is the same.} value of $s'$ such that the span of $\pa_{\theta_p}$, $\pa_{\varphi_p}$ degenerates, 
then $\trch\to -\infty$ towards this point along $\gamma_{p}$.
\end{proposition}

\begin{proof}

We choose an orthonormal frame of $S$ near $p$, denoted by $\{E_1,E_2\}$, and extend it as the Fermi frame by the equation (denote $E_4=L$)
\begin{equation}\label{eq:Fermiframe}
    D_{E_4} E_a=-\zeta_a E_4
\end{equation}
where $\zeta_a=g(D_{E_a}E_4,E_3)$ (with $E_3$ being the null companion of $E_4$ orthogonal to $E_a$). We have
\begin{equation*}
E_4(E_i(s))=[E_4,E_i]s=(D_{E_4} E_i-D_{E_i} E_4)(s)=(-\zeta_i E_4-D_{E_i} E_4)s=-\chi(E_i,E_j)E_j(s)
\end{equation*}
so we see that if the time function $s$ satisfies $E_i(s)=0$ initially, then we always have $E_i(s)=0$. Also,
\begin{equation*}
    E_4\, g(E_a,E_b)=g(-\zeta_a E_4,E_b)+g(E_a,-\zeta_b E_4)=0,\quad  E_4\, g(E_a,E_4)=g(-\zeta_a E_4,E_4)=0.
\end{equation*}
Therefore, $E_a$ will always be an orthonormal frame orthogonal to $E_4$.

Since $X=\pa_{\theta_p},\pa_{\vphi_p}$ are tangent to $\{s=\const\}$, we write $X=X^a E_a$. The condition $[L,X]=0$ implies 
\begin{equation*}
    \frac{d}{ds}(X^a(s))E_a+X^a D_L E_a=X^a D_{E_a}L, \text{ hence }\frac{d}{ds}X^a(s)=\chi_{ab} X^b.
\end{equation*}
In particular, this determines the initial condition on $\frac{d}{ds}X^a(s)$.

The Jacobi equation now gives
\begin{equation*}
    \frac{d^2}{ds^2}X^a(s)=\a(E_a,E_b)X^b.
\end{equation*}
Note that $E_a$, $\a$ (defined under the frame $\{E_3,E_4,E_a\}$) are known quantities, so this is a linear equation of $X^a$. Therefore, $X^a$ will be smooth when the initial data and spacetime are smooth.

Consider the deformation matrix $M$, defined by $X^a(s)=M_b^a(s) X^b(0)$. It satisfies 
\begin{equation*}
    \frac{d}{ds}M_b^a(s)=\chi_{ac} M_b^c,
\end{equation*}
which implies
\begin{equation*}
    \frac{d}{ds}\big(\log\det M\big)=\trch.
\end{equation*}
By assumption, we see that $M$ becomes degenrate, i.e., $\det M\to 0^+$ along $\ga_p$ as $s\to s^*$. This implies that $\trch \to -\infty$ towards the point.
\end{proof}

Therefore, if $\trch$ is not diverging to $-\infty$ at any point, then $M$ will always be a linear isomorphism.
In this case, consider the map
\begin{equation*}
    i_\tau\colon S\mapsto \Sigma_\tau
\end{equation*}
where $\Sigma_\tau$ is a level hypersurface of a time function, and we define $i_\tau(p)$ to be the unique point of the intersection $\ga_p\cap \Si_\tau$. 
The tangent map of $i_\tau$ is given by
\begin{equation*}
    \pa_{\theta_p}|_{S}\mapsto \pa_{\theta_p}-\frac{\pa_{\theta_p}\tau}{L(\tau)}L,\quad \pa_{\vphi_p}|_S \mapsto \pa_{\vphi_p}-\frac{\pa_{\vphi_p} \tau}{L(\tau)}L
\end{equation*}
We know that $L(\tau)>0$ when $\tau$ is a time function; it is then clear that we obtain two linearly independent vectors, so the tangent map is injective.
Therefore, $i_\tau$ is an immersion. 
Since the map is also smooth in $\tau$ and $L$ is transversal to $\Si_\tau$, we also see that it gives an immersion from $[0,T]\times  S\to \MM$, $(\tau,p)\mapsto i_\tau(p)$.

{\bf Acknowledgements.} Both authors would like to thank IHES for their hospitality during visits. 
X.C. thanks his advisor, Hans Lindblad, for supporting his many visits to Columbia University 
through the Simons Collaboration Grant 638955. Klainerman  was supported  by the NSF grant 2201031.

\end{document}